\def\confversion{0}
\def\ifconf{\ifnum\confversion=1}
\def\ifnotconf{\ifnum\confversion=0}

\documentclass[11 pt]{article}
\def\showauthornotes{0}
\def\showkeys{0}
\def\showdraftbox{0}

\usepackage{xspace,xcolor,enumerate}
\usepackage{amsmath,amssymb}
\usepackage{amsthm}
\usepackage[toc,page]{appendix}
\usepackage{thmtools}
\usepackage{thm-restate}
\usepackage{color,graphicx}
\usepackage{boxedminipage}
\usepackage{makecell}
\usepackage{tabularx}

\ifnum\showkeys=1
\usepackage[color]{showkeys}
\fi

\definecolor{darkred}{rgb}{0.5,0,0}
\definecolor{darkgreen}{rgb}{0,0.35,0}
\definecolor{darkblue}{rgb}{0,0,0.55}

\usepackage[pdfstartview=FitH,pdfpagemode=UseNone,colorlinks,linkcolor=darkblue,filecolor=darkred,citecolor=darkgreen,urlcolor=darkred,pagebackref]{hyperref}

\usepackage[capitalise,nameinlink]{cleveref}
\usepackage[T1]{fontenc}
\usepackage{mathtools,dsfont,bbm}
\usepackage{mathpazo}
\ifnum\showauthornotes=1
\newcommand{\Authornote}[2]{{\sf\small\color{red}{[#1: #2]}}}
\newcommand{\Authorcomment}[2]{{\sf \small\color{gray}{[#1: #2]}}}
\newcommand{\Authorfnote}[2]{\footnote{\color{red}{#1: #2}}}
\else
\newcommand{\Authornote}[2]{}
\newcommand{\Authorcomment}[2]{}
\newcommand{\Authorfnote}[2]{}
\fi

\ifnum\showdraftbox=1
\newcommand{\draftbox}{\begin{center}
  \fbox{%
    \begin{minipage}{2in}%
      \begin{center}%
        \begin{Large}%
          \textsc{Working Draft}%
        \end{Large}\\
        Please do not distribute%
      \end{center}%
    \end{minipage}%
  }%
\end{center}
\vspace{0.2cm}}
\else
\newcommand{\draftbox}{}
\fi

\newtheorem{theorem}{Theorem}[section]

\newtheorem{definition}[theorem]{Definition}
\newtheorem{notation}[theorem]{Notation}
\newtheorem{lemma}[theorem]{Lemma}
\newtheorem{remark}[theorem]{Remark}

\newtheorem{corollary}[theorem]{Corollary}
\newtheorem{claim}[theorem]{Claim}
\newtheorem{fact}[theorem]{Fact}

\newtheorem{algo}[theorem]{Algorithm}

\newenvironment{algorithm}[3]
        {\noindent\begin{boxedminipage}{\textwidth}\begin{algo}[#1]\ \par
        {\begin{tabular}{r l}
        \textbf{Input} & #2\\
        \textbf{Output} & #3
        \end{tabular}\par\enskip}}
        {\end{algo}\end{boxedminipage}}

\def\FullBox{\hbox{\vrule width 6pt height 6pt depth 0pt}}

\def\qed{\ifmmode\qquad\FullBox\else{\unskip\nobreak\hfil
\penalty50\hskip1em\null\nobreak\hfil\FullBox
\parfillskip=0pt\finalhyphendemerits=0\endgraf}\fi}

\def\qedsketch{\ifmmode\Box\else{\unskip\nobreak\hfil
\penalty50\hskip1em\null\nobreak\hfil$\Box$
\parfillskip=0pt\finalhyphendemerits=0\endgraf}\fi}

\def\to{\rightarrow}
\def\eps{\varepsilon}
\def\epsilon{\varepsilon}

\def\eps{\epsilon}

\def\phi{\varphi}

\def\implies{\Rightarrow}

\newcommand{\defeq}{:=}

\newcommand{\ie}{i.e.,\xspace}
\newcommand{\eg}{e.g.,\xspace}
\newcommand{\etal}{et al.\xspace}

\newcommand{\mper}{\,.}
\newcommand{\mcom}{\,,}

\newcommand{\R}{{\mathbb R}}
\newcommand{\E}{{\mathbb E}}

\newcommand{\N}{{\mathbb{N}}}

\newcommand{\F}{{\mathbb F}}

\newcommand{\pmone}{\{-1,1\}\xspace}

\newcommand{\indicator}[1]{\mathds{1}_{\{#1\}}}

\usepackage{nicefrac}
\newcommand{\abs}[1]{\ensuremath{\left\lvert #1 \right\rvert}}

\newcommand{\norm}[1]{\ensuremath{\left\lVert #1 \right\rVert}}

\newcommand{\ip}[2] {\ensuremath{\left\langle #1 , #2 \right\rangle}}

\newcommand{\one}{{\mathbf{1}}}

\newcommand{\Esymb}{\mathbb{E}}
\newcommand{\Psymb}{\mathbb{P}}

\DeclareMathOperator*{\ExpOp}{\Esymb}
\DeclareMathOperator*{\ProbOp}{\Psymb}

\makeatletter

\def\Ex#1{%
    \ProbabilityRender{\Esymb}{#1}%
}

\def\ProbabilityRender#1#2{%fancy probability command
  \@ifnextchar\bgroup%
  {\renderwithdist{#1}{#2}}
   {\singlervrender{#1}{#2}}
}
\def\singlervrender#1#2{%
   \ensuremath{\mathchoice
       {{#1}\left[ #2 \right]}
       {{#1}[ #2 ]}
       {{#1}[ #2 ]}
       {{#1}[ #2 ]}
   }
}
\def\renderwithdist#1#2#3{%
   \@ifnextchar\bgroup
   {\superfancyrender{#1}{#2}{#3}}
   {\ensuremath{\mathchoice
      {\underset{#2}{#1}\left[ #3 \right]}
      {{#1}_{#2}[ #3 ]}
      {{#1}_{#2}[ #3 ]}
      {{#1}_{#2}[ #3 ]}
     }
   }
}
\def\superfancyrender#1#2#3#4#5{
   \ensuremath{\mathchoice
      {\underset{#1}{{#1}}\left#4 #3 \right#5}
      {{#1}_{#2}#4 #3 #5}
      {{#1}_{#2}#4 #3 #5}
      {{#1}_{#2}#4 #3 #5}
   }
}
\makeatother

\newfont{\inhead}{eufm10 scaled\magstep1}

\newcommand{\poly}{{\mathrm{poly}}}
\newcommand{\polylog}{{\mathrm{polylog}}}

\DeclareMathOperator{\cov}{\operatorname {Cov}}

\DeclareMathOperator*{\argmin}{\arg\!\min}
\DeclareMathOperator*{\argmax}{\arg\!\max}

\renewcommand{\iff}{\ensuremath{\Leftrightarrow}}

\newcommand{\ceil}[1]{\ensuremath{\left\lceil #1 \right\rceil}}
\newcommand{\floor}[1]{\ensuremath{\left\lfloor #1 \right\rfloor}}

\newcommand{\problemmacro}[1]{\textsf{#1}}

\newcommand{\inparen}[1]{\left(#1\right)}             %\inparen{x+y}  is (x+y)
\newcommand{\inbraces}[1]{\left\{#1\right\}}           %\inbrace{x+y}  is {x+y}
\DeclareSymbolFont{extraup}{U}{zavm}{m}{n}
\DeclareMathSymbol{\varheart}{\mathalpha}{extraup}{86}
\DeclareMathSymbol{\vardiamond}{\mathalpha}{extraup}{87}

\def\matr#1{\mathsf{#1}}

\def\rand#1{\mathbf{#1}}
\def\rv#1{\rand #1}

\def\One{\mathbb 1}
\def\one{\mathbf 1}

\def\Cc{\mathcal C}

\def\assn{\sigma}
\def\tree{\mathcal T}
\def\Aye{\matr A}

\def\Ess{\matr S}

\def\Emm{\matr M}
\def\Jay{\matr J}
\def\Gee{\matr G}
\def\ess{\mathfrak s}

\DeclarePairedDelimiter\set{\lbrace}{\rbrace}
\DeclarePairedDelimiter\parens{\lparen}{\rparen}

\DeclareMathOperator{\bias}{bias}
\DeclareMathOperator{\lift}{dsum}
\DeclareMathOperator{\dsum}{dsum}

\newcommand{\tswap}[3]{\Ess_{#1,#2,#3}}

\newcommand{\fswap}[3]{\mathfrak{S}_{#1,#2,#3}}

\newcommand{\sswap}[1]{\Ess_{#1,#1}^{\bigtriangleup}}

\newcommand{\lict}[1]{$(1/2-#1)$-\text{close}}

\usepackage{tikz}
\newcommand*\circled[1]{\tikz[baseline=(char.base)]{
            \node[shape=circle,draw,inner sep=1pt] (char) {#1};}}
\DeclareMathOperator{\zigzag}{\circled{{\rm z}}}
\usepackage{tikz}
\usetikzlibrary{shapes.symbols}
\usepackage{float}
\usepackage{graphicx}

\usepackage[top=1in, bottom=1in, left=1.25in, right=1.25in]{geometry}
\usepackage{amsmath}
\usepackage{amssymb}
\usepackage{amsthm}
\usepackage{ytableau}
\usepackage{mathtools}

\begin{document}

\title{Unique Decoding of Explicit $\epsilon$-balanced Codes Near \\
       the Gilbert--Varshamov Bound}

\author{
Fernando Granha Jeronimo\thanks{{\tt University of Chicago}. {\tt granha@uchicago.edu}. Supported in part by NSF grant CCF-1816372. } 
\and
Dylan Quintana\thanks{{\tt University of Chicago}. {\tt dquintana@uchicago.edu}}
\and
Shashank Srivastava\thanks{{\tt TTIC}. {\tt shashanks@ttic.edu}. Supported in part by NSF grant CCF-1816372.}
\and
Madhur Tulsiani\thanks{{\tt TTIC}. {\tt madhurt@ttic.edu}. Supported by NSF grant CCF-1816372.} 
}

%\date{\today}

\setcounter{page}{0}

\date{}

\maketitle
\draftbox
\thispagestyle{empty}

The Gilbert--Varshamov bound (non-constructively) establishes the
existence of binary codes of distance $1/2 -\epsilon$ and rate
$\Omega(\epsilon^2)$ (where an upper bound of $O(\epsilon^2\log(1/\epsilon))$
is known). Ta-Shma [STOC 2017] gave an explicit construction
of $\eps$-balanced binary codes, where any two distinct codewords are at a  
distance between $1/2 -\eps/2$ and $1/2+\eps/2$, achieving a near optimal
rate of $\Omega(\epsilon^{2+\beta})$, where $\beta \to 0$ as $\epsilon \to 0$. 

We develop unique and list decoding algorithms for (a slight modification of) 
the family of codes constructed by Ta-Shma, in the adversarial error model. 
We prove the following results for $\epsilon$-balanced codes with block 
length $N$ and rate $\Omega(\epsilon^{2+\beta})$ in this family:
\begin{itemize}
\item For all $\epsilon, \beta > 0$ there are explicit codes 
which can be uniquely decoded up to an error of half the minimum distance 
in time $N^{O_{\epsilon, \beta}(1)}$.
\item For any fixed constant $\beta$ independent of $\eps$, there 
is an explicit construction of codes which can be uniquely decoded up to an error of half the 
minimum distance in time $(\log(1/\epsilon))^{O(1)} \cdot N^{O_\beta(1)}$.
\item For any $\epsilon > 0$,  there are explicit $\epsilon$-balanced codes with 
rate $\Omega(\epsilon^{2+\beta})$ which can be list decoded up to error 
$1/2 - \epsilon'$ in time $N^{O_{\epsilon,\epsilon',\beta}(1)}$, where $\epsilon', \beta \to 0$ as $\epsilon \to 0$.
\end{itemize}

The starting point of our algorithms is the framework for list
decoding direct-sum codes develop in Alev \etal [SODA 2020], which
uses the Sum-of-Squares SDP hierarchy. The rates obtained there were
quasipolynomial in $\epsilon$. Here, we show how to overcome the far
from optimal rates of this framework obtaining \emph{unique decoding}
algorithms for explicit binary codes of near optimal rate. These codes
are based on simple modifications of Ta-Shma's construction.

\newpage

\ifnotconf
\pagenumbering{roman}
\tableofcontents
\clearpage
\fi

\pagenumbering{arabic}
\setcounter{page}{1}

\section{Introduction}\label{sec:intro}

Binary error correcting codes have pervasive
applications~\cite{G10:ICM,GRS:coding:notes} and yet we are far from
understanding some of their basic properties~\cite{G09Survey}. For
instance, until very recently no explicit binary code achieving
distance $1/2-\epsilon/2$ with rate near $\Omega(\epsilon^2)$ was known,
even though the existence of such codes was (non-constructively)
established long ago~\cite{G52,V57} in what is now referred as the
Gilbert--Varshamov (GV) bound. On the impossibility side, a rate upper
bound of $O(\epsilon^2\log(1/\epsilon))$ is known for binary codes of
distance $1/2-\epsilon/2$ (e.g.,~\cite{Delsarte75,MRRW77,NavonS09}).  

In a breakthrough result~\cite{Ta-Shma17}, Ta-Shma gave an explicit
construction of binary codes achieving nearly optimal distance versus
rate trade-off, namely, binary codes of distance $1/2-\epsilon/2$ with
rate $\Omega(\epsilon^{2+\beta})$ where $\beta$ vanishes as $\epsilon$
vanishes~\footnote{In fact, Ta-Shma obtained
$\beta=\beta(\epsilon)=\Theta(((\log{\log{1/\epsilon}})/\log{1/\epsilon})^{1/3})$
and thus $\lim_{\epsilon \to 0} \beta(\epsilon)=0$.}. 
Actually, Ta-Shma
obtained $\epsilon$-balanced binary linear codes, that is, linear binary
codes with the additional property that non-zero codewords have
Hamming weight bounded not only below by $1/2-\epsilon/2$ but also above
by $1/2+\epsilon/2$, and this is a fundamental property in the study of
pseudo-randomness~\cite{NN90,AGHP92}. 

While the codes constructed by Ta-Shma are explicit, they were not known to
admit efficient decoding algorithms, while such results are known for 
codes with smaller rates.
In particular, an explicit binary code due to Guruswami and
Rudra~\cite{GuruswamiR06} is known to be even list decodable at an
error radius $1/2-\epsilon$ with rate $\Omega(\epsilon^3)$.
We consider the following question:
\begin{center}
  \emph{Do explicit binary codes near the GV bound admit 
	an efficient decoding algorithm?}
\end{center}

Here, we answer this question in the affirmative by providing an
efficient~\footnote{By ``efficient'', we mean polynomial
time. Given the fundamental nature of the problem of decoding nearly
optimal binary codes, it is an interesting open problem to make these
techniques viable in practice.} unique decoding algorithm for
(essentially) Ta-Shma's code construction, which we refer as Ta-Shma
codes. More precisely, by building on Ta-Shma's construction and using
our unique decoding algorithm we have the following result.

\begin{restatable}[Unique Decoding]{theorem}{TheoMainUniqueDec}\label{theo:main}
  For every $\epsilon > 0$ sufficiently small, there are explicit binary linear Ta-Shma codes
  $\Cc_{N,\epsilon,\beta} \subseteq \mathbb{F}_2^N$ for infinitely many values $N \in \mathbb{N}$ with
  \begin{enumerate}[(i)]
   \item distance at least $1/2 - \epsilon/2$ (actually $\epsilon$-balanced),
   \item rate $\Omega(\epsilon^{2 + \beta})$ where $\beta = O(1/(\log_2(1/\epsilon))^{1/6})$, and
   \item a unique decoding algorithm with running time $N^{O_{\epsilon,\beta}(1)}$.
  \end{enumerate}
  Furthermore, if instead we take  $\beta > 0$ to be an arbitrary constant, the running time
  becomes $(\log(1/\epsilon))^{O(1)} \cdot N^{O_{\beta}(1)}$ (fixed polynomial time).
\end{restatable}

We can also perform ``gentle'' list decoding in the following sense
(note that this partially implies~\cref{theo:main}).
\begin{theorem}[Gentle List Decoding]\label{theo:gentle_list_decoding}
  For every $\epsilon > 0$ sufficiently small, there are explicit binary linear Ta-Shma codes $\Cc_{N,\epsilon,\beta} \subseteq \mathbb{F}_2^N$
  for infinitely many values $N \in \mathbb{N}$ with
  \begin{enumerate}[(i)]
   \item distance at least $1/2 - \epsilon/2$ (actually $\epsilon$-balanced),
   \item rate $\Omega(\epsilon^{2 + \beta})$ where $\beta = O(1/(\log_2(1/\epsilon))^{1/6})$, and
   \item a list decoding algorithm that decodes within radius $1/2 - 2^{-\Theta((\log_2(1/\epsilon))^{1/6})}$
                in time $N^{O_{\epsilon,\beta}(1)}$.
  \end{enumerate}
\end{theorem}

We observe that the exponent in the running time
$N^{O_{\epsilon,\beta}(1)}$ appearing in~\cref{theo:main}
and~\cref{theo:gentle_list_decoding} depends on $\epsilon$. This
dependence is no worse than $O(\log\log(1/\epsilon))$, and if $\beta > 0$ is
taken to be an arbitrarily constant (independent of $\eps$), 
the running time becomes $(\log(1/\epsilon))^{O(1)} \cdot N^{O_{\beta}(1)}$. 
Avoiding this dependence in the exponent when
$\beta=\beta(\epsilon)$ is an interesting open
problem. Furthermore, obtaining a list decoding radius of
$1/2-\epsilon/2$ in~\cref{theo:gentle_list_decoding} with the same rate
(or even $\Omega(\epsilon^2)$) is another very interesting open problem
and related to a central open question in the adversarial error
regime~\cite{G09Survey}.

\medskip
\noindent \textbf{Direct sum codes.} \enspace
Our work can be viewed within the broader context of developing
algorithms for the decoding of direct sum codes. Given a (say linear)
code $\Cc \subseteq \F_2^n$ and a collection of tuples $W \subseteq [n]^t$,
the code $\dsum_{W}(\Cc)$ with block length $\abs{W}$ is defined as
	$$\dsum_{W}(\Cc) = \left\{(z_{w_1} + z_{w_2} + \cdots + z_{w_t})_{w \in W} \mid z \in \Cc \right\}.$$
The direct sum operation has been used for several applications in coding and complexity 
theory~\cite{ABNNR92, IW97, GI01, ImpagliazzoKW09, DinurS14, DDGEKS15, Chan16, 
DinurK17, A02:icm}.
It is easy to see that if $\Cc$ is $\eps_0$-balanced for a constant $\eps_0$, then for
any $\eps > 0$, choosing $W$ to be a random collection of tuples of size $O(n/\eps^2)$
results in $\dsum_W(\Cc)$ being an $\eps$-balanced code. 
The challenge in trying to construct good codes using this approach
is to find explicit constructions of (sparse) collections $W$ which are 
``pseudorandom'' enough to yield a similar distance amplification as above. 
On the other hand, the challenge in decoding such codes is to identify 
notions of ``structure'' in such collections $W$, which can be 
exploited by decoding algorithms.

In Ta-Shma's construction~\cite{Ta-Shma17}, such a pseudorandom collection $W$ was
constructed by considering an expanding graph $G$ over the vertex set $[n]$, and 
generating $t$-tuples using sufficiently long walks of length $t-1$ over the so-called $s$-wide replacement
product of $G$ with another (small) expanding graph $H$.
Roughly speaking, this graph product is a generalization of the celebrated zig-zag
product~\cite{RVW00} but with $s$ different steps of the zig-zag product 
instead of a single one.
Ta-Shma's construction can also be viewed as a clever way of selecting a 
\emph{sub-collection} of all walks in $G$, which refines an earlier construction
suggested by Rozenman and Wigderson \cite{RW08} (and also analyzed by Ta-Shma) 
using \emph{all} walks of length $t-1$.

\medskip
\noindent \textbf{Identifying structures to facilitate decoding.} \enspace
For the closely related direct product construction (where the entry
corresponding to $w \in W$ is the entire $t$-tuple $(z_{w_1}, \ldots,
z_{w_t})$) which amplifies distance but increases the alphabet size,
it was proved by Alon \etal \cite{ABNNR92} that the resulting code
admits a unique decoding algorithm if the incidence graph
corresponding to the collection $W$ is a good sampler. Very recently,
it was proved by Dinur \etal~\cite{DinurHKNT19} that such a direct
product construction admits list decoding if the incidence graph is a
``double sampler''. The results of~\cite{DinurHKNT19} also apply to
direct sum, but the use of double samplers pushes the rate away from
near optimality.

For the case of direct sum codes, the decoding task
can be phrased as a maximum $t$-XOR problem with the additional constraint
that the solution must lie in $\Cc$. More precisely, given
$\tilde{y} \in \mathbb{F}_2^{W}$ within the unique decoding radius of
$\dsum_W(\Cc)$, we consider the following optimization problem
$$
\argmin_{z \in \Cc}~ \Delta(\tilde{y},\dsum_W(z)),
$$
where $\Delta(\cdot,\cdot)$ is the (normalized) Hamming distance.
While maximum $t$-XOR is in general hard to solve to even any non-trivial
degree of approximation~\cite{H97:stoc}, previous work by the authors~\cite{AJQST19}
identified a structural condition on $W$ called ``splittability'' under which
the above constraint satisfaction problem can be solved (approximately) resulting in
efficient unique and list decoding algorithms. 
However, by itself the splittability condition is too crude to be applicable
to codes such as the ones in Ta-Shma's construction.
The requirements it places on the expansion of $G$ are too strong 
and the framework in \cite{AJQST19} is only able 
to obtain algorithms for direct sum codes with rate 
$2^{-(\log(1/\eps))^{2}} \ll \eps^{2+\beta}$.

The conceptual contribution of this work can be viewed as identifying
a different recursive structure in direct sums generated by expander walks,
which allows us to view the construction as giving a sequence of codes
$\Cc_0,\Cc_1, \ldots, \Cc_{\ell}$. Here, $\Cc_0$ is the starting code $\Cc$
and $\Cc_{\ell}$ is the final desired code, and each element in the sequence
can be viewed as being obtained via a direct sum operation 
on the preceding code. Instead of considering a ``one-shot'' decoding task of
finding an element of $\Cc_0$, this facilitates an iterative approach where at
each step we reduce the task of decoding the code $\Cc_i$ to decoding for 
$\Cc_{i-1}$, using the above framework from \cite{AJQST19}. 
Such an iterative approach with a sequence of codes was also used (in a 
very different setting) in a work of Guruswami and Indyk~\cite{GI03} constructing
codes over a large alphabet which are list decodable in linear time via spectral algorithms.

Another simple and well-known (see \eg ~\cite{GuruswamiI04}) observation, 
which is very helpful in our setting, is the use of list decoding 
algorithms for unique decoding. 
For a code with distance $1/2 - \eps/2$, unique decoding can be obtained by 
list decoding at a much smaller error radius of (say) $1/2 - 1/8$.
This permits a much more efficient application
of the framework from~\cite{AJQST19}, with a milder 
dependence on the expansion of the graphs $G$ and $H$ 
in Ta-Shma's construction, resulting in higher rates.
We give a more detailed overview of our approach in~\cref{sec:strategy}.

\medskip
\noindent \textbf{Known results for random ensembles.} \enspace
While the focus in this work is on explicit constructions, there 
are several known (non-explicit) constructions 
of random ensembles of binary codes near or achieving the
Gilbert--Varshamov bound (e.g.,~\cref{table:gv_bound_cmp}). 
Although it is usually straightforward to ensure the desired rate in such constructions, 
the distance only holds with high probability. Given a sample code from such ensembles,
certifying the minimum distance is usually not known to be polynomial
time in the block length. 
Derandomizing such constructions is also a possible avenue for obtaining optimal codes, 
although such results remain elusive to this date (to the best of our knowledge).

One of the simplest constructions is that of random binary linear
codes in which the generator matrix is sampled uniformly. This random
ensemble achieves the GV bound with high probability, but its decoding
is believed to be computationally hard~\cite{MMT11}.

Much progress has been made on binary codes by using results for larger
alphabet codes~\cite{G09Survey}. Codes over non-binary alphabets
with optimal (or nearly optimal) parameters are
available~\cite{vanLint99,Stichtenoth08,GuruswamiR06} and thanks to
this availability a popular approach to constructing binary codes has
been to concatenate such large alphabet codes with binary
ones. Thommesen~\cite{T83} showed that by concatenating Reed--Solomon
(RS) codes with random binary codes (one random binary code for each
position of the outer RS code) it is possible to achieve the GV bound.
Note that Thommesen codes arise from a more structured ensemble than
random binary linear codes. This additional structure enabled
Guruswami and Indyk~\cite{GuruswamiI04} to obtain efficient decoding
algorithms for the non-explicit Thommesen codes (whose minimum distance is not
known to admit efficient certification). This kind of concatenation
starting from a large alphabet code and using random binary codes,
which we refer as Thommesen-like, has been an important technique in
tackling binary code constructions with a variety of properties near
or at the GV bound. An important drawback in several such Thommesen-like
code constructions is that they end up being non-explicit (unless
efficient derandomization or brute-force is viable).

Using a Thommesen-like construction, Gopi \etal~\cite{GKORS16} showed
non-explicit constructions of locally testable and locally correctable
binary codes approaching the GV bound. More recently, again with a
Thommesen-like construction, Hemenway \etal \cite{HRW17} obtained
non-explicit near linear time unique decodable codes at the GV bound
improving the running time of Guruswami and Indyk~\cite{GuruswamiI04} (and
also the decoding rates).
We summarize the results discussed so far
in~\cref{table:gv_bound_cmp}.
%%
%% Adversarial setting (Hamming)
%%
\begin{table}[h!]
\centering
\scalebox{0.8}{
\begin{tabularx}{1.25\textwidth}{|l|>{\raggedright\arraybackslash}X|l|l|l|>{\raggedright\arraybackslash}X|l|}
  \hline
  \multicolumn{7}{|c|}{Binary Code Results near the Gilbert--Varshamov bound} \\
  \hline
   \textbf{Who?}  & \textbf{Construction} & \textbf{GV} & \textbf{Explicit} & \textbf{Concatenated} & \textbf{Decoding}  & \textbf{Local} \\
  \hline
  \cite{G52,V57} & existential & yes & no & no & no & n/a\\
  \hline
  \cite{T83}  & Reed--Solomon + random binary & yes & no & yes & no & n/a\\
  \hline
  \cite{GuruswamiI04} &  Thommesen~\cite{T83} & yes & no & yes & unique decoding & n/a\\
  \hline
  \cite{GKORS16} & Thommesen-like & yes & no & yes & unique decoding & LTC/LCC\\
  \hline
  \cite{HRW17} & Thommesen-like & yes & no & yes & near linear time unique decoding & n/a\\  
  \hline
  \cite{Ta-Shma17} & Expander-based & $\Omega(\epsilon^{2+\beta})$ & yes & no & no & n/a\\
  \hline
  \hline
  this paper  & Ta-Shma~\cite{Ta-Shma17} & $\Omega(\epsilon^{2+\beta})$ & yes & no & gentle list decoding & n/a\\
  \hline
\end{tabularx}}
\caption{GV bound related results for binary codes.}\label{table:gv_bound_cmp}
\end{table}

There are also non-explicit constructions known to
achieve list decoding capacity~\cite{GuruswamiR08,MosheiffRRSW19}
(being concatenated or LDPC/Gallager~\cite{Gallager62} is not an
obstruction to achieve capacity).  
Contrary to the other results in this subsection, Guruswami and
Rudra~\cite{Guruswami05,GuruswamiR06,G09Survey}, also using a
Thommesen-like construction, obtained explicit codes that are
efficiently list decodable from radius $1/2 - \epsilon$ with rate
$\Omega(\epsilon^3)$.
This was done by concatenating the so-called folded
Reed--Solomon codes with a derandomization of a binary ensemble of
random codes.
%

%%
%% Probabilistic setting (Shannon)
%%

\medskip
\noindent \textbf{Results for non-adversarial error models.} \enspace
All the results mentioned above are for the adversarial error model of
Hamming \cite{H50, G10:ICM}. In the setting of random corruptions (Shannon
model), the situation seems to be better understood thanks to the
seminal result on explicit polar codes of
Arikan~\cite{Arikan08}. 
More recently, in another breakthrough
Guruswami \etal~\cite{GuruswamiRY19} showed that polar codes can
achieve almost linear time decoding with near optimal convergence to
capacity for the binary symmetric channel. This result gives an explicit code construction
achieving parameter trade-offs similar to Shannon's randomized
construction~\cite{S48} while also admitting very efficient encoding
and decoding.
%%
%% Hybrid (channel with memory)
%%
Explicit capacity-achieving constructions are also known for bounded memory 
channels~\cite{ShaltielKS19} which restrict the power of the adversary 
and thus interpolate between the Shannon and Hamming models.

\section{Preliminaries and Notation}\label{sec:prelim}

\subsection{Codes}

We briefly recall some standard code terminology.
Given $z,z' \in \F_2^n$, recall that the relative
Hamming distance between $z$ and $z'$
is \text{$\Delta(z,z') \coloneqq \abs{\set{i \mid z_i\ne
z_i'}}/n$}. A binary code is any subset $\Cc \subseteq \F_2^n$.
The distance of $\Cc$ is defined as $\Delta(\Cc)
\coloneqq \min_{z\ne z'} \Delta(z,z')$ where $z,z' \in \Cc$. We say
that $\Cc$ is a linear code if $\Cc$ is a linear subspace of
$\mathbb{F}_2^n$. The rate of $\Cc$ is $\log_2(\abs{\Cc})/n$.

Instead of discussing the distance of a binary code, it will often be
more natural to phrase results in terms of its bias.

\begin{definition}[Bias]
  The \emph{bias} of a word $z \in \F_2^n$ is defined as $\bias(z) \coloneqq \abs{\E_{i \in [n]} (-1)^{z_i}}$.
  The bias of a code $\Cc$ is the maximum bias of any non-zero codeword in $\Cc$.
\end{definition}

\begin{definition}[$\epsilon$-balanced Code]
  A binary code $\Cc$ is \emph{$\epsilon$-balanced} if $\bias(z+z') \le \epsilon$ for every pair of distinct $z,z' \in \Cc$.
\end{definition}

\begin{remark}
 For linear binary code $\Cc$, the condition $\bias(\Cc) \leq \epsilon$ is equivalent to $\Cc$ being an $\epsilon$-balanced code.
\end{remark}

\subsection{Direct Sum Lifts}

Starting from a code $\Cc \subseteq \F_2^n$, we amplify its distance by considering the \textit{direct sum lifting} operation based on a collection $W(k) \subseteq [n]^k$.
The direct sum lifting maps each codeword of $\Cc$ to a new word in $\F_2^{|W(k)|}$ by taking the $k$-XOR of its entries on each element of $W(k)$.

\begin{definition}[Direct Sum Lifting]
  Let $W(k) \subseteq [n]^k$.
  For $z \in \F_2^n$, we define the \emph{direct sum lifting} as $\dsum_{W(k)}(z) = y$ such that $y_{\ess} = \sum_{i \in \ess} z_i$ for all $\ess \in W(k)$.
  The direct sum lifting of a code $\Cc \subseteq \F_2^n$ is
  $$
  \dsum_{W(k)}(\Cc) = \{\dsum_{W(k)}(z) \mid z \in \mathcal C\}.
  $$
  We will omit $W(k)$ from this notation when it is clear from context.
\end{definition}

\begin{remark}
  We will be concerned with collections $W(k) \subseteq [n]^k$ arising
  from length-$(k-1)$ walks on expanding structures (mostly on the
  $s$-wide replacement product of two expander graphs).
\end{remark}

We will be interested in cases where the direct sum lifting reduces
the bias of the base code; in~\cite{Ta-Shma17}, structures with such a
property are called \emph{parity samplers}, as they emulate the reduction in
bias that occurs by taking the parity of random samples.

\begin{definition}[Parity Sampler]
  A collection $W(k) \subseteq [n]^k$ is called an \emph{$(\epsilon_0, \epsilon)$-parity sampler} if for all $z \in \F_2^{n}$ with $\bias(z) \leq \epsilon_0$, we have $\bias(\dsum_{W(k)}(z)) \leq \epsilon$.
\end{definition}

\subsection{Linear Algebra Conventions}
All vectors considered in this paper are taken to be column vectors, and are
multiplied on the left with any matrices or operators acting on them. 
Consequently, given an indexed sequence of operators $\Gee_{k_1}, \ldots, \Gee_{k_2}$
(with $k_1 \leq k_2$) corresponding to steps  $k_1$ through $k_2$ of a walk, we 
expand the product $\prod_{i=k_1}^{k_2} G_i$ as 
\[
\prod_{i=k_1}^{k_2} G_i ~\defeq~ G_{k_2} \cdots G_{k_1} \mper
\]
Unless otherwise stated, all inner products for vectors in coordinate spaces are
taken to be with respect to the (uniform) probability measure on the coordinates.
Similarly, all inner products for functions are taken to be with respect to the 
uniform measure on the inputs. 
All operators considered in this paper are normalized to have singular values 
at most 1.

\section{Proof Overview}\label{sec:strategy}

The starting point for our work is the framework developed
in~\cite{AJQST19} for decoding direct sum codes, obtained by starting
from a code $\Cc \subseteq \F_2^n$ and considering all parities
corresponding to a set of $t$-tuples $W(t) \subseteq [n]^t$.
Ta-Shma's near optimal $\eps$-balanced codes are
constructed by starting from a code with constant rate and constant
distance and considering such a direct sum lifting.  The set of
tuples $W(t)$ in his construction corresponds to a set of walks of
length $t-1$ on the $s$-wide replacement product of an expanding graph
$G$ with vertex set $[n]$ and a smaller expanding graph $H$. The
$s$-wide replacement product can be thought of here as a way of
constructing a much smaller pseudorandom subset of the set of all
walks of length $t-1$ on $G$, which yields a similar distance
amplification for the lifted code.

\paragraph{The simplified construction with expander walks.}
While we analyze Ta-Shma's construction later in the paper, it is
instructive to first consider a $W(t)$ simply consisting of all walks
of length $t-1$ on an expander.  This construction, based on a
suggestion of Rozenman and Wigderson~\cite{RW08}, was also analyzed by
Ta-Shma~\cite{Ta-Shma17} and can be used to obtain $\epsilon$-balanced
codes with rate $\Omega(\eps^{4+o(1)})$. It helps to illustrate many
of the conceptual ideas involved in our proof, while avoiding some
technical issues.

Let $G$ be a $d$-regular expanding graph with vertex set $[n]$ and the
(normalized) second singular value of the adjacency operator $\Aye_G$
being $\lambda$.
Let $W(t) \subseteq [n]^t$ denote the set of $t$-tuples corresponding to
all walks of length $t-1$, with $N = \abs{W(t)} = n \cdot
d^{t-1}$. Ta-Shma proves that  for all $z \in \F_2^n$, $W(t)$ satisfies
\[
\bias(z) ~\leq~ \eps_0 
\quad \implies \quad 
\bias(\dsum_{W(t)}(z)) ~\leq~ (\eps_0 + 2\lambda)^{\lfloor (t-1)/2 \rfloor} \mcom
\]
\ie $W(t)$ is an \emph{$(\eps_0, \eps)$-parity sampler} for $\eps = (\eps_0 + 2\lambda)^{\lfloor (t-1)/2 \rfloor}$. 
Choosing $\eps_0 = 0.1$ and $\lambda = 0.05$ (say), we can choose $d =
O(1)$ and obtain the $\eps$-balanced code $\Cc' = \dsum_{W(t)}(\Cc)$ with
rate $d^{-(t-1)} = \eps^{O(1)}$ (although the right constants matter a
lot for optimal rates).
\paragraph{Decoding as constraint satisfaction.}
\newcommand{\maxtxor}{\problemmacro{MAX t-XOR}\xspace}
The starting point for our work is the framework in \cite{AJQST19}
which views the task of decoding $\tilde{y}$ with
$\Delta(\Cc',\tilde{y}) < (1-\eps)/4 - \delta$ (where the distance of
$\Cc'$ is $(1-\eps)/2$) as an instance of the $\maxtxor$ problem
(see~\cref{fig:approx_unique_dec}).
The goal is to find
\[
\argmin_{z \in \Cc}\Delta\inparen{\dsum_{W(t)}(z), \tilde{y}},
\]
which can be rephrased as
\[
\argmax_{z \in \Cc}  \Ex{w=(i_1, \ldots, i_t) \in W(t)}{\indicator{z_{i_1} + \cdots + z_{i_t} = \tilde{y}_w}}.
\]
It is possible to ignore the condition that $z \in \Cc$ if the
collection $W(t)$ is a slightly stronger parity sampler. For any solution
$\tilde{z} \in \mathbb{F}_2^n$ (not necessarily in $\Cc$) such that
	$$\Delta(\dsum_{W(t)}(\tilde{z}), \tilde{y}) < \frac{1-\eps}{4} + \delta,$$
we have
	$$\Delta(\dsum_{W(t)}(\tilde{z}), \dsum_{W(t)}(z)) < \frac{1-\eps}{2}$$
by the triangle inequality, and thus $\bias(\dsum_{W(t)}(z-\tilde{z})) > \eps$.
If $W(t)$ is not just an $(\eps_0, \eps)$-parity sampler, but in fact a
$((1+\eps_0)/2, \eps)$-parity sampler, this would imply
$\bias(z-\tilde{z}) > (1+\eps_0)/2$.
Thus, $\Delta(z,\tilde{z}) < (1-\eps_0)/4$ (or
$\Delta(z,\overline{\tilde{z}}) < (1-\eps_0)/4$) and we can use a
unique decoding algorithm for $\Cc$ to find $z$ given $\tilde{z}$.

\begin{figure}[h!]
	\centering
	\begin{tikzpicture}[scale=0.8, every node/.style={scale=0.8, fill=white, inner sep=0}]
		\fill[black] (0,0) circle (0.1cm);
		\draw (0,0) circle (2.5cm);
%		\draw (0,0) circle (3.5cm);
%		\node[left] at (260:2.5) {\begin{tabular}{c} List decoding radius \\ $(1/2-\sqrt{\epsilon})$ \end{tabular}};
%		\draw (0,0) circle (5cm);
		\node[right] at (80:1.9) {\begin{tabular}{c} Small approximation error $\delta$ \\ (comparable to $\epsilon$) \end{tabular}};				
		\draw[dashed](0,0) circle (2cm);
		\node[right] at (0.2,0) {$\tilde{y}$};
		\fill[black] (0.7,1) circle (0.1cm);
		\node[right] at (0.9,1) {$y$};
		\node[left] at (100:1) {\begin{tabular}{c} Unique decoding radius \\ $((1-\epsilon)/4)$ \end{tabular}};
%		\node[right] at (280:4.25) {\begin{tabular}{c} Constant approximation \\ error (0.1) \end{tabular}};
		\draw[thick, ->] (0,0) -- (100:2.5);
%		\draw[thick, ->] (0,0) -- (260:5);
		\draw[thick, <->] (80:2) -- (80:2.5);
%		\draw[thick, <->] (280:3.5) -- (280:5);
	\end{tikzpicture}
	\caption{Unique decoding ball along with error from approximation.} \label{fig:approx_unique_dec}
\end{figure}
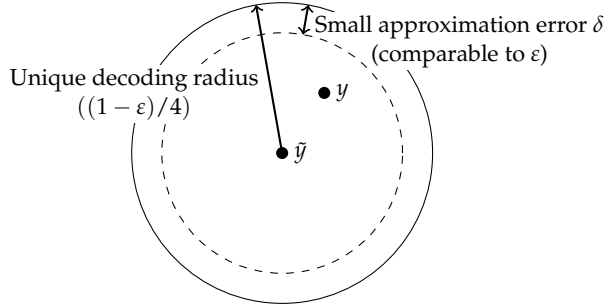

The task of finding such a $z \in \Cc$ boils down to finding a
solution $\tilde{z} \in \mathbb{F}_2^n$ to a $\maxtxor$ instance, up
to a an additive loss of $O(\delta)$ in the fraction of constraints
satisfied by the optimal solution. While this is hard to do in
general~\cite{Hastad97,Grigoriev01}, \cite{AJQST19} (building
on \cite{AJT19}) show that this can be done if the instance satisfies
a special property called
\emph{splittability}. 
To define this, we let $W[t_1,t_2] \subset [n]^{t_2-t_1+1}$ denote the
collection of $(t_2-t_1+1)$-tuples obtained by considering the indices
between $t_1$ and $t_2$ for all tuples in $W(t)$. We also assume that all
$w \in W[t_1,t_2]$ can be extended to the same number of tuples in $W(t)$
(which is true for walks).
\begin{definition}[Splittability ~(informal)] A collection $W(t) \subseteq [n]^t$ is said to be
$\tau$-splittable, if $t=1$ (base case) or there exists $t' \in [t-1]$ such that:
\begin{enumerate}
\item The matrix $\Ess \in \R^{W[1,t'] \times W[t'+1,t]}$
defined by $\Ess(w,w') ~=~ \indicator{ww' \in W}$ 
has normalized second singular value at most $\tau$ 
(where $ww'$ denotes the concatenated tuple).
\item The collections $W[1,t']$ and $W[t'+1,t]$ are $\tau$-splittable.
\end{enumerate}
\end{definition}
For example, considering walks in $G$ of length $3$ ($t=4$) and
$t'=2$, we get that $W[1,2] = W[3,4] = E$, the set of oriented edges
in $G$. Also $\Ess(w,w')=1$ if and only if the second vertex of $w$
and first vertex of $w'$ are adjacent in $G$. Thus, up to permutation
of rows and columns, we can write the normalized version of $\Ess$ as
$\Aye_G \otimes \Jay_d/d$ where $\Aye_G$ is normalized adjacency
matrix of $G$ and $\Jay_d$ denotes the $d \times d$ matrix of 1s.
Hence such a $W(t)$ satisfies $\sigma_2(\Ess) \le \tau$ with
$\tau=\sigma_2(\Aye_G)$, and a similar proof works for walks of all
lengths.

The framework in \cite{AJQST19} and \cite{AJT19} gives that if $W(t)$ is
$\tau$-splittable for $\tau = (\delta/2^t)^{O(1)}$, then the above
instance of $\maxtxor$ can be solved to additive error $O(\delta)$
using the Sum-of-Squares (SOS) SDP hierarchy.
Broadly speaking, splittability allows one to (recursively) treat
instances as expanding instances of problems with two ``tuple
variables'' in each constraint, which can then be analyzed using known
algorithms for 2-CSPs~\cite{BarakRS11, GuruswamiS11} in the SOS
hierarchy.
Combined with parity sampling, this yields a unique decoding
algorithm.
Crucially, this framework can also be extended to perform \emph{list
decoding}\footnote{ While unique decoding can be thought of as
recovering a single solution to a constraint satisfaction problem, the
goal in the list decoding setting can be thought of as obtaining a
``sufficiently rich'' set of solutions which forms a good cover. This is
achieved in the framework by adding an entropic term to the
semidefinite program, which ensures that the SDP solution satisfies
such a covering property.  }  up to a radius of $1/2 - \sqrt{\eps}
- \delta$ under a similar condition on $\tau$, which will be very
useful for our application.

While the above can yield decoding algorithms for suitably expanding
$G$, the requirement on $\tau$ (and hence on $\lambda$) makes the rate
much worse.
We need $\delta = O(\eps)$ (for unique decoding) and $t =
O(\log(1/\eps))$ (for parity sampling), which requires $\lambda
= \eps^{\Omega(1)}$, yielding only a quasipolynomial rate for the code
(recall that we could take $\lambda=O(1)$ earlier yielding polynomial
rates).

\paragraph{Unique decoding: weakening the error requirement.} We first observe that it is possible to
get rid of the dependence $\delta = O(\eps)$ above by using
the \emph{list decoding} algorithm for unique decoding. It suffices to
take $\delta = 0.1$ and return the closest element from the the list
of all codewords up to an error radius $1/2 - \sqrt{\eps} - 0.1$, if
we are promised that $\Delta(\tilde{y}, \Cc)$ is within the unique
decoding radius (see~\cref{fig:approx_req}).
However, this alone does not improve the rate as we still need the
splittability (and hence $\lambda$) to be $2^{-\Omega(t)}$ with $t =
O(\log(1/\eps))$.

\begin{figure}[h!]
	\centering
	\begin{tikzpicture}[scale=0.7, every node/.style={scale=0.7, fill=white, inner sep=0}]
		\fill[black] (0,0) circle (0.1cm);
		\draw (0,0) circle (2.5cm);
		\draw[dashed] (0,0) circle (4cm);
		\node[left] at (260:2.5) {\begin{tabular}{c} List decoding radius \\ $(1/2-\sqrt{\epsilon})$ \end{tabular}};
		\draw (0,0) circle (5cm);
%		\node[right] at (80:1.9) {\begin{tabular}{c} Small approximation error $\delta$ \\ (comparable to $\epsilon$) \end{tabular}};				
%		\draw(0,0) circle (2cm);
		\node[right] at (0.2,0) {$\tilde{y}$};                                
		\fill (0.7,1) circle (0.1cm);
		\node[right] at (0.9,1) {$y$};
		\node[left] at (100:1) {\begin{tabular}{c} Unique decoding radius \\ $(1/4-\epsilon/4)$ \end{tabular}};
		\node[right] at (280:4.25) {\begin{tabular}{c} Constant approximation \\ error (0.1) \end{tabular}};

		\draw[thick, ->] (0,0) -- (100:2.5);
		\draw[thick, ->] (0,0) -- (260:5);
%		\draw[thick, <->] (80:2) -- (80:2.5);
		\draw[thick, <->] (280:4) -- (280:5);
	\end{tikzpicture}
	\caption{Unique decoding and list decoding balls along with error from approximation. Note that the list decoding ball contains the unique decoding ball even after allowing for a relatively large amount of error.} \label{fig:approx_req}
\end{figure}
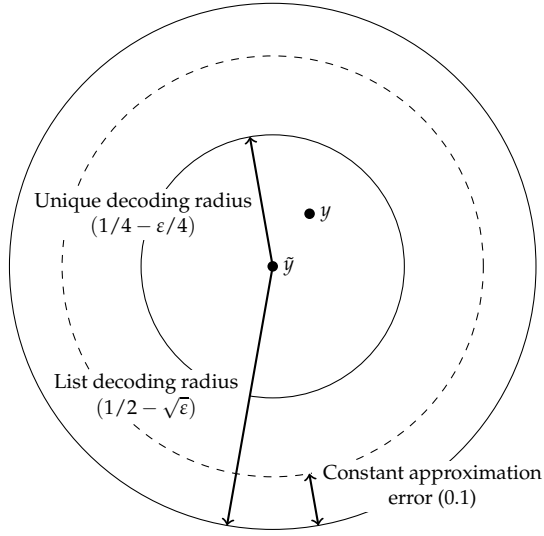

\paragraph{Code cascades: handling the dependence on walk length.} To avoid the
dependence of the expansion on the length $t-1$ of the walk (and hence
on $\eps$), we avoid the ``one-shot'' decoding above, and instead
consider a sequence of intermediate codes between $\Cc$ and $\Cc'$.
Consider the case when $t = k^{2}$, and instead of computing
$t$-wise sums of bits in each $z \in \F_2^n$, we first compute
$k$-wise sums according to walks of length $k-1$ on $G$, and then a
$k$-wise sum of these values. In fact, the second sum can also be
thought of as arising from a length $k-1$ walk on a different graph,
with vertices corresponding to (directed) walks with $k$ vertices in
$G$, and edges connecting $w$ and $w'$ when the last vertex of $w$ is
connected to the first one in $w'$ (this is similar to the matrix
considered for defining splittability). We can thus think of a
sequence of codes $\Cc_0, \Cc_1, \Cc_2$ with $\Cc_0 = \Cc$ and $\Cc_2
= \Cc'$, and both $\Cc_1$ and $\Cc_2$ being $k$-wise direct sums.
More generally, when $t = k^{\ell}$ for an appropriate constant $k$ we
can think of a sequence $\Cc = \Cc_0, \Cc_1, \ldots, \Cc_{\ell}
= \Cc'$, where each is an $k$-wise direct sum of the previous code,
obtained via walks of length $k-1$ (hence $k$ vertices) in an
appropriate graph. We refer to such sequences (defined formally in
\cref{sec:code_cascading}) as \emph{code cascades} (see~\cref{fig:code_cascade}).

\begin{figure}[h!]
  \centering
\begin{tikzpicture}[scale=0.7, every node/.style={scale=0.7}]]
\draw  (-6,3.5) rectangle node {$\mathcal{C}_0$} (-5,-2.5);
\draw  (-3.5,3.5) rectangle node {$\mathcal{C}_1$} (-2.5,-2.5);
\draw  (0.5,3.5) rectangle node {$\mathcal{C}_{i-1}$} (1.5,-2.5);
\draw  (3.5,3.5) rectangle node {$\mathcal{C}_{i}$} (4.5,-2.5);
\draw  (8,3.5) rectangle node {$\mathcal{C}_{\ell}$} (9,-2.5);
\node (v1) at (-5,0.5) {};
\node (v2) at (-3.5,0.5) {};
\node at (-1,0.5) {$\cdots$};
\node at (1,0.5) {};
\node (v4) at (3.5,0.5) {};
\node (v8) at (4.5,0.5) {};
\node at (6.5,0.5) {$\cdots$};
\node (v11) at (8,0.5) {};
\draw  (v1) edge[->] node[above] {$\lift$} (v2);
\node (v3) at (1.5,0.5) {};
\draw  (v3) edge[->] node[above] {$\lift$}  (v4);
\node (v5) at (-5.5,-3) {$\epsilon_0$};
\node at (-3,-3) {$\epsilon_1$};
\node at (1,-3) {$\epsilon_{i-1}$};
\node at (4,-3) {$\epsilon_i$};
\node at (8.5,-3) {$\epsilon_{\ell}=\epsilon$};
\draw  plot[smooth, tension=.7] coordinates {(v5)};
\draw  plot[smooth, tension=.7] coordinates {(v5)};
\draw  plot[smooth, tension=.7] coordinates {(-5,-3.25) (1,-6) (8,-3.25)};
\node at (1,-6.5) {\small Refined parity sampling via Ta-Shma's walk};
\draw  plot[smooth, tension=.7] coordinates {(1.5,-3) (2.5,-3.5) (3.5,-3)};
\node at (2.5,-4) {\small Crude parity sampling via Markov chain walk};
\node (v6) at (-2.5,0.5) {};
\node (v7) at (-1.5,0.5) {};
\node (v9) at (6,0.5) {};
\node (v10) at (7,0.5) {};
\draw  (v6) edge[->] (v7);
\draw  (v8) edge[->] (v9);
\draw  (v10) edge[->] (v11);
\node (v12) at (-0.5,0.5) {};
\node (v13) at (0.5,0.5) {};
\draw  (v12) edge[->] (v13);
\end{tikzpicture}
  \caption{Code cascading.}\label{fig:cascading}\label{fig:code_cascade}
\end{figure}
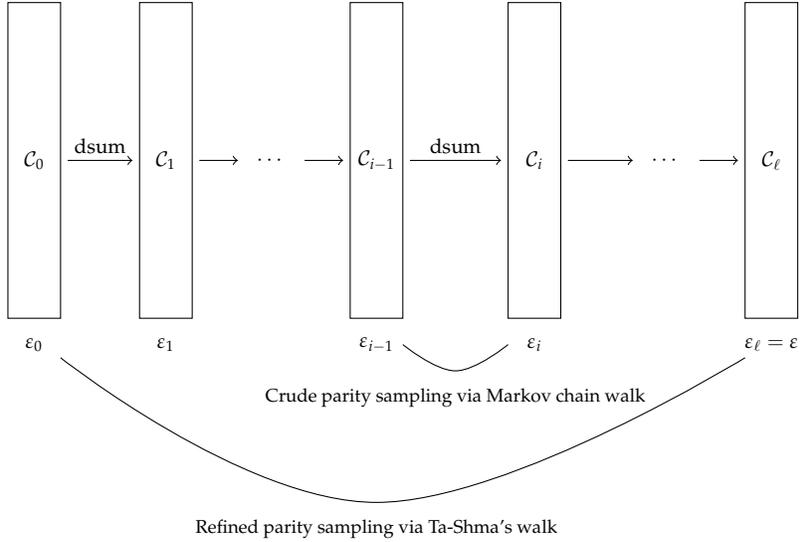

Instead of applying the decoding framework above to directly reduce
the decoding of a corrupted codeword from $\Cc'$ to the unique
decoding problem in $\Cc$, we apply it at each level of a cascade,
reducing the unique decoding problem in $\Cc_i$ to that in
$\Cc_{i-1}$.
If the direct sum at each level of the cascade is an
$(\eta_0, \eta)$-parity sampler, the list decoding algorithm at radius
$1/2 - \sqrt{\eta}$ suffices for unique decoding even if $\eta$ is a
(sufficiently small) constant independent of $\eps$. This implies that
we can take $k$ to be a (suitably large) constant.
This also allows the splittability (and hence $\lambda$) to be
$2^{-O(k)} = \Omega(1)$, yielding polynomial rates. We present the
reduction using cascades in
\cref{sec:main_result} and the parameter choices in \cref{sec:ta-shma_param_basic}.
The specific versions of the list decoding results from \cite{AJQST19}
needed here are instantiated in \cref{sec:instantiation_list_dec}.

While the above allows for polynomial rate, the \emph{running time} of
the algorithm is still exponential in the number of levels $\ell$
(which is $O(\log t) = O(\log\log(1/\eps))$) since the list decoding
for each level potentially produces a list of size $\poly(n)$, and
recursively calls the decoding algorithm for the previous level on
each element of the list.  We obtain a fixed polynomial time algorithm
by ``pruning'' the list at each level of the cascade before invoking
the decoding algorithm for the previous level, while only slightly
increasing the parity sampling requirements.  The details are
contained in \cref{sec:main_result}.

\paragraph{Working with Ta-Shma's construction.}
Finally, to obtain near-optimal rates, we need to work with with
Ta-Shma's construction, where the set of tuples $W(t) \subseteq [n]^t$
corresponds to walks arising from an $s$-wide replacement product of
$G$ with another expanding graph $H$.
One issue that arises is that the collection of walks $W(t)$ as defined
in \cite{Ta-Shma17} does not satisfy the important splittability
condition required by our algorithms. However, this turns out to be
easily fixable by modifying each step in Ta-Shma's construction to be
exactly according to the zig-zag product of Reingold, Vadhan and
Wigderson~\cite{RVW00}.
We present Ta-Shma's construction and this modification
in \cref{sec:ta_shma_tweaking}.

We also verify that the tuples given by Ta-Shma's construction satisfy
the conditions for applying the list decoding framework,
in \cref{sec:satisfying_framework}. While the sketch above stated this
in terms of splittability, the results in \cite{AJQST19} are in terms
of a more technical condition called \emph{tensoriality}. We show
in \cref{sec:satisfying_framework} that this is indeed implied by
splittability, and also prove splittability for (the modified version
of) Ta-Shma's construction.

\section{Ta-Shma's Construction: A Summary and Some Tweaks}\label{sec:ta_shma_tweaking}
In this section, we first discuss the $s$-wide replacement product
that is central to Ta-Shma's construction of optimal
$\epsilon$-balanced codes, and then we describe the construction
itself (we refer the reader to~\cite{Ta-Shma17} for formal details
beyond those we actually need here).

As mentioned before, we will also need to modify Ta-Shma's
construction~\cite{Ta-Shma17} a little to get \textit{splittability}
which is a notion of expansion of a collection $W(k) \subseteq [n]^k$
(and it is formally defined in~\cref{def:splittability}). The reason
for this simple modification is that this \emph{splittability}
property is required by the list decoding framework. Note that we are
not improving the Ta-Shma code parameters; in fact, we need to argue
why with this modification we can still achieve Ta-Shma's
parameters. Fortunately, this modification is simple enough that we
will be able to essentially reuse Ta-Shma's original
analysis. In \cref{sec:tweaks}, we will also have the opportunity to
discuss, at an informal level, the intuition behind some parameter
trade-offs in Ta-Shma codes which should provide enough motivation
when we instantiate these codes in~\cref{sec:ta-shma_param_basic}.

\subsection{The $s$-wide Replacement Product}\label{sec:s_wide_replacement_prod}

Ta-Shma's code construction is based on the so-called $s$-wide
replacement product~\cite{Ta-Shma17}.  This is a derandomization of
random walks on a graph $G$ that will be defined via a product
operation of $G$ with another graph $H$
(see \autoref{def:s_wide_replacement} for a formal definition). We
will refer to $G$ as the
\emph{outer} graph and $H$ as the \emph{inner} graph in this construction.

Let $G$ be a $d_1$-regular graph on vertex set $[n]$ and $H$ be a
$d_2$-regular graph on vertex set $[d_1]^s$, where $s$ is any positive
integer.  Suppose the neighbors of each vertex of $G$ are labeled 1,
2, \dots, $d_1$.  For $v \in V(G)$, let $v_G[j]$ be the $j$-th
neighbor of $v$.  The $s$-wide replacement product is defined by
replacing each vertex of $G$ with a copy of $H$, called a ``cloud''.
While the edges within each cloud are determined by $H$, the edges
between clouds are based on the edges of $G$, which we will define via
operators $\matr G_0, \matr G_1, \dots, \matr G_{s-1}$.  The $i$-th
operator $\matr G_i$ specifies one inter-cloud edge for each vertex
$(v, (a_0, \dots, a_{s-1})) \in V(G) \times V(H)$, which goes to the
cloud whose $G$ component is $v_G[a_i]$, the neighbor of $v$ in $G$
indexed by the $i$-th coordinate of the $H$ component.  (We will
resolve the question of what happens to the $H$ component after taking
such a step momentarily.)

Walks on the $s$-wide replacement product consist of steps with two
different parts: an intra-cloud part followed by an inter-cloud part.
All of the intra-cloud substeps simply move to a random neighbor in
the current cloud, which corresponds to applying the operator $\matr
I \otimes \matr A_H$, where $\matr A_H$ is the normalized adjacency
matrix of $H$.  The inter-cloud substeps are all deterministic, with
the first moving according to $\matr G_0$, the second according to
$\matr G_1$, and so on, returning to $\matr G_0$ for step number
$s+1$.  The operator for such a walk taking $t-1$ steps on the $s$-wide
replacement product is
$$
\prod_{i=0}^{t-2} \matr G_{i \bmod s} (\matr I \otimes \matr A_H).
$$

Observe that a walk on the $s$-wide replacement product yields a walk
on the outer graph $G$ by recording the $G$ component after each step
of the walk.  The number of $(t-1)$-step walks on the $s$-wide replacement
product is
$$
|V(G)| \cdot |V(H)| \cdot d_2^{t-1} = n \cdot d_1^s \cdot d_2^{t-1},
$$
since a walk is completely determined by its intra-cloud steps.  If
$d_2$ is much smaller than $d_1$ and $t$ is large compared to $s$,
this is less than $n d_1^{t-1}$, the number of $(t-1)$-step walks on $G$
itself.  Thus the $s$-wide replacement product will be used to
simulate random walks on $G$ while requiring a reduced amount of
randomness (of course this simulation is only possible under special
conditions, namely, when we are uniformly distributed on each cloud).

To formally define the $s$-wide replacement product, we must consider
the labeling of neighbors in $G$ more carefully.

\begin{definition}[Rotation Map]
	Suppose $G$ is a $d_1$-regular graph on $[n]$.
	For each $v \in [n]$ and $j \in [d_1]$, let $v_G[j]$ be the $j$-th neighbor of $v$ in $G$.
	Based on the indexing of the neighbors of each vertex, we define the rotation map~\footnote{This kind of map is denoted rotation map in the zig-zag terminology~\cite{RVW00}.}
        $\textup{rot}_G \colon [n] \times [d_1] \to [n] \times [d_1]$ such that for every $(v,j) \in [n] \times [d_1]$,
        $$
        \textup{rot}_G((v,j)) = (v',j') \iff v_G[j] = v' \text{ and } v'_G[j']=v.
        $$
	Furthermore, if there exists a bijection $\phi \colon [d_1] \to [d_1]$ such that for every $(v,j) \in [n] \times [d_1]$,
        $$
        \textup{rot}_G((v,j)) = (v_G[j],\phi(j)),
        $$
	then we call $\textup{rot}_G$ \emph{locally invertible}.
\end{definition}

If $G$ has a locally invertible rotation map, the cloud label after
applying the rotation map only depends on the current cloud label, not
the vertex of $G$.  In the $s$-wide replacement product, this
corresponds to the $H$ component of the rotation map only depending on
a vertex's $H$ component, not its $G$ component.  We define the
$s$-wide replacement product as described before, with the inter-cloud
operator $\matr G_i$ using the $i$-th coordinate of the $H$ component,
which is a value in $[d_1]$, to determine the inter-cloud step.

\begin{definition}[$s$-wide replacement product]\label{def:s_wide_replacement}
  Suppose we are given the following:
  \begin{itemize}
    \item A $d_1$-regular graph $G=([n],E)$ together with a locally invertible rotation map
          $\textup{rot}_G \colon [n] \times [d_1] \to [n] \times [d_1]$.
    \item A $d_2$-regular graph $H = ([d_1]^s,E')$.
  \end{itemize}
  And we define:
  \begin{itemize}
    \item For $i \in \set{0,1,\dots,s-1}$, we define $\textup{Rot}_i \colon [n] \times [d_1]^s \to [n] \times [d_1]^s$ as,
          for every $v \in [n]$ and $(a_0,\dots,a_{s-1}) \in [d_1]^s$,
          $$
          \textup{Rot}_i((v, (a_0,\dots,a_{s-1}))) \coloneqq (v', (a_0,\dots,a_{i-1},a_i',a_{i+1},\dots,a_{s-1})),
          $$
          where $(v',a_i') = \textup{rot}_G(v,a_i)$.
    \item Denote by $\matr G_i$ the operator realizing $\textup{Rot}_i$ and
          let $\matr A_H$ be the normalized random walk operator of $H$. Note that $\matr G_i$ is
          a permutation operator corresponding to a product of transpositions.
  \end{itemize}

  Then $t-1$ steps of the $s$-wide replacement product are given by the operator
  $$
  \prod_{i=0}^{t-2} \matr G_{i \bmod s} (\matr I \otimes \matr A_H).
  $$
\end{definition}

Ta-Shma instantiates the $s$-wide replacement product with an outer
graph $G$ that is a Cayley graph, for which locally invertible
rotation maps exist generically.

\begin{remark}
   Let $R$ be a group and $A \subseteq R$ where the set $A$ is closed under inversion. For every Cayley graph
   $\textup{Cay}(R,A)$, the map $\phi \colon A \to A$ defined as
   $\phi(g) = g^{-1}$ gives rise to the locally invertible rotation
   map
   $$
   \textup{rot}_{\textup{Cay}(R,A)}((r,a)) = (r\cdot a ,a^{-1}),
   $$
   for every $r \in R$, $a \in A$.
\end{remark}

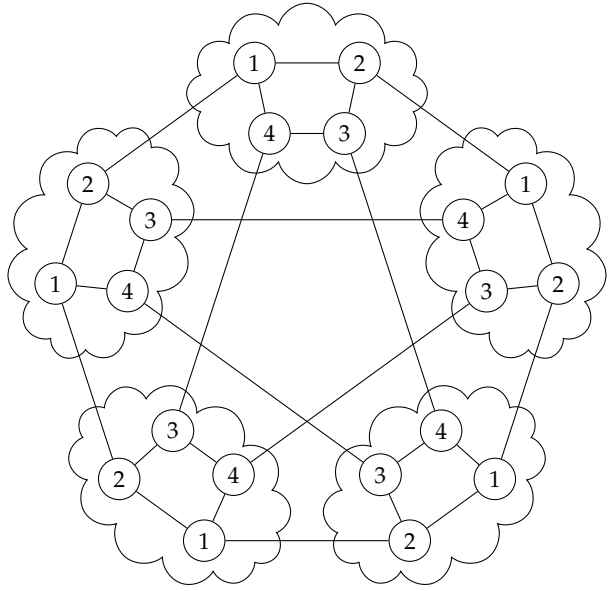
\begin{figure}[h!]
	\centering
	\begin{tikzpicture}[scale=0.8,every node/.style={scale=0.8}]
		\foreach \i in {0,...,4}
		{
			\node (v1\i) [circle,draw] at (102+72*\i:4.2) {1};
			\node (v2\i) [circle,draw] at (78+72*\i:4.2) {2};
			\node (v3\i) [circle,draw] at (78+72*\i:3) {3};
			\node (v4\i) [circle,draw] at (102+72*\i:3) {4};
		}
		\foreach \i in {0,...,4}
		{
			\draw (v1\i) -- (v2\i) -- (v3\i) -- (v4\i) -- (v1\i);
			
			\pgfmathparse{int(mod(\i+1,5))};
			\def\j{\pgfmathresult};
			\draw (v1\i) -- (v2\j);
			
			\pgfmathparse{int(mod(\i+2,5))};
			\def\j{\pgfmathresult};
			\draw (v4\i) -- (v3\j);
			
			\node[cloud,draw,minimum width=4cm,minimum height=3cm,rotate=72*\i,cloud puffs=12] at (90+72*\i:3.6) {};
		}
	\end{tikzpicture}
	\caption{An example of the 1-wide replacement product with outer graph $G = K_5$ and inner graph $H = C_4$.
		Vertices are labeled by their $H$ components.
		Note that the rotation map is locally invertible, with $\phi(1) = 2$, $\phi(2) = 1$, $\phi(3) = 4$, and $\phi(4) = 3$.}
\end{figure}

\subsection{The Construction}

Ta-Shma's code construction works by starting with a constant bias
code $\mathcal{C}_0$ in $\mathbb{F}_2^{n}$ and boosting to arbitrarily
small bias using direct sum liftings. Recall that the direct sum
lifting is based on a collection $W(t) \subseteq [n]^{t}$, which
Ta-Shma obtains using $t-1$ steps of random walk on the $s$-wide
replacement product of two regular expander graphs $G$ and $H$. The
graph $G$ is on $n$ vertices (same as blocklength of the base code)
and other parameters like degrees $d_1$ and $d_2$ of $G$ and $H$
respectively are chosen based on target code parameters.

To elaborate, every $t-1$ length walk on the replacement product gives
a sequence of $t$ outer vertices or $G$-vertices, which can be seen as
an element of $[n]^{t}$. This gives the collection $W(t)$ with $|W(t)|
= n \cdot d_1^s \cdot d_2^{t-1}$ which means the rate of lifted code
is smaller than the rate of $\mathcal{C}_0$ by a factor of $d_1^s
d_2^{t-1}$. However, the collection $W(t)$ is a parity sampler and
this means that the bias decreases (or the distance increases). The
relationship between this decrease in bias and decrease in rate with
some careful parameter choices allows Ta-Shma to obtain nearly optimal
$\epsilon$-balanced codes.

\subsection{Tweaking the Construction}\label{sec:tweaks}

Recall the first $s$ steps in Ta-Shma's construction are given by the
operator
$$
\matr G_{s-1}(\matr I \otimes \Aye_H) \matr G_{s-2} \cdots G_{1}(\matr I \otimes \Aye_H) \matr G_{0} (\matr I \otimes \Aye_H).
$$
Naively decomposing the above operator into the product of operators
$\prod_{i=0}^{s-1} \matr G_{i} (\matr I \otimes \Aye_H)$ is not good
enough to obtain the \emph{splittability} property which would hold
provided $\sigma_2(\matr G_{i}(\matr I \otimes \Aye_H))$ was small
for every $i$ in $\set{0,\ldots,s-1}$. However, each $\matr G_i (\matr
I \otimes \Aye_H)$ has $\abs{V(G)}$ singular values equal to $1$
since $G_{i}$ is an orthogonal operator and $(\matr I \otimes \Aye_H)$
has $\abs{V(G)}$ singular values equal to $1$. To avoid this issue we
will tweak the construction to be the following product
$$
\prod_{i=0}^{s-1} (\matr I \otimes \Aye_H) \matr G_{i} (\matr I \otimes \Aye_H).
$$

The operator $(\matr I \otimes \Aye_H) \matr G_{i} (\matr
I \otimes \Aye_H)$ is exactly the walk operator of the zig-zag product
$G \zigzag H$ of $G$ and $H$ with a rotation map given by the
(rotation map) operator $\matr G_i$. This tweaked construction is slightly
simpler in the sense that $G \zigzag H$ is an undirected graph. We
know by the zig-zag analysis that $(\matr I \otimes \Aye_H) \matr
G_{i} (\matr I \otimes \Aye_H)$ is expanding as long $G$ and $H$ are
themselves expanders.  More precisely, we have a bound that follows
from~\cite{RVW00}.

\begin{fact}\label{fact:zig_zag_bound}
  Let $G$ be an outer graph and $H$ be an inner graph used in the $s$-wide replacement product.
  For any integer $0 \leq i \leq s-1$,
  $$
  \sigma_2((I \otimes \Aye_H) G_i (I \otimes \Aye_H)) \leq \sigma_2(G) + 2 \cdot \sigma_2(H) + \sigma_2(H)^2.
  $$
\end{fact}

This bound will imply \emph{splittability} as shown
in~\cref{sec:ta-shma:splittability}. We will need to argue that this
modification still preserves the correctness of the parity sampling
and that it can be achieved with similar parameter trade-offs.

The formal definition of a length-$t$ walk on this slightly modified
construction is given below.
\begin{definition}
  Let $t \in \mathbb{N}$, $G$ be a $d_1$-regular graph and $H$
  be a $d_2$-regular graph on $d_1^s$ vertices. Given a starting vertex $(v,h) \in V(G) \times
  V(H)$, a $(t-1)$-step walk on the tweaked $s$-wide replacement product of $G$ and $H$ is
  a tuple $((v_0,h_0),\dots, (v_{t-1},h_{t-1})) \in (V(G) \times V(H))^{t}$ such that
  \begin{itemize}
    \item $(v_0,h_0) = (v,h)$, and
    \item for every $0 \le i < t-1$, we have $(v_i,h_i)$ adjacent to $(v_{i+1},h_{i+1})$
          in $(\matr I \otimes \Aye_H) \matr G_{i \bmod s} (\matr I \otimes \Aye_H)$.
  \end{itemize}
  Note that each $(\matr I \otimes \Aye_H) \matr G_{i \bmod s} (\matr I \otimes \Aye_H)$
  is a walk operator of a $d_2^2$-regular graph. Therefore, the starting vertex $(v,h)$ together
  with a degree sequence $(m_1,\dots,m_t) \in [d_2^2]^{t-1}$ uniquely defines a $(t-1)$-step walk.
\end{definition}

\subsubsection{Parity Sampling}\label{sec:tweaked_parity_sampling}

We argue informally why parity sampling still holds with similar
parameter trade-offs. Later in~\cref{sec:ta-shma_spectral_analysis},
we formalize a key result underlying parity sampling and,
in~\cref{sec:ta-shma_param_basic}, we compute the new trade-off
between bias and rate in some regimes. In~\cref{sec:s_wide_replacement_prod},
the definition of the original $s$-wide replacement product as a
purely graph theoretic operation was given. Now, we explain how
Ta-Shma used this construction for parity sampling obtaining codes
near the GV bound.

For a word $z \in \mathbb{F}_2^{V(G)}$ in the base code, let $\matr
P_z$ be the diagonal matrix, whose rows and columns are indexed by
$V(G) \times V(H)$, with $(\matr P_z)_{(v,h), (v,h)} =
(-1)^{z_v}$. Proving parity sampling requires analyzing the operator
norm of the following product
\begin{equation}\label{eq:tweaked_tashma}
  \matr P_z \prod_{i=0}^{s-1} \matr (\matr I \otimes \Aye_H) \matr G_{i} \matr P_z (\matr I \otimes \Aye_H),
\end{equation}
when $\bias(z) \le \epsilon_0$. Let $\one \in \mathbb{R}^{V(G)\times
V(H)}$ be the all-ones vector and
$W$ be the collection of all $(t-1)$-step walks on the tweaked $s$-wide
replacement product. Ta-Shma showed (and it is not difficult to
verify) that
$$
\bias\left(\dsum_{W}(z) \right) = \left\lvert \ip{\one}{\matr P_z \prod_{i=0}^{t-2} \matr (\matr I \otimes \Aye_H) \matr G_{i \bmod s} \matr P_z (\matr I \otimes \Aye_H) \one} \right\rvert.
$$
From the previous equation, one readily deduces that
$$
\bias\left(\dsum_{W}(z) \right) \le \sigma_1\left(\matr P_z \prod_{i=0}^{s-1} \matr (\matr I \otimes \Aye_H) \matr G_{i} \matr P_z (\matr I \otimes \Aye_H)\right)^{\lfloor (t-1)/s \rfloor}.
$$
Set $\matr B \coloneqq \matr P_z \prod_{i=0}^{s-1} \matr (\matr
I \otimes \Aye_H) \matr G_{i} \matr P_z (\matr I \otimes \Aye_H)$. To
analyze the operator norm of $\matr B$, we will first need some
notation. Note that $\matr B$ is an operator acting on the space
$\mathcal{V} = \mathbb{R}^{V(G)} \otimes \mathbb{R}^{V(H)}$. Two of
its subspaces play an important role in the analysis, namely,
\begin{gather*}
\mathcal{W}^{\parallel} = \textup{span}\set{ a \otimes b \in \mathbb{R}^{V(G)} \otimes \mathbb{R}^{V(H)}\mid b = \one} \text{ and }
\mathcal{W}^{\perp} = (\mathcal{W}^{\parallel})^{\perp}.
\end{gather*}
Note that the complement subspace is with respect to the standard inner product. Observe that $\mathcal{V} = \mathcal{W}^{\parallel} \oplus \mathcal{W}^{\perp}$. Given arbitrary
unit vectors $v,w \in \mathcal{V}$, Ta-Shma considers the inner product
\begin{equation}\label{eq:bilinear_ta_shma}
  \ip{v}{\prod_{i=0}^{s-1} \matr (\matr I \otimes \Aye_H) \matr G_{i} \matr P_z (\matr I \otimes \Aye_H) w}.
\end{equation}
Each time an operator $(\matr I \otimes \Aye_H)$ appears in the above
expression, the next step of the walk can take one out of $d_2$
possibilities and thus the rate suffers a multiplicative decrease of
$1/d_2$. We think that we are ``paying'' $d_2$ for this step of the
walk. The whole problem lies in the trade-off between rate and
distance, so the crucial question now is how much the norm decreases
as we pay $d_2$. For a moment, suppose that the norm always decreases
by a factor of $\lambda_2 \coloneqq \sigma_2(H)$ per occurrence of
$(\matr I \otimes \Aye_H)$. If in this hypothetical case we could
further assume $\lambda_2 = 1/\sqrt{d_2}$, then if $B$ was a product
containing $\lceil \log_{\lambda_2}(\epsilon) \rceil$ factors of $(\matr
I \otimes \Aye_H)$, the final bias would be at most $\epsilon$ and the
rate would have suffered a multiplicative decrease of (essentially)
$\epsilon^2$ and we would be done. 

Of course, this was an oversimplification. The general strategy is roughly the above, but a
beautiful non-trivial step is needed. Going back to the bilinear
form~\cref{eq:bilinear_ta_shma}, if $w \in \mathcal{W}^{\perp}$ (or $v \in
\mathcal{W}^{\perp}$), we pay $d_2$ and we do obtain a norm decrease of
$\lambda_2$. More generally, note that can decompose $w =
w^{\parallel} + w^{\perp}$ with
$w^{\parallel} \in \mathcal{W}^{\parallel}$ and
$w^{\perp} \in \mathcal{W}^{\perp}$ (decompose $v = v^{\parallel} +
v^{\perp}$ similarly) and we can carry this process iteratively
collecting factors of $\lambda_2$. However, we are stuck with several
terms of the form for $0\le k_1 \le k_2 < s$,
$$
\ip{v_{k_1}^{\parallel}}{\prod_{i=k_1}^{k_2} \matr (\matr I \otimes \Aye_H) \matr G_{i} \matr P_z (\matr I \otimes \Aye_H) w_{k_2}^{\parallel}},
$$
with
$v_{k_1}^{\parallel},w_{k_2}^{\parallel} \in \mathcal{W}^{\parallel}$,
and for which the preceding naive norm decrease argument fails. This
is the point in the analysis where the structure of the $s$-wide
replacement product is used. Since $v_{k_1}^{\parallel},
w_{k_2}^{\parallel} \in \mathcal{W}^{\parallel}$, these vectors are uniform on each ``cloud'', i.e., copy of $H$. Recall
that a vertex in $H$ is an $s$-tuple $(m_1,\dots,m_s) \in
[d_1]^s$. Ta-Shma leverages the fact of having a uniform such tuple to
implement $k_2 -k_1 +1$ (up to $s$) steps of random walk on $G$. More
precisely, Ta-Shma obtains the following beautiful result:

\begin{theorem}[Adapted from Ta-Shma~\cite{Ta-Shma17}]\label{theo:ta_shma_simulation}
 Let $G$ be a locally invertible graph of degree $d_1$, $H$ be a Cayley graph on $\mathbb{F}_2^{s\log d_1}$, and $0 \le k_1 \le k_2 < s$ be integers. If $v^{\parallel} = v \otimes 1$ and $w^{\parallel} = w \otimes 1$, then
 $$
 \ip{v^{\parallel}}{\prod_{i=k_1}^{k_2} \matr G_{i} \matr (\matr I \otimes \Aye_H) \matr P_z  w^{\parallel}} =  \ip{v}{\left(\matr \Aye_G \matr M_z\right)^{k_2-k_1+1} w}
 $$
 where $\matr M_z \in \mathbb{R}^{V(G) \times V(G)}$ is the diagonal matrix defined as $(\matr M_z)_{v,v} \coloneqq (-1)^{z_v}$ for $v \in V(G)$.
\end{theorem}

\begin{remark}
Note that the walk operator in this theorem corresponds to the original construction. \autoref{theo:ta_shma_simulation} was used by Ta-Shma to obtain \autoref{fact:ta-shma_main} whose \autoref{cor:tweaked_ta-shma_spectral_analysis} corresponds to the modified construction.
\end{remark}

Ta-Shma proved \autoref{theo:ta_shma_simulation} under the more general condition that $H$ is 0-pseudorandom. Roughly speaking, this property means that if we start with a distribution that is uniform over the clouds, and walk according to fixed $H$-steps $j_0,j_1,\cdots , j_{s-1} $, then the distribution of $G$-vertices obtained will be identical to the distribution obtained if we were doing the usual random walk on $G$. We will always choose $H$ to be a Cayley graph on $\mathbb{F}_2^{s\log d_1}$, which will imply that $H$ is also 0-pseudorandom. The proof of \autoref{theo:ta_shma_simulation} crucially uses the product structure of $\mathbb{F}_2^{s\log d_1}$: every vertex of $H$ can be represented by $s$ registers of $\log d_1$ bits each, and both inter-cloud and intra-cloud steps can be seen as applying register-wise bijections using some canonical mapping between $[d_1]$ and $\mathbb{F}_2^{\log d_1}$.

Ta-Shma's original parity sampling proof required $\epsilon_0 + 2\theta
+ 2\sigma_2(G) \le \sigma_2(H)^2$, where $\epsilon_0$ is the initial bias
and $\theta$ is an error parameter arising from a number theoretic
construction of Ramanujan graphs for the outer graph $G$. This is because $\epsilon_0 + 2\theta
+ 2\sigma_2(G)$ is the reduction of bias in every two steps while taking a walk on $G$ (see \autoref{theo:ta_shma_bias_simple}). Having $\epsilon_0 + 2\theta+ 2\sigma_2(G) \le \sigma_2(H)^2$ ensured that after establishing \autoref{theo:ta_shma_simulation}, we were collecting enough reduction for $d_2^2$ price we paid for two steps. In the modified construction, we now have $d_2^2$ possibilities for each step in $(\matr
I \otimes \Aye_H^2)$ (so $d_2^4$ price for two steps), and so if instead
we have $\epsilon_0 + 2\theta + 2\sigma_2(G) \le \sigma_2(H)^4$ in the
modified construction, we claim that the correctness of the parity
sampling analysis is preserved as well as (essentially) the trade-off
between walk length and norm decay. Fortunately, Ta-Shma's parameters decouple and we can choose parameters to satisfy the above requirement.

\begin{remark}
  This modification on the $s$-replacement product of $G$ and $H$
  essentially~\footnote{Except at the first and last factors in the
  product of operators.} amounts to taking a different inner graph $H$
  which can be factored as $H=\sqrt{H} \sqrt{H}$ (and is still
  $0$-pseudorandom).
\end{remark}

\subsubsection{Spectral Analysis of the Modified Construction}\label{sec:ta-shma_spectral_analysis}

We formally show that we don't loose much by going from
Ta-Shma's original $s$-wide product construction to its tweaked version. The
key technical result obtained by Ta-Shma is the following, which is
used to analyze the bias reduction as a function of the total number
walk steps $t-1$.
 
\begin{fact}[Theorem 24 abridged~\cite{Ta-Shma17}]\label{fact:ta-shma_main}
 If $H$ is a Cayley graph on $\mathbb{F}_2^{s\log d_1}$ and $\epsilon_0 + 2 \cdot \theta + 2 \cdot \sigma_2(G) \le \sigma_2(H)^2$, then
 $$
 \norm{\prod_{i=0}^{s-1} \matr P_z \matr G_i (\matr I \otimes \matr A_H)}_{\textup{op}} \le \sigma_2(H)^s + s \cdot \sigma_2(H)^{s-1} + s^2 \cdot \sigma_2(H)^{s-3},
 $$
 where $\matr P_z \in \mathbb{R}^{(V(G)\times V(H)) \times (V(G)\times V(H))}$ is the \emph{sign operator} of a $\epsilon_0$ biased word $z \in \mathbb{F}_2^{V(G)}$
 defined as a diagonal matrix with $(P_z)_{(v,h),(v,h)} = (-1)^{z_v}$ for every $(v,h) \in V(G) \times V(H)$.
\end{fact}

We reduce the analysis of Ta-Shma's tweaked construction
to~\cref{fact:ta-shma_main}. In doing so, we only lose one extra step
as shown below.

\begin{corollary}\label{cor:tweaked_ta-shma_spectral_analysis}
 If $H^2$ is a Cayley graph on $\mathbb{F}_2^{s\log d_1}$ and $\epsilon_0 + 2 \cdot \theta + 2 \cdot \sigma_2(G) \le \sigma_2(H)^4$, then
 $$
 \norm{\prod_{i=0}^{s-1} (\matr I \otimes \matr A_H) \matr P_z \matr G_i (\matr I \otimes \matr A_H)}_{\textup{op}} \le \sigma_2(H^2)^{s-1} + (s-1) \cdot \sigma_2(H^2)^{s-2} + (s-1)^2 \cdot \sigma_2(H^2)^{s-4},
 $$
 where $\matr P_z$ is the \emph{sign operator} of an $\epsilon_0$-biased word $z \in \mathbb{F}_2^{V(G)}$
 as in~\cref{fact:ta-shma_main}.
\end{corollary}

\begin{proof}
  We have
  \begin{align*}
    \norm{\prod_{i=0}^{s-1} (\matr I \otimes \matr A_H) \matr P_z \matr G_i (\matr I \otimes \matr A_H)}_{\text{op}} &\le \norm{(\matr I \otimes \matr A_H)}_{\text{op}} \norm{\prod_{i=1}^{s-1}  \matr P_z \matr G_i (\matr I \otimes \matr A_H^2)}_{\text{op}} \norm{\matr P_z \matr G_{0} (\matr I \otimes \matr A_H)}_{\text{op}}\\
                                                  &\le \norm{\prod_{i=1}^{s-1}  \matr P_z \matr G_i (\matr I \otimes \matr A_H^2)}_{\text{op}}\\
                                                  &\le \sigma_2(H^2)^{s-1} + (s-1) \cdot \sigma_2(H^2)^{s-2} + (s-1)^2 \cdot \sigma_2(H^2)^{s-4},
  \end{align*}
  where the last inequality follows from~\cref{fact:ta-shma_main}.
\end{proof}

\begin{remark}
  We know that in the modified construction $H^2$ is a Cayley graph since $H$ is a Cayley graph.
\end{remark}

From this point onward, we will be working exclusively with the
modified construction instead of using it in its original
form.  Any references to Ta-Shma's construction or the $s$-wide
replacement product will actually refer to the modified versions
described in this section.

\section{Code Cascading}\label{sec:code_cascading}

A code cascade is a sequence of codes generated by starting with a
base code $\Cc_0$ and recursively applying lifting operations.

\begin{definition}
   We say that a sequence of codes $\Cc_0, \Cc_1,\ldots, \Cc_{\ell}$
   is a \emph{code cascade} provided $\Cc_{i} =\lift_{W_i(t_i)}(\Cc_{i-1})$ for
   every $i \in [\ell]$.
   Each $W_i(t_i)$ is a subset of $[n_{i-1}]^{t_i}$, where $n_{i-1}=|W_{i-1}(t_{i-1})|$ is the block length of the code $\Cc_{i-1}$.
\end{definition}

Let us see how code cascades may be useful for decoding. Suppose we
wish to lift the code $\Cc_0$ to $\Cc_{\ell}$, and there is some
$W(t) \subseteq [n_0]^t$ such that $\Cc_{\ell} = \dsum_{W(t)}
(\Cc_0)$. In our case of bias boosting, this $t$ will depend on the
target bias $\epsilon$. However, the expansion requirement of the
list-decoding framework of \cite{AJQST19} has a poor dependence on
$t$. A way to work around this issue is to go from $\Cc_0$ to
$\Cc_{\ell}$ via a code cascade as above such that each $t_i$ is a
constant independent of the final bias but $\prod\limits_{i=1}^{\ell}
t_i = t$ (which means $\ell$ depends on $\epsilon$). The final code
$\Cc_{\ell}$ of the cascade is the same as the code obtained from
length-$(t-1)$ walks. While decoding will now become an $\ell$-level
recursive procedure, the gain from replacing $t$ by $t_i$ will
outweigh this loss, as we discuss below.

\subsection{Warm-up: Code Cascading Expander Walks}\label{sec:warmup}

We now describe the code cascading construction and unique decoding
algorithm in more detail. Let $G=(V,E)$ be a $d$-regular graph with
uniform distribution over the edges. Let $m$ be a sufficiently large
positive integer, which will be the number of vertices of the walks
used for the lifting between consecutive codes in the cascade.  At
first, it will be crucial that we can take $m=O(1)$ so that the
triangle inequality arising from the analysis of the lifting between
two consecutive codes involves a constant number of terms. We
construct a recursive family of codes as follows.

\begin{itemize}
  \item Start with a code $\mathcal{C}_0$ which is linear
        and has constant bias $\epsilon_0$.
  \item Define the code $\mathcal{C}_1 = \dsum_{W(m)}(\Cc_0)$, which is the direct sum lifting over the collection $W(m)$ of all
        length-$(m-1)$ walks on $G$ using the code $\mathcal{C}_0$.
  \item Let $\widehat{G}_i = (V_i,E_i)$ be the (directed) graph where $V_i$ is the collection of all walks on $m^i$ vertices on $G$ with
        two walks $(v_1,\dots,v_{m^i})$ and $(u_1,\dots,u_{m^i})$ connected iff $v_{m^i}$ is adjacent to $u_1$ in $G$.
  \item Define $\mathcal{C}_i$ to be the direct sum lifting on the collection $W_i(m)$ of all length-$(m-1)$ walks on $G_{i-1}$ using
        the code $\mathcal{C}_{i-1}$, i.e., $\mathcal{C}_i = \dsum_{W_i(m)}(\Cc_{i-1})$.
  \item Repeat this process to yield a code cascade $\mathcal{C}_0,\dots,\mathcal{C}_{\ell}$.
\end{itemize}

Thanks to the definition of the graphs $\widehat{G}_i$ and the recursive nature
of the construction, the final code $\Cc_{\ell}$ is the same as the
code obtained from $\Cc_0$ by taking the direct sum lifting over all
walks on $t = m^{\ell}$ vertices of $G$. We can use Ta-Shma's analysis
(building on the ideas of Rozenman and Wigderson~\cite{RW08})
for the simpler setting of walks over a single expander graph to
determine the amplification in bias that occurs in going from $\Cc_0$
all the way to $\Cc_{\ell}$.

\begin{theorem}[Adapted from Ta-Shma~\cite{Ta-Shma17}]\label{theo:ta_shma_bias_simple}
  Let $\Cc$ be an $\epsilon_0$-balanced linear code, and let $\Cc' = \dsum_{W(t)}(\Cc)$ be the direct sum lifting
  of $\Cc$ over the collection of all length-$(t-1)$ walks $W(t)$ on a graph $G$. Then
  $$
  \bias(\Cc') \leq (\epsilon_0 + 2 \sigma_2(G))^{\floor{(t-1)/2}}.
  $$
\end{theorem}

If $\sigma_2(G) \leq \epsilon_0/2$ and $\ell
= \ceil{\log_m(2 \log_{2\epsilon_0}(\epsilon) + 3)}$, taking $t = m^{\ell} \geq 2 \log_{2 \epsilon_0}(\epsilon)+3$ in the above theorem shows that the
final code $\Cc_{\ell}$ is $\epsilon$-balanced.  Observe that the required
expansion of the graph $G$ only depends on the constant initial bias
$\epsilon_0$, not on the desired final bias $\epsilon$.  It will be
important for being able to decode with better parameters that both
$\sigma_2(G)$ and $m$ are constant with respect to $\epsilon$; only
$\ell$ depends on the final bias (with more care we can make
$\sigma_2(G)$ depend on $\epsilon$, but we restrict this analysis to
Ta-Shma's refined construction on the $s$-wide replacement product).

As mentioned before, to uniquely decode $\Cc_{\ell}$ we will
inductively employ the list decoding machinery for expander walks
from~\cite{AJQST19}.  The list decoding algorithm can decode a direct
sum lifting $\Cc' = \dsum_{W(m)}(\Cc)$ as long as the graph $G$ is
sufficiently expanding, the walk length $m-1$ is large enough, and the
base code $\Cc$ has an efficient unique decoding algorithm
(see \cref{theo:list_dec_hammer} for details).

The expansion requirement ultimately depends on the desired list
decoding radius of $\Cc'$, or more specifically, on how close the list
decoding radius is to $1/2$.  If the distance of $\Cc'$ is at most
$1/2$, its unique decoding radius is at most $1/4$, which means list
decoding at the unique decoding radius is at a constant difference
from $1/2$ and thus places only a constant requirement on the
expansion of $G$.  In the case of the code cascade $\Cc_i
= \dsum_{W_i(m)}(\Cc_{i-1})$, unique decoding of $\Cc_{i-1}$ is
guaranteed by the induction hypothesis.  It is not too difficult to
see that each graph $\widehat{G}_i$ will have the same second singular value as
$G$, so we can uniquely decode $\Cc_i$ if $G$ meets the constant
expansion requirement and $m$ is sufficiently large.

\subsection{Code Cascading Ta-Shma's Construction}\label{sec:code_cascading_ta_shma}

We will now describe how to set up a code cascade based on walks on an $s$-wide replacement product.
Consider the $s$-wide replacement product of the outer graph
$G$ with the inner graph $H$. The first $s$ walk steps are given by
the walk operator
$$
\prod_{i=0}^{s-1} (\matr I \otimes \Aye_H) \matr G_{i} (\matr I \otimes \Aye_H).
$$
Let $\matr A_{s-1} \coloneqq (\matr I \otimes \Aye_H) \matr G_{s-2} (\matr
I \otimes \Aye_H) \cdots (\matr I \otimes \Aye_H) \matr G_{0} (\matr
I \otimes \Aye_H)$. If the total walk length $t-1$ is a multiple of $s$,
the walks are generated using the operator
$$
\left( (\matr I \otimes \Aye_H) \matr G_{s-1} (\matr I \otimes \Aye_H) \matr A_{s-1} \right)^{(t-1)/s}.
$$
Here $(\matr I \otimes \Aye_H) \matr G_{s-1}
(\matr I \otimes \Aye_H)$ is used as a ``binding'' operator to connect two
walks containing $s$ vertices at level $\Cc_2$, $s^2$ vertices at
level $\Cc_3$, and so on. More precisely, we form the following code
cascade.
\begin{itemize}
  \item $\Cc_0$ is an $\epsilon_0$-balanced linear code efficiently uniquely decodable from a constant radius.
  \item $\Cc_1 = \lift_{W_1(s)}(\Cc_0)$, where $W_1(s)$ is the set of length-(s-1) walks given by the operator
        $$
        \underbrace{(\matr I \otimes \Aye_H) \matr G_{s-2} (\matr I \otimes \Aye_H)}_{\text{$(s-2)$th step}} \cdots \underbrace{(\matr I \otimes \Aye_H) \matr G_{0} (\matr I \otimes \Aye_H)}_{\text{$0$th step}}.
        $$
  \item $\Cc_2 = \lift_{W_2(s)}(\Cc_1)$, where $W_2(s)$ is the set of length-$(s-1)$ walks over the vertex set $W_1(s)$ (with the latter being the set of length-$(s-1)$ walks on the replacement product graph as mentioned above).
  \item $\Cc_{i+1} = \lift_{W_{i+1}(s)}(\Cc_i)$, where $W_{i+1}(s)$ is the set of length-$(s-1)$ walks~\footnote{For simplicity we chose the number of vertices in all walks of the cascade to be $s$,
        but it could naturally be some $s_i \in \mathbb{N}$ depending on $i$.}
        over the vertex set $W_i(s)$.
        Similarly to the cascade of expander walks above,
        the lift can be thought of as being realized by taking walks using a suitable operator analogous to $\widehat{G}_i$. Since its description is more technical we postpone
        its definition (see~\cref{def:swap_operator}) to~\cref{sec:ta-shma:splittability} where it is actually used.
  \item $\Cc_{\ell}$ denotes the final code in the sequence, which will later be chosen so that its bias is at most $\epsilon$.
\end{itemize}

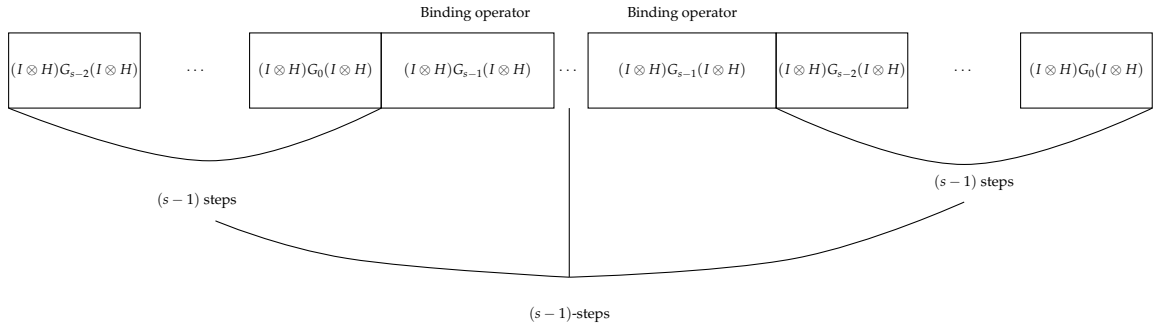
\begin{figure}[h!]
  \centering
  \begin{tikzpicture}[scale=0.5, every node/.style={scale=0.5}]]
\draw  (7,4.5) rectangle node {$(I \otimes H) G_{s-1} (I \otimes H)$ } (12,2.5);
\draw  (-8.4067,4.5) rectangle  node {$(I \otimes H) G_{s-2} (I \otimes H)$ }(-4.9067,2.5);
\draw  (-2,4.5) rectangle node {$(I \otimes H) G_{0} (I \otimes H)$} (1.5,2.5) node (v2) {};
\node at (-3.4067,3.5) {$\cdots$};
\draw  (1.5,4.5) rectangle node {$(I \otimes H) G_{s-1} (I \otimes H)$ } (6.0933,2.5) node (v3) {};
\draw  (12,4.5) rectangle  node {$(I \otimes H) G_{s-2} (I \otimes H)$ }(15.5,2.5);
\draw  (18.5,4.5) rectangle node {$(I \otimes H) G_{0} (I \otimes H)$} (22,2.5) node (v6) {};
\node at (17,3.5) {$\cdots$};
\node (v1) at (-8.4067,2.5) {};
\node at (-2,2.4616) {};
\draw  plot[smooth, tension=.7] coordinates {(v1)};
\draw  plot[smooth, tension=.7] coordinates {(v1) (-3.0736,1.0977) (v2)};
\node (v7) at (-2.9067,-0.5) {};
\node at (-3.4067,0.05) {$(s-1)$ steps};
\node (v5) at (12,2.5) {};
\draw  plot[smooth, tension=.7] coordinates {(v5) (17,1) (v6)};
\node (v9) at (6.5,2.5) {};
\node (v10) at (17,0) {};
\node at (17.2561,0.5208) {$(s-1)$ steps};
\node at (6.5,3.5) {$\cdots$};
\node (v8) at (6.5,-2) {};
\draw  plot[smooth, tension=.7] coordinates {(v7)};
\draw  plot[smooth, tension=.7] coordinates {(v7) (0.5933,-1.5) (v8)};
\draw  plot[smooth, tension=.7] coordinates {(v9) (6.5,0.5) (v8)};
\draw  plot[smooth, tension=.7] coordinates {(v8) (12.5,-1.5) (v10)};
\node at (6.5,-3) {$(s-1)$-steps};
\node at (4,5) {Binding operator};
\node at (9.5,5) {Binding operator};
\end{tikzpicture}
  \caption{Two levels of code cascading for Ta-Shma's construction
           involving codes $\Cc_0$, $\Cc_1$ and $\Cc_2$ (to make the
           notation compact we used $H$ to denote
           $\Aye_H$).}\label{fig:ta-shma_cascading}
\end{figure}

\section{Unique Decoding of Ta-Shma Codes}\label{sec:main_result}

We show how code cascading together with list decoding for each level
of the cascade allow us to obtain an efficient unique decoding
algorithm for Ta-Shma's construction.  We obtain a sequence of results
of increasing strength culminating in~\cref{theo:main} (which we
recall below for convenience). The approach is as follows: we use
several different instantiations of Ta-Shma's construction, each
yielding a value of $s$ (for the $s$-wide replacement product) and
expansion parameters for the family of outer and inner graphs, and
show how the list decoding framework can be invoked in the associated
cascade for each one.

\TheoMainUniqueDec*

In this section, we will fit these objects and tools together assuming
the parameters are chosen to achieve the required rates and the
conditions for applying the list decoding results are satisfied.  The
concrete instantiations of Ta-Shma codes are done
in~\cref{sec:ta-shma_param_basic}.  Establishing that the list
decoding framework can be applied to this construction is done
in~\cref{sec:satisfying_framework} after which the framework is
finally instantiated in~\cref{sec:instantiation_list_dec}.

Ta-Shma uses the direct sum lifting on an $s$-wide replacement product
graph to construct a family of $\epsilon$-balanced codes
$\Cc_{N,\epsilon,\beta}$ with rate $\Omega(\epsilon^{2+\beta})$ and
finds parameters for such codes to have the required bias and rate.
We will discuss unique decoding results for several versions of these
codes.  Throughout this section, we will use collections $W(k)$ which
will always be either the set of walks with $k=s$ vertices on an
$s$-wide replacement product graph (corresponding to the first level
of the code cascade), which we denote $W[0,s-1]$, or a set of walks
where the vertices are walks on a lower level of the code cascade.

\subsection{Unique Decoding via Code Cascading}

To perform unique decoding we will use the machinery of list decoding
from~\cref{theo:list_dec_hammer} (proven later
in~\cref{sec:instantiation_list_dec}), which relies on the list
decoding framework of~\cite{AJQST19}.  Proving that this framework can
be applied to Ta-Shma's construction is one of our technical
contributions.

\begin{theorem}[List decoding direct sum lifting]\label{theo:list_dec_hammer}
  Let $\eta_0 \in (0,1/4)$ be a constant, $\eta \in (0,\eta_0)$, and
  	$$k \geq k_0(\eta) \coloneqq \Theta(\log(1/\eta)).$$
  Suppose $\Cc \subseteq \F_2^n$ is an $\eta_0$-balanced linear code and $\Cc' = \dsum_{W(k)}(\Cc)$ is the direct sum lifting of $\Cc$ on a $\tau$-splittable collection of walks $W(k)$.
  There exists an absolute constant $K > 0$ such that if
  	$$\tau \leq \tau_0(\eta,k) \coloneqq \frac{\eta^{8}}{K \cdot k \cdot 2^{4k}},$$
  then the code $\Cc'$ is $\eta$-balanced and can be efficiently list decoded in the following sense:

  If $\tilde{y}$ is \lict{\sqrt{\eta}} to $\Cc'$, then we
         can compute the list
         $$
         \mathcal{L}(\tilde{y},\Cc,\Cc')\coloneqq \left\{(z,\dsum_{W(k)}(z)) \mid z \in \Cc, \Delta\parens*{\dsum_{W(k)}(z),\tilde{y}} \le \frac{1}{2} - \sqrt{\eta}\right\}
         $$
         in time
         $$
         n^{O(1/\tau_0(\eta,k)^4)} \cdot f(n),
         $$
         where $f(n)$ is the running time of a unique decoding algorithm for $\Cc$.
         Otherwise, we return $\mathcal{L}(\tilde{y},\Cc,\Cc')=\emptyset$ with the same running time of the preceding case.
\end{theorem}

Note that the requirement on $k$ in the above theorem is necessary for
the lifted code $\Cc'$ to be $\eta$-balanced.  Splittability will
imply that the walk collection $W(k)$ is expanding, which gives us
parity sampling for large $k$.  Specifically, $k$ must be large enough
for $W(k)$ to be a $(1/2 + \eta_0/2, \eta)$-parity sampler.

Applying the list decoding tool above, we can perform unique decoding
in the regime of $\eta_0$, $\eta$, and $k$ being constant.  With these
choices the expansion required for splittability and the parity
sampling strength are only required to be constants.

\begin{lemma}[Decoding Step]\label{lemma:dec_step}
  Let $\eta_0 \in (0,1/4)$ and $\eta < \min\{\eta_0, 1/16\}$.
  If $W(k)$ is a walk collection on vertex set $[n]$ with $k \geq k_0(\eta)$ and splittability $\tau \leq \tau_0(\eta,k)$, where $k_0$ and $\tau_0$ are as in~\cref{theo:list_dec_hammer}, we have the following unique decoding property:

  If $\Cc \subseteq \mathbb{F}_2^n$ is an \text{$\eta_0$-balanced} linear code that can be uniquely decoded in time $f(n)$, then $\Cc' = \dsum_{W(k)}(\Cc)$ is an $\eta$-balanced code that can be uniquely decoded in time $n^{O(1/\tau_0(\eta,k)^4)} \cdot f(n)$.
\end{lemma}

\begin{proof}
Using \cref{theo:list_dec_hammer}, we can list decode $\Cc'$ up to a radius of $1/2 - \sqrt{\eta}$ for any $\eta$ if we have the appropriate parameters $k$ and $\tau$.
Let $\tilde y \in \F_2^{W(k)}$ be a received word within the unique decoding radius of $\Cc'$.
To perform unique decoding, we simply run the list decoding algorithm on $\tilde y$ and return the codeword on the resulting list which is closest to $\tilde y$; this will yield the correct result as long as the list decoding radius is larger than the unique decoding radius.
It suffices to have $1/2 - \sqrt{\eta} > 1/4 \geq \Delta(\Cc')/2$.
We choose parameters as follows:

\begin{enumerate}
\item Take $\eta < 1/16$ to ensure $1/2 - \sqrt{\eta} > 1/4$.
\item Let $k_0 =  \Theta(\log(1/\eta))$ be the smallest integer satisfying the assumption in~\cref{theo:list_dec_hammer} with the chosen $\eta$.
Take $k \geq k_0$.
\item Take $\tau \leq \tau_0(\eta,k) = \eta^8 / (K \cdot k \cdot 2^{4k})$.
\end{enumerate}

Note that $k$ and $\tau$ satisfy the conditions of~\cref{theo:list_dec_hammer}, so we can use this theorem to list decode a received word $\tilde y$ in time $n^{O(1/\tau_0(\eta,k)^4)} \cdot f(n)$.
To unique decode, we return the closest $y$ on the list to $\tilde y$ (or failure if the list is empty).
\end{proof}

Iteratively using the decoding step given by~\cref{lemma:dec_step}
above, we obtain unique decodability of all codes in a cascade
(under suitable assumptions).

\begin{lemma}[Code Cascade Decoding]\label{lemma:code_cascading}
  Let $\eta_0 \in (0,1/4)$ and $\eta < \min\{\eta_0, 1/16\}$.
  Suppose $\Cc_0 \subseteq \mathbb{F}_2^{n_0},\Cc_1 \subseteq \mathbb{F}_2^{n_1}, \dots,\Cc_{\ell} \subseteq \mathbb{F}_2^{n_{\ell}}$ is a code cascade where $\Cc_0$ is an $\eta_0$-balanced linear code that can be uniquely decoded in time $g(n_0)$.

  If for every $i \in [\ell]$ we have that $\Cc_{i}$ is obtained from $\Cc_{i-1}$ by a $\tau_i$-splittable walk collection $W_i(k_i)$ on vertex set $[n_{i-1}]$
  with $k_i \ge k_0(\eta)$ and $\tau_i \le \tau_0(\eta,k_i)$, where $k_0$ and $\tau_0$ are as in~\cref{theo:list_dec_hammer}, then $\Cc_{\ell}$ is uniquely decodable in time
  $$
  g(n_0) \cdot \prod_{i=1}^{\ell} n_{i-1}^{O(1/\tau_0(\eta,k_i)^4)}.
  $$
\end{lemma}

\begin{proof}
  Induct on $i \in [\ell]$ applying~\cref{lemma:dec_step} as the induction step.
  The code $\Cc_i$ produced during each step will have bias at most $\eta < \eta_0$, so the hypothesis of~\cref{lemma:dec_step} will be met at each level of the cascade.
\end{proof}

We are almost ready to prove our first main theorem establishing decodability close to the Gilbert--Varshamov bound.
We will need parameters for an instantiation of Ta-Shma's code that achieves the desired distance and rate (which will be developed in~\cref{sec:round_1}) and a lemma relating splittability to the spectral properties of the graphs used in the construction (to be proven in~\cref{sec:ta-shma:splittability}).

\begin{lemma}[Ta-Shma Codes I]\label{lemma:ta_shma_construction_I}
 For any $\beta > 0$, there are infinitely many values of $\epsilon \in (0,1/2)$ (with 0 as an accumulation point) such that for infinitely many values of $N \in \N$, there are explicit binary Ta-Shma codes $\Cc_{N,\epsilon,\beta} \subseteq \F_2^N$ with
 \begin{enumerate}[(i)]
   \item distance at least $1/2 - \epsilon/2$ (actually $\epsilon$-balanced), and
   \item rate $\Omega(\epsilon^{2 + \beta})$.
 \end{enumerate}
 Furthermore, $\Cc_{N,\epsilon,\beta}$ is the direct sum lifting of a base code $\Cc_0 \subseteq \F_2^{n_0}$ using the collection of walks $W[0,t-1]$ on the $s$-wide replacement product of two graphs $G$ and $H$, with the following parameters:
 
 \begin{itemize}
	\item $s \geq s_0 \coloneqq \max\{128, 26/\beta\}$.
	\item The inner graph $H$ is a regular graph with $\sigma_2(H) \leq \lambda_2$, where $\lambda_2 = (16s^3 \log s)/s^{2s^2}$.
	\item The outer graph $G$ is a regular graph with $\sigma_2(G) \leq \lambda_1$, where $\lambda_1 = \lambda_2^4/6$.
	\item The base code $\Cc_0$ is unique decodable in time $n_0^{O(1)}$ and has bias $\epsilon_0 \leq \lambda_2^4/3$.
	\item The number of vertices $t$ in the walks satisfies $\lambda_2^{2(1-5/s)(1-1/s)(t-1)} \leq \epsilon$.
 \end{itemize}
\end{lemma}

\begin{lemma}\label{lemma:splittability_preview}
  Let $W(k)$ be either the collection $W[0,s-1]$ of walks of length $s$ on the $s$-wide replacement product with outer graph $G$ and inner graph $H$ or the collection of walks over the vertex set $W[0,r]$, where $r \equiv -1 \pmod s$.
  Then $W(k)$ is $\tau$-splittable with $\tau = \sigma_2(G) + 2 \sigma_2(H) + \sigma_2(H)^2$.
\end{lemma}

The statement of this first decoding theorem is more technical than~\cref{theo:main}, but it will be easier to prove while the latter will build on this theorem with a more careful tuning of parameters.

\begin{theorem}[Main I]\label{theo:main_1}
  For every $\beta > 0$, there are infinitely many values $\epsilon \in (0,1/2)$ (with $0$ an accumulation point) such that
  for infinitely many values of $N \in \mathbb{N}$ there are explicit binary linear Ta-Shma codes
  $\Cc_{N,\epsilon,\beta} \subseteq \mathbb{F}_2^N$ with
  \begin{enumerate}[(i)]
   \item distance at least $1/2 - \epsilon/2$ (actually $\epsilon$-balanced),
   \item rate $\Omega(\epsilon^{2 + \beta})$, and
   \item a unique decoding algorithm with running time $N^{O_{\beta}(\log(\log(1/\epsilon)))}$.
  \end{enumerate}
\end{theorem}

\begin{proof}
    We proceed as follows:
    \begin{enumerate}
      \item Let $\eta_0 = 1/10$ and $\eta = 1/100$ (these choices are arbitrary; we only need $\eta_0 < 1/4$, $\eta < 1/16$, and $\eta < \eta_0$).
      Let $k_0 = k_0(\eta)$ be the constant from~\cref{theo:list_dec_hammer} with this value of $\eta$.
      
      \item Given $\beta > 0$, \cref{lemma:ta_shma_construction_I} provides a value $s_0$ such that the direct sum lifting on the $s$-wide replacement product with $s \ge s_0$ can achieve a rate of $\Omega(\epsilon^{2+\beta})$ for infinitely many $\epsilon \in (0,1/2)$.
      Choose $s$ to be an integer larger than both $k_0$ and $s_0$ that also satisfies
        \begin{equation}
        	s^2 \cdot \left(\frac{s}{16}\right)^{-s^2} \leq \frac{\eta^8}{4K},\label{eq:splittability_s_bound}
        \end{equation}
      where $K$ is the constant from~\cref{theo:list_dec_hammer}.
      
      \item Use~\cref{lemma:ta_shma_construction_I} with this value of $s$ to obtain graphs $G$ and $H$ and a base code $\Cc_0$ having the specified parameters $\lambda_1$, $\lambda_2$, $\epsilon_0$, and $t$, with the additional requirement that $t = s^{\ell}$ for some integer $\ell$.
      These parameter choices ensure that the resulting code $\Cc_{N,\epsilon,\beta}$ has the desired distance and rate.
      Since $s \geq 128$, we have $\lambda_2 = (16s^3 \log s)/s^{2s^2} \leq s^{-s^2}$.
      From the choice of $t$ satisfying $\lambda_2^{2(1-5/s)(1-1/s)(t-1)} \leq \epsilon$, we deduce that $\ell = O(\log(\log(1/\epsilon)))$.
      Note also that the bias $\epsilon_0$ of the code $\Cc_0$ is smaller than $\eta_0$.
      
      \item Create a code cascade with $\ell$ levels using the $s$-wide replacement product of the graphs $G$ and $H$ as in~\cref{sec:code_cascading_ta_shma}, starting with $\Cc_0$ and ending with the final code $\Cc_{\ell} = \Cc_{N,\epsilon,\beta}$.
      As the total number of vertices in a walk is $t = s^{\ell}$, each level of the code cascade will use walks with $s$ vertices.
      Let $\Cc_0,\Cc_1,\dots,\Cc_{\ell}$ be the sequence of codes in this cascade.
            
      \item In order to satisfy the splittability requirement of~\cref{lemma:code_cascading}, the walk collection $W_i(s)$ at each level of the code cascade must be $\tau$-splittable, where $\tau \leq \tau_0(\eta, s^2)$.
      (We use $k=s^2$ instead of $k=s$ in the requirement for a technical reason that will be clear in~\cref{sec:careful_analysis_continum_for_beta}.)
      The bounds on the singular values of $G$ and $H$ and~\cref{lemma:splittability_preview} ensure that
            $$\tau =  \sigma_2(G) + 2 \sigma_2(H) + \sigma_2(H)^2 \leq 4\lambda_2 \leq 4s^{-s^2},$$
      which is smaller than $\tau_0(\eta, s^2) = \eta^8/(K \cdot s^2 \cdot 2^{4s^2})$ by~\cref{eq:splittability_s_bound}
            
      \item As all hypotheses of~\cref{lemma:code_cascading} are satisfied by the code cascade, we apply it to conclude that $\Cc_{N,\epsilon,\beta}$ is uniquely decodable in time
            $$
            g(n_0) \cdot \prod_{i=1}^{\ell} n_{i-1}^{O(1/\tau_0(\eta,s)^4)} \le N^{O(1)} \cdot \prod_{i=1}^{\ell} N^{O_{\beta}(1)} = N^{O_{\beta}(\log(\log(1/\epsilon)))},
            $$
            where we use that $\Cc_0$ is uniquely decodable in time $n_0^{O(1)}$, $1/\tau_0(\eta,s) = 2^{O(1/\beta)}$, $n_{i-1} < n_{\ell} = N$ for every $i \in [\ell]$, and $\ell = O(\log(\log(1/\epsilon)))$.
    \end{enumerate}
\end{proof}

In the code cascade constructed in \cref{theo:main_1}, the final number of vertices in a walk is $t = s^{\ell}$, where $s$ is a sufficiently large constant
that does not depend on $\epsilon$.
The limited choices for $t$ place some restrictions on the values of the final bias $\epsilon$ that can be achieved.
To achieve any bias $\epsilon$ for $\Cc_{\ell}$ we need to choose the parameters more carefully, which will be done in~\cref{sec:careful_analysis_continum_for_beta} to yield our next main result.

\begin{theorem}[Main II]\label{theo:main_2}
  For every $\beta > 0$ and every $\epsilon > 0$ with $\beta$ and $\epsilon$ sufficiently
  small, there are explicit binary linear Ta-Shma codes $\Cc_{N,\epsilon,\beta} \subseteq \mathbb{F}_2^N$
  for infinitely many values $N \in \mathbb{N}$ with
  \begin{enumerate}[(i)]
   \item distance at least $1/2 - \epsilon/2$ (actually $\epsilon$-balanced),
   \item rate $\Omega(\epsilon^{2 + \beta})$, and
   \item a unique decoding algorithm with running time $N^{O_{\beta}(\log(\log(1/\epsilon)))}$.
  \end{enumerate}
\end{theorem}

Ta-Shma obtained codes of rate $\Omega(\epsilon^{2+\beta})$ with
vanishing $\beta$ as $\epsilon$ goes to zero. We obtain a unique
decoding algorithm for this regime (with slightly slower decreasing
$\beta$ as $\epsilon$ vanishes).
More precisely, using the parameters described in~\cref{sec:round_vanishing_esp} and the running time analysis in~\cref{sec:fixed_poly_time}, we obtain the
following theorem which is our main result for unique decoding.

\begin{theorem}[Main Unique Decoding (restatement of~\cref{theo:main})]
  For every $\epsilon > 0$ sufficiently small, there are explicit binary linear Ta-Shma codes
  $\Cc_{N,\epsilon,\beta} \subseteq \mathbb{F}_2^N$ for infinitely many values $N \in \mathbb{N}$ with
  \begin{enumerate}[(i)]
   \item distance at least $1/2 - \epsilon/2$ (actually $\epsilon$-balanced),
   \item rate $\Omega(\epsilon^{2 + \beta})$ where $\beta = O(1/(\log_2(1/\epsilon))^{1/6})$, and
   \item a unique decoding algorithm with running time $N^{O_{\epsilon,\beta}(1)}$.
  \end{enumerate}
  Furthermore, if instead we take  $\beta > 0$ to be an arbitrary constant, the running time
  becomes $(\log(1/\epsilon))^{O(1)} \cdot N^{O_{\beta}(1)}$ (fixed polynomial time).
\end{theorem}

\cref{theo:gentle_list_decoding} about gentle list decoding is
proved in~\cref{sec:round_4} after instantiating Ta-Shma codes in some
parameter regimes in the preceding parts
of~\cref{sec:ta-shma_param_basic}.

\subsection{Fixed Polynomial Time}\label{sec:fixed_poly_time}

In~\cref{theo:main_2}, a running time of $N^{O_{\beta}(\log(\log(1/\epsilon)))}$ was obtained to decode Ta-Shma codes $\Cc_{N,\epsilon,\beta}$ of distance $1/2-\epsilon/2$ and rate $\Omega(\epsilon^{2+\beta})$ for constant $\beta >0$ and block length $N$.
The running time contains an exponent which depends on the bias $\epsilon$ and is therefore not fixed polynomial time.
We show how  to remove this dependence in this regime of $\beta > 0$ being an arbitrary constant.
More precisely, we show the following.

\begin{theorem}[Fixed PolyTime Unique Decoding]\label{theo:main_fixed_poly_time}
  Let $\beta > 0$ be an arbitrary constant. For every $\epsilon > 0$ sufficiently small, there are
  explicit binary linear Ta-Shma codes $\Cc_{N,\epsilon,\beta} \subseteq \mathbb{F}_2^N$ for
  infinitely many values $N \in \mathbb{N}$ with
  \begin{enumerate}[(i)]
   \item distance at least $1/2 - \epsilon/2$ (actually $\epsilon$-balanced),
   \item rate $\Omega(\epsilon^{2 + \beta})$, and
   \item a unique decoding algorithm with fixed polynomial running time $(\log(1/\epsilon))^{O(1)} \cdot N^{O_{\beta}(1)}$.
  \end{enumerate}
\end{theorem}

The list decoding framework finds a list of pairs $(z,y=\lift(z))$ of size at most $N^{(1/\tau_0(\eta,k))^{O(1)}}$ at each level of the code cascade and recursively issues decoding calls to all $z$ in this list.
Since the number of lifts in the cascade is $\Omega(\log(\log(1/\epsilon)))$, we end up with an overall running time of $N^{O_{\beta}(\log(\log(1/\epsilon)))}$.

We will describe a method of pruning these lists which will lead to fixed polynomial running time.
Let $1/2-\sqrt{\eta}$, where $\eta > 0$ is a constant, be the list decoding radius used for a unique decoding step in the cascade.
To achieve fixed polynomial time we will prune this polynomially large list of words to a constant size at each inductive step in~\cref{lemma:code_cascading}.
As we are working with parameters within the Johnson bound, the actual list of codewords has constant size $(1/\eta)^{O(1)}$.
At every step, we will be able to find a small sublist whose size only depends on $\eta$ that has a certain covering property, and then issue decoding calls to this much smaller list.

\begin{definition}[$\zeta$-cover]
  Let $W(k) \subseteq [n]^k$, $\Cc \subseteq \mathbb{F}_2^n$ be a code, $A \subseteq \Cc$, and $\mathcal L = \{(z, \dsum_{W(k)}(z)) \mid z \in A\}$.
  We say that $\mathcal L' = \{(z^{(1)}, \dsum_{W(k)}(z^{(1)})), \dots, (z^{(m)}, \dsum_{W(k)}(z^{(m)}))\}$ is a \emph{$\zeta$-cover} of $\mathcal L$ if for every $(z,y) \in \mathcal L$, there exists some $(z',y') \in \mathcal L'$ with $\bias(z-z') > 1 - 2\zeta$ (that is, either $\Delta(z,z') < \zeta$ or $\Delta(z,z') > 1-\zeta$).
\end{definition}

\begin{lemma}[Cover Compactness]\label{lemma:cover_compactness}
  Let $W(k) \subseteq [n]^k$, $\Cc \subseteq \mathbb{F}_2^n$ be a linear $\eta_0$-balanced code, $\Cc' = \dsum_{W(k)}(\Cc)$ be an $\eta$-balanced code, and $\tilde{y} \in \mathbb{F}_2^{W(k)}$.
  Define
    $$\mathcal{L}(\tilde{y},\Cc,\Cc')\coloneqq \left\{(z,\dsum_{W(k)}(z)) \mid z \in \Cc, \Delta\parens*{\dsum_{W(k)}(z),\tilde{y}} \le \frac{1}{2} - \sqrt{\eta}\right\}.$$
  Suppose $\mathcal{L}'$ is a $\zeta$-cover of $\mathcal{L}(\tilde y, \Cc, \Cc')$ for some $\zeta < 1/2$.
  Further, suppose that for every $(z',y') \in \mathcal{L}'$, we have $\Delta\parens*{y',\tilde{y}} \le 1/2 - \sqrt{\eta}$.
  If $W(k)$ is a $(1-2\zeta, \eta)$-parity sampler, then there exists $\mathcal{L}'' \subseteq \mathcal{L}'$ with $\abs{\mathcal{L}''} \le 1/\eta$ which is a $(2\zeta)$-cover of $\mathcal L$.
\end{lemma}

\begin{proof}
  Build a graph where the vertices are pairs $(z',y') \in \mathcal{L}'$ and two vertices $(z^{(i)},y^{(i)})$, $(z^{(j)},y^{(j)})$ are connected iff $\bias(z^{(i)} - z^{(j)}) > 1-2\zeta$.
  Let $\mathcal{L}''$ be any maximal independent set of this graph.
  Any two vertices $(z^{(i)},y^{(i)}),(z^{(j)},y^{(j)}) \in \mathcal{L}''$ have $\bias(z^{(i)} - z^{(j)}) \leq 1-2\zeta$ and thus $\bias(y^{(i)} - y^{(j)}) \leq \eta$ since $W(k)$ is a $(1-2\zeta,\eta)$-parity sampler.
  This means that $\set{y'' \mid (z'',y'') \in \mathcal{L}''}$ is a code of distance at least $1/2-\eta/2$.
  By the condition that $\Delta(y'',\tilde y) \leq 1/2 - \sqrt{\eta}$ for all $(z'', y'') \in \mathcal L''$ and the Johnson bound, we have $\abs{\mathcal{L}''} \le 1/\eta$.
  
  Finally, we will show that $\mathcal{L}''$ is a $(2\zeta)$-cover of $\mathcal L$.
  Let $(z,y) \in \mathcal L$.
  As $\mathcal{L}'$ is a $\zeta$-cover of $\mathcal L$, there exists a pair $(z',y') \in \mathcal L$ with $\bias(z-z') > 1 - 2\zeta$, so $z$ is within distance $\zeta$ of either $z'$ or its complement $\overline{z'}$.
  The construction of $\mathcal{L}''$ as a maximal independent set ensures that there is some $(z'',y'') \in \mathcal L''$ with $\bias(z'-z'') > 1 - 2\zeta$, so $z''$ is also within distance $\zeta$ of either $z'$ or its complement $\overline{z'}$.
  Applying the triangle inequality in all of the possible cases, we see that either $\Delta(z,z'') < 2\zeta$ or $\Delta(z,z'') > 1-2\zeta$, which implies $\mathcal{L}''$ is a $(2\zeta)$-cover of $\mathcal L$.
\end{proof}

To decode in fixed polynomial time, we use a modification of the list decoding result~\cref{theo:list_dec_hammer} that outputs a $\zeta$-cover $\mathcal{L}'$ of the list of codewords $\mathcal L$ instead of the list itself.
\Cref{theo:list_dec_hammer} recovers the list $\mathcal{L}$ by finding $\mathcal{L}'$ and unique decoding every element of it.
To get $\mathcal{L}'$, we use the same algorithm, but stop before the final decoding step.
This removes the unique decoding time $f(n)$ of the base code from the running time of the list decoding algorithm.
We will apply~\cref{lemma:cover_compactness} after each time we call this $\zeta$-cover algorithm to pare the list down to a constant size before unique decoding; note that this loses a factor of 2 in the strength of the cover.
To compensate for this, we will use a collection $W(k)$ with stronger parity sampling than required for~\cref{theo:list_dec_hammer}.
In that theorem, $W(k)$ was a $(1/2+\eta_0/2, \eta)$-parity sampler to ensure that we obtained words within the list decoding radius $(1/4-\eta_0/4)$ of the base code.
By using a stronger parity sampler, the words in the pruned list $\mathcal{L}''$ will still be within the unique decoding radius even after accounting for the loss in the bias from cover compactness, which means decoding will still be possible at every level of the cascade.
Fortunately, improving the parity sampling only requires increasing the walk length $k$ by a constant multiplicative factor.
The cover retrieval algorithm below will be proven in~\cref{sec:instantiation_list_dec}.

\begin{theorem}[Cover retrieval for direct sum lifting]\label{theo:list_dec_mallet}
  Let $\eta_0 \in (0,1/4)$ be a constant, $\eta \in (0,\eta_0)$, $\zeta = 1/8-\eta_0/8$, and
  	$$k \geq k_0'(\eta) \coloneqq \Theta(\log(1/\eta)).$$
  Suppose $\Cc \subseteq \F_2^n$ is an $\eta_0$-balanced linear code and $\Cc' = \dsum_{W(k)}(\Cc)$ is the direct sum lifting of $\Cc$ on a $\tau$-splittable collection of walks $W(k)$.
  There exists an absolute constant $K > 0$ such that if
  	$$\tau \leq \tau_0(\eta,k) \coloneqq \frac{\eta^{8}}{K \cdot k \cdot 2^{4k}},$$
  then the code $\Cc'$ is $\eta$-balanced, $W(k)$ is a $(1-2\zeta, \eta)$-parity sampler, and we have the following:

  If $\tilde{y}$ is \lict{\sqrt{\eta}} to $\Cc'$, then we can compute a $\zeta$-cover $\mathcal{L}'$ of the list
         $$
         \mathcal{L}(\tilde{y},\Cc,\Cc')\coloneqq \left\{(z,\dsum_{W(k)}(z)) \mid z \in \Cc, \Delta\parens*{\dsum_{W(k)}(z),\tilde{y}} \le \frac{1}{2} - \sqrt{\eta}\right\}
         $$
         in which $\Delta(y',\tilde y) \leq 1/2 - \sqrt{\eta}$ for every $(z',y') \in \mathcal{L}'$, in time
         $$
         n^{O(1/\tau_0(\eta,k)^4)}.
         $$
         Otherwise, we return $\mathcal{L}'=\emptyset$ with the same running time of the preceding case.
\end{theorem}

We now have all of the pieces necessary to prove~\cref{theo:main_fixed_poly_time}.
The process is essentially the same as our earlier unique decoding algorithm, except we use the cover retrieval algorithm from~\cref{theo:list_dec_mallet} instead of the full list decoding from~\cref{theo:list_dec_hammer}.
This allows us to insert a list pruning step in between obtaining the $\zeta$-cover and calling the unique decoding algorithm for the previous level of the cascade.

\begin{proof}[Proof of~\cref{theo:main_fixed_poly_time}]
We use the code $\Cc_{N,\epsilon,\beta}$ from~\cref{theo:main_2} to get the desired distance and rate, with the slight modification of ensuring $s$ is larger than $k_0'$ from~\cref{theo:list_dec_mallet} rather than $k_0$ from~\cref{theo:list_dec_hammer}.

Each level of the code cascade between $\Cc_{i-1}$ and $\Cc_i$ uses constant $\eta_0 < 1/4$ and $\eta < \min\{\eta_0, 1/16\}$, which allows for decoding in a similar fashion to \cref{lemma:dec_step} and~\cref{lemma:code_cascading}.
The difference is that we use~\cref{theo:list_dec_mallet} as the decoding step to obtain a $\zeta$-cover $\mathcal{L}'$ of the list $\mathcal L(\tilde y, \Cc_{i-1}, \Cc_i)$ for $\tilde y \in \F_2^{n_i}$, where $\zeta = 1/8-\eta_0/8$.
By~\cref{lemma:cover_compactness} and the fact that the walk collection is a $(1-2\zeta, \eta)$-parity sampler, $\mathcal L$ has a $(2\zeta)$-cover $\mathcal {L}'' \subseteq \mathcal{L}'$ of size at most $1/\eta$.
The covering property says that for every $(z,y) \in \mathcal L$, there exists some $(z'',y'') \in \mathcal{L}''$ such that $z$ is within distance $2\zeta = 1/4-\eta_0/4$ of either $z''$ or its complement $\overline{z''}$.
This is the unique decoding radius of the $\eta_0$-balanced code $\Cc_{i-1}$, so we can recursively decode the list
	$$\mathcal{L}'' \cup \{(\overline{z''}, \dsum(\overline{z''})) \mid (z'', \dsum(z'')) \in \mathcal{L}''\}$$
to obtain the complete list of codewords in $\Cc_{i-1}$.

Now we analyze the running time.
On each level of the code cascade, we run the cover retrieval algorithm once to get $\mathcal{L}'$, prune the cover to get $\mathcal{L}''$, and then feed the union of $\mathcal{L}''$ and its complement (which has size at most $2/\eta$) into the unique decoding algorithm for the next level of the cascade.
Letting $T_i(n_i)$ be the running time of unique decoding a single word in the code $\Cc_i \subseteq \F_2^{n_i}$, we have the following recurrence:
$$
T_i(n_i) \le n_i^{O(1/\tau_0(\eta,k)^4)} + \frac{2}{\eta} \cdot T_{i-1}(n_{i-1}) \quad \text{ and } \quad T_0(n_0) = n_0^{O(1)}.
$$
Note that the base code $\Cc_0$ has constant bias $\epsilon_0$ and thus it has a fixed polynomial time decoding algorithm (e.g.~\cref{theo:main_2}).
The height of the recursive call tree is the number of levels in the code cascade, which is $\ell = O(\log(\log(1/\epsilon)))$, as in the proof of~\cref{theo:main_1}.
Each node of this tree has a constant branching factor of $2/\eta$.
Thus, the tree has $(\log(1/\epsilon))^{O(1)}$ nodes, each of which costs at most $n_i^{O(1/\tau_0(\eta,k)^4)} \leq N^{O(1/\tau_0(\eta,k)^4)}$ time.
Furthermore, in this regime of $\beta > 0$ being a constant, $k$ is constant as well as $\eta$, so we have $N^{O(1/\tau_0(\eta,k)^4)} = N^{O_{\beta}(1)}$ and the total running time is $(\log(1/\epsilon))^{O(1)} \cdot N^{O_{\beta}(1)}$.
\end{proof}

\section{Satisfying the List Decoding Framework Requirements}\label{sec:satisfying_framework}

The list decoding framework of~\cite{AJQST19} is capable of decoding codes obtained from direct sum liftings, provided they satisfy a few requisite properties.
The framework was originally shown to work for expander walks; we need to adapt it to our case of a code cascade based on walks on the $s$-wide replacement product.
We will start with a broad overview of the list decoding algorithm and point out where various requirements arise.

The problem of finding a list of codewords in a direct sum lifting close to a received word can be viewed as finding approximate solutions to a $k$-XOR instance.
This is done by solving a particular SOS program and rounding the resulting solution.
The algorithm is unable to perform rounding if the $k$-XOR instance is based on an arbitrary collection of walks $W(k)$; it can only handle liftings in which $W(k)$ satisfies a property called \emph{tensoriality}.
If $W(k)$ is tensorial, the SOS local variables in the solution can be approximated by product distributions, which will allow us to obtain a list of solutions by independent rounding.
Tensoriality for expander walks is a consequence of a simpler property known as \emph{splittability}, which is a certain measure of the expansion of a walk collection.

Unfortunately, the list returned by the rounding process will not contain codewords directly---instead, we only get a guarantee that all of the codewords we are looking for have a weak agreement (just over 1/2) with something on this list.
We will find the desired codewords by relying on the parity sampling of $W(k)$.
If $W(k)$ is a sufficiently good parity sampler, weak agreement in the lifted space corresponds to a much stronger agreement in the ground space.
This will allow us to recover the codewords using the unique decoding algorithm of the base code.

To recap, applying the list decoding framework in our setting requires doing the following:
\begin{enumerate}
  \item Proving parity sampling for the walks used in the code cascade (\cref{sec:parity_sampling_walk_over_walk}).
  \item Showing that the walk collection of the $s$-wide replacement product is splittable (\cref{sec:ta-shma:splittability}).
  \item Making Ta-Shma's construction compatible with the Sum-of-Squares machinery (\cref{sec:sos_integration}) and then obtaining tensoriality from splittability (\cref{sec:tensoriality}).
\end{enumerate}

An additional complication is introduced by using a code cascade instead of a single decoding step: the above requirements need to be satisfied at every level of the cascade.
The details of the proofs will often differ between the first level of the cascade, which is constructed using walks on the $s$-wide replacement product, and higher levels, which are walks on a directed graph whose vertices are walks themselves.
Once we have established all of the necessary properties, we will instantiate the list decoding framework in~\cref{sec:instantiation_list_dec}.

We will first define some convenient notation which will be used throughout this section.

\begin{notation}
  Let $G$ be a $d_1$-regular outer graph and $H$ be a $d_2$-regular
  inner graph used in Ta-Shma's $s$-wide replacement product.

  Let $0 \le k_1 \le k_2$ be integers. We define $W[k_1,k_2]$ to
  be the set of all walks starting at time $k_1$ and ending at time $k_2$ in
  Ta-Shma's construction.  More precisely, since $G$ and $H$ are
  regular graphs, the collection $W[k_1,k_2]$ contains all walks
  obtained by sampling a uniform vertex $(v,h) \in V(G) \times V(H)$
  and applying the operator
  $$
  (\matr I \otimes \Aye_H) \matr G_{k_2-1} (I \otimes \Aye_H) \cdots (\matr I \otimes \Aye_H) \matr G_{k_1} (I \otimes \Aye_H),
  $$
  where the index $i$ of each $G_i$ is taken modulo $s$. Observe that
  when $k_1 = k_2$, we have $W[k_1,k_2] = V(G) \times V(H)$.
\end{notation}

We define a family of Markov operators which will play a similar
role to the graphs $\widehat{G}_i$ from the cascade described in~\cref{sec:warmup},
but for Ta-Shma's construction rather than expander walks.

\begin{definition}[Split Operator]\label{def:swap_operator}
  Let $0 \le k_1 \le k_2 < k_3$. We define the graph walk split operator
  $$
  \tswap{k_1}{k_2}{k_3} \colon \mathbb{R}^{W[k_2+1,k_3]} \to \mathbb{R}^{W[k_1,k_2]}
  $$
  such that for every $f \in \mathbb{R}^{W[k_2+1,k_3]}$,
  \begin{align*}
      \left(\tswap{k_1}{k_2}{k_3}(f)\right)(w) \coloneqq \E_{w': ww' \in W[k_1,k_3]}[f(w')],
  \end{align*}
  where $ww'$ denotes the concatenation of the walks
  $w$ and $w'$.  The operator $\tswap{k_1}{k_2}{k_3}$  can be
  defined more concretely in matrix form such that for every $w
  \in W[k_1,k_2]$ and $w' \in W[k_2+1,k_3]$,
  \begin{align*}
      \left(\tswap{k_1}{k_2}{k_3}\right)_{w,w'} = \frac{\One_{ww' \in W[k_1,k_3]}}{|\{\tilde w: w\tilde w \in W[k_1,k_3]\}|}
      	= \frac{\One_{ww' \in W[k_1,k_3]}}{d_2^{2(k_3-k_2)}}.
  \end{align*}
\end{definition}

\subsection{Parity Sampling for the Code Cascade}\label{sec:parity_sampling_walk_over_walk}

To be able to apply the list decoding machinery to the code cascade $\Cc_0 \subseteq \mathbb{F}_2^{n_0}, \Cc_1 \subseteq \mathbb{F}_2^{n_1}, \dots, \Cc_{\ell} \subseteq \mathbb{F}_2^{n_{\ell}}$, we need the direct sum lifting at every level to be a parity sampler.
The first level in the cascade uses walks directly on the $s$-wide replacement product, which we can show is a good parity sampler using the spectral properties proven in~\cref{sec:tweaked_parity_sampling}.
However, it will be more convenient for calculating parameters later on to prove a weaker result, which will suffice for our purposes since we only need to obtain constant bias for every level of the cascade.
We analyze the parity sampling of these walks with the same strategy Ta-Shma employed to show parity sampling for walks on expander graphs (which resulted in \cref{theo:ta_shma_bias_simple}).

\begin{claim}\label{claim:code_cascade_bias_1}
  Let $W[0,s-1]$ be the collection of walks on the $s$-wide replacement product of the graphs $G$ and $H$ and $z \in \mathbb{F}_2^{V(G)}$ be a word with $\bias(z) \le \eta_0$.
  Let $\matr P_z$ be the diagonal matrix with entries $(\matr P_z)_{(v,h),(v,h)} = (-1)^{z_v}$ for $(v,h) \in V(G) \times V(H)$.
  If $\sigma_2((\matr I \otimes \matr A_H) \matr \matr G_i (\matr I \otimes \matr A_H)) \leq \gamma$ for all $0 \leq i \leq s-2$, then
  	$$\norm{\prod_{i=0}^{s-2} (\matr I \otimes \matr A_H) \matr G_i (\matr I \otimes \matr A_H) \matr P_z }_2 \leq (\eta_0 + 2\gamma)^{\lfloor (s-1)/2 \rfloor}.$$
\end{claim}

\begin{proof}
  Let $0 \leq j < s-2$ be even.
  Take a vector $v \in \R^{V(G) \times V(H)}$ with $\norm{v}_2 = 1$ and let $v^{\parallel}$ and $v^{\perp}$ be its parallel and orthogonal components to the all ones vector.
  For $0 \leq i \leq s-2$, let $\matr A_i = (\matr I \otimes \matr A_H) \matr G_i (\matr I \otimes \matr A_H)$.
  Consider two terms $\matr A_{j+1} \matr P_z \matr A_j \matr P_z$ of the product appearing in the claim.
  Since $\matr P_z$ is unitary, $\norm{\matr A_{j+1} \matr P_z \matr A_j \matr P_z}_2 = \norm{\matr A_{j+1} \matr P_z \matr A_j}_2$.
  We have
  \begin{align*}
    \norm{\matr A_{j+1} \matr P_z \matr A_j v}_2 & \le \norm{\matr A_{j+1} \matr P_z \matr A_j v^{\parallel}}_2 + \norm{\matr A_{j+1} \matr P_z \matr A_j v^{\perp}}_2\\
                                         & \le \norm{\matr A_{j+1} \matr P_z \matr A_j v^{\parallel}}_2 + \norm{\matr A_j v^{\perp}}_2\\
                                         & \le \norm{\matr A_{j+1} \matr P_z v^{\parallel}}_2 + \sigma_2(\matr A_j)\\
                                         & \le \norm{\matr A_{j+1} (\matr P_z v^{\parallel})^{\parallel}}_2 + \norm{\matr A_{j+1} (\matr P_z v^{\parallel})^{\perp}}_2 + \sigma_2(\matr A_j)\\
                                         & \le \norm{(\matr P_z v^{\parallel})^{\parallel}}_2 + \sigma_2(\matr A_{j+1}) + \sigma_2(\matr A_j)\\
                                         & \le \eta_0 + 2 \gamma.
  \end{align*}
  Applying this inequality to every two terms of the product, the result follows.
\end{proof}

\begin{corollary}\label{cor:code_cascade_parity_sampling_1}
  Let $W[0,s-1]$ be the collection of walks on the $s$-wide replacement product of the graphs $G$ and $H$ and $\eta_0 > 0$.
  If $\sigma_2((\matr I \otimes \matr A_H) \matr \matr G_i (\matr I \otimes \matr A_H)) \leq \gamma$ for all $0 \leq i \leq s-2$, then $W[0,s-1]$ is an $(\eta_0, \eta)$-parity sampler, where $\eta = (\eta_0 + 2\gamma)^{\lfloor (s-1)/2 \rfloor}$.
\end{corollary}

\begin{proof}
  Let $z \in \F_2^n$ have bias at most $\eta_0$.
  The bias of $\dsum_{W[0,s-1]}(z)$ is given by \footnote{This is slightly different from the expression for the bias given in \cref{sec:tweaks}, but both are equal since moving on the $H$ component of the graph doesn't affect the bit assigned to a vertex.}
  	$$\bias(\dsum_{W[0,s-1]}(z)) = \abs{\ip{\one}{\matr P_z \left(\prod_{i=0}^{s-2} (\matr I \otimes \matr A_H) \matr G_i (\matr I \otimes \matr A_H) \matr P_z \right) \one}},$$
  where $\matr P_z$ is the diagonal matrix with entries $(\matr P_z)_{(v,h),(v,h)} = (-1)^{z_v}$ for $(v,h) \in V(G) \times V(H)$ and $\one$ is the all-ones vector.
  Since $\matr P_z$ is unitary, we have
  	  $$\bias(\dsum_{W[0,s-1]}(z)) \leq \norm{\prod_{i=0}^{s-2} (\matr I \otimes \matr A_H) \matr G_i (\matr I \otimes \matr A_H) \matr P_z}_{2} \leq (\eta_0+2\gamma)^{\lfloor (s-1)/2 \rfloor} = \eta$$
  by \cref{claim:code_cascade_bias_1}.
  Hence $W[0,s-1]$ is an $(\eta_0, \eta)$-parity sampler.
\end{proof}

For higher levels of the cascade, we need to prove parity sampling for collections of walks over walks.
Since the walks on the first level contain $s$ vertices, when we take walks on higher levels, the operator linking different walks together will always use $G_{s-1}$ as the walk operator for the $G$ step.
Thus we can consider a more specific form of the split operator where we split at a time parameter that is one less than a multiple of $s$.

\begin{definition}
  Let $r \equiv -1 \pmod{s}$ be a positive integer. We define the operator $\sswap{r}$
  as
  $$
  \sswap{r} = \tswap{k_1}{k_2}{k_3},
  $$
  where $k_1 = 0$, $k_2 = r$, and $k_3 = 2r +1$. In this case,
  $W[k_1,k_2] = W[k_2+1,k_3]$.
\end{definition}

All levels of the code cascade beyond the first use walks generated by the directed operator $\sswap{r}$.
Proving parity sampling for these walks is analogous to the proof of \cref{cor:code_cascade_parity_sampling_1}, but slightly simpler since the walk operator doesn't change with each step.

\begin{claim}\label{claim:code_cascade_bias_2}
  Let $r \equiv -1 \pmod{s}$ be a positive integer and $z \in \mathbb{F}_2^{W[0,r]}$ be a word with $\bias(z) \le \eta_0$.
  Let $\widetilde{\matr P}_z$ be the diagonal matrix with entries $(\widetilde{\matr P}_z)_{w,w} = (-1)^{z_w}$ for $w \in W[0,r]$.
  For every integer $k \geq 1$, we have
  $$
  \norm{\left(\sswap{r}\widetilde{\matr P}_z\right)^{k-1}}_2  \le \left(\eta_0 + 2 \cdot \sigma_2\left(\sswap{r}\right) \right)^{\lfloor (k-1)/2\rfloor}.
  $$
\end{claim}

\begin{proof}
  Take a vector $v \in \R^{W[0,r]}$ with $\norm{v}_2 = 1$ and let $v^{\parallel}$ and $v^{\perp}$ be its parallel and orthogonal components to the all ones vector.
  Since $\widetilde{\matr P}_z$ is unitary, $\norm{\sswap{r} \widetilde{\matr P}_z \sswap{r} \widetilde{\matr P}_z}_2 = \norm{\sswap{r} \widetilde{\matr P}_z \sswap{r}}_2$.
  We have
  \begin{align*}
    \norm{\sswap{r} \widetilde{\matr P}_z \sswap{r} v}_2 & \le \norm{\sswap{r} \widetilde{\matr P}_z \sswap{r} v^{\parallel}}_2 + \norm{\sswap{r} \widetilde{\matr P}_z \sswap{r} v^{\perp}}_2\\
                                         & \le \norm{\sswap{r} \widetilde{\matr P}_z \sswap{r} v^{\parallel}}_2 + \norm{\sswap{r} v^{\perp}}_2\\
                                         & \le \norm{\sswap{r} \widetilde{\matr P}_z v^{\parallel}}_2 + \sigma_2(\sswap{r})\\
                                         & \le \norm{\sswap{r} (\widetilde{\matr P}_z v^{\parallel})^{\parallel}}_2 + \norm{\sswap{r} (\widetilde{\matr P}_z v^{\parallel})^{\perp}}_2 + \sigma_2(\sswap{r})\\
                                         & \le \norm{(\widetilde{\matr P}_z v^{\parallel})^{\parallel}}_2 + \sigma_2(\sswap{r}) + \sigma_2(\sswap{r})\\
                                         & \le \eta_0 + 2\cdot \sigma_2(\sswap{r}).
  \end{align*}
  As $\norm{(\sswap{r}\widetilde{\matr P}_z)^{k-1}}_2 \le \norm{(\sswap{r}\widetilde{\matr P}_z)^2}^{\lfloor (k-1)/2\rfloor}$, the result follows.
\end{proof}

\begin{corollary}\label{cor:code_cascade_parity_sampling_2}
  Let $r \equiv -1 \pmod s$ be a positive integer and $\eta_0 > 0$.
  The collection of walks $W(k)$ with $k$ vertices over the vertex set $W[0,r]$ using random walk operator $\sswap{r}$ is an $(\eta_0, \eta)$-parity sampler, where $\eta = (\eta_0 + 2 \cdot \sigma_2(\sswap{r}))^{\lfloor (k-1)/2 \rfloor}$.
\end{corollary}

\begin{proof}
  Let $z \in \F_2^{W[0,r]}$ have bias at most $\eta_0$.
  The bias of the direct sum lifting of $z$ is given by
  	$$\bias(\dsum_{W(k)}(z)) = \abs{\ip{\one}{\widetilde{\matr P}_z (\sswap{r} \widetilde{\matr P}_z)^{k-1} \one}},$$
  where $\widetilde{\matr P}_z$ is the diagonal matrix with entries $(\widetilde{\matr P}_z)_{w,w} = (-1)^{z_w}$ for $w \in W[0,r]$ and $\one$ is the all-ones vector.
  Since $\widetilde{\matr P}_z$ is unitary, we have
  	$$\abs{\ip{\one}{\widetilde{\matr P}_z (\sswap{r} \widetilde{\matr P}_z)^{k-1} \one}} \leq \norm{\left(\sswap{r} \widetilde{\matr P}_z \right)^{k-1}}_2 \leq \left(\eta_0 + 2 \cdot \sigma_2\left(\sswap{r}\right) \right)^{\lfloor (k-1)/2\rfloor} = \eta$$
  by~\cref{claim:code_cascade_bias_2}.
  Hence $W(k)$ is an $(\eta_0, \eta)$-parity sampler.
\end{proof}

\subsection{Splittability of Ta-Shma's Construction}\label{sec:ta-shma:splittability}

We investigate the splittability of the collection of walks generated by Ta-Shma's construction.
In order to formally define this property, we will need the concept of an interval splitting tree, which describes how a walk is split into smaller and smaller pieces.

\begin{definition}[Interval Splitting Tree]
  We say that a binary rooted tree $\tree$ is a $k$-interval splitting tree if it has
  exactly $k$ leaves and
  \begin{itemize}
    \item the root of $\tree$ is labeled with $(0,m,k-1)$ for some $m \in \set{0,1,\dots,k-2}$, and
    \item each non-leaf non-root vertex $v$ of $\tree$ is labeled with $(k_1,k_2,k_3)$ for some
          integer $k_2 \in [k_1,k_3-1]$. Suppose $(k_1',k_2',k_3')$ is the label assigned to the parent of $v$. 
          If $v$ is a left child, we must have $k_1=k_1'$ and $k_3=k_2'$; otherwise, we must
          have $k_1=k_2'+1$ and $k_3=k_3'$.
  \end{itemize}
\end{definition}

Given an interval splitting tree $\tree$, we can naturally associate a split operator $\tswap{k_1}{k_2}{k_3}$ to each internal node $(k_1,k_2,k_3)$.
The splittability of a collection $W[0,k-1]$ of $k$-tuples is a notion of expansion at every node in the splitting tree.

\begin{definition}[$(\tree, \tau)$-splittability]\label{def:splittability}
  The collection $W[0,k-1]$ is said to be
  $(\tree, \tau)$-splittable if $\tree$ is a $k$-interval splitting
  tree and
  $$
  \sigma_2(\tswap{k_1}{k_2}{k_3}) \le \tau
  $$
  for every internal node $(k_1,k_2,k_3)$ of $\tree$.

  If there exists some $k$-interval splitting tree $\tree$ such that
  $W[0,k-1]$ is $(\tree, \tau)$-splittable, then $W[0,k-1]$ will be
  called $\tau$-splittable.
\end{definition}

In order to prove that the collection of walks in Ta-Shma's construction is splittable, a split operator $\tswap{k_1}{k_2}{k_3}$ can be related to the walk operator $(I \otimes \Aye_H) G_{k_2} (I \otimes \Aye_H)$ as shown below.
This structural property will allow us to deduce spectral properties of $\tswap{k_1}{k_2}{k_3}$ from the spectrum of $(I \otimes \Aye_H) G_{k_2} (I \otimes \Aye_H)$.

\begin{lemma}\label{lemma:swap_matrix_rep}
  Let $0 \leq k_1 \leq k_2 < k_3$.
  Suppose $G$ is a $d_1$-regular outer graph on vertex set $[n]$ with walk operator $G_{k_2}$ used at step $k_2$ of a walk on the $s$-wide replacement product and $H$ is a $d_2$-regular inner graph on vertex set $[m]$ with normalized random walk operator $\Aye_H$.
  Then there are orderings of the rows and columns of the representations of $\tswap{k_1}{k_2}{k_3}$ and $\Aye_H$ as matrices such that
  $$
  \tswap{k_1}{k_2}{k_3} = \left( (I \otimes \Aye_H) G_{k_2} (I \otimes \Aye_H) \right) \otimes \Jay/d_2^{2(k_3-k_2-1)},
  $$
  where $\Jay \in \mathbb{R}^{[d_2]^{2(k_2-k_1)} \times [d_2]^{2(k_3-k_2-1)}}$ is the all ones matrix.
\end{lemma}

\begin{proof}
  Partition the set of walks $W[k_1,k_2]$ into the sets
  $W_{1,1}, \dots, W_{n,m}$, where $w \in W_{i,j}$ if the last vertex
  of the walk $w_{k_2} = (v_{k_2},h_{k_2})$ satisfies $v_{k_2} = i$
  and $h_{k_2} = j$. Similarly, partition $W[k_2+1,k_3]$ into the sets
  $W_{1,1}', \dots, W_{n,m}'$, where $w' \in W_{i,j}'$ if the first
  vertex of the walk $w'_1 = (v_1,h_1)$ satisfies $v_1 = i$ and $h_1 =
  j$.  Note that $\abs{W_{i,j}} = d_2^{2(k_2-k_1)}$ and $\abs{W_{i,j}'} =
  d_2^{2(k_3-k_2-1)}$ for all $(i,j) \in [n]\times[m]$, since there are $d_2^2$ choices for each step of the walk.

  Now order the rows of the matrix $\tswap{k_1}{k_2}{k_3}$ so that all
  of the rows corresponding to walks in $W_{1,1}$ appear first,
  followed by those for walks in $W_{1,2}$, and so on in
  lexicographic order of the indices $(i,j)$ of $W_{i,j}$, with an
  arbitrary order within each set.  Do a similar re-ordering of the
  columns for the sets $W_{1,1}', \dots, W_{1,m}'$.  Observe that
  \begin{align*}
  \left(\tswap{k_1}{k_2}{k_3}\right)_{w,w'} &= \frac{\One_{ww' \in W[k_1,k_3]}}{d_2^{2(k_3-k_2)}} \\
                                          & = \frac{d_2^2 \cdot (\text{weight of  transition from } (v_{k_2},h_{k_2}) \text{ to } (v_1', h'_1)  \text{ in } (I \otimes \Aye_H)G_{k_2} (I \otimes \Aye_H))}{d_2^{2(k_3-k_2)}},
  \end{align*}
  which only depends on the adjacency of the last vertex of $w$ and
  the first vertex of $w'$.  If the vertices $w_{k_2} = (v_{k_2},h_{k_2})$ and $w_1'=(v_1,h_1)$
  are adjacent, then
  $$
  \left(\tswap{k_1}{k_2}{k_3}\right)_{w,w'} = \left((I \otimes \Aye_H)G_{k_2}(I \otimes \Aye_H)\right)_{(v_{k_2},h_{k_2}),(v_1', h'_1)}/d_2^{2(k_3-k_2-1)},
  $$
  for every $w \in W_{w_{k_2}}$ and $w' \in W_{w_{k_1}}'$; otherwise,
  $\left(\tswap{k_1}{k_2}{k_3}\right)_{w,w'} = 0$.  Since the walks in
  the rows and columns are sorted according to their last and first
  vertices, respectively, the matrix $\tswap{k_1}{k_2}{k_3}$ exactly
  matches the tensor product
  $((I \otimes \Aye_H)G_{k_2} (I \otimes \Aye_H)) \otimes \Jay/d_2^{2(k_3-k_2-1)}$.
\end{proof}

\begin{corollary}\label{cor:split_walk_spectral_gap}
  Let $0 \leq k_1 \leq k_2 < k_3$.
  Suppose $G$ is a $d_1$-regular outer graph with walk operator $G_{k_2}$ used at step $k_2$ of a walk on the $s$-wide replacement product and $H$ is a $d_2$-regular inner graph with normalized random walk operator $\Aye_H$.
  Then
  $$
  \sigma_2(\tswap{k_1}{k_2}{k_3}) = \sigma_2((I \otimes \Aye_H) G_{k_2} (I \otimes \Aye_H)).
  $$
\end{corollary}

\begin{proof}
  Using~\cref{lemma:swap_matrix_rep} and the fact that
  $$
  \sigma_2(((I \otimes \Aye_H)G_{k_2} (I \otimes \Aye_H)) \otimes \Jay/d_2^{2(k_3-k_2-1)})  = \sigma_2((I \otimes \Aye_H)G_{k_2} (I \otimes \Aye_H)),
  $$
  the result follows.
\end{proof}

\begin{remark}
\Cref{cor:split_walk_spectral_gap} is what causes the splittability argument to break down for Ta-Shma's original construction, as $\sigma_2(\matr G_{k_2} (\matr I \otimes \matr A_H)) = 1$.
\end{remark}

By combining this result with the spectral bound from \cref{fact:zig_zag_bound}, we find that the collection of walks of length $s$ on the $s$-wide replacement product is $(\mathcal T, \tau)$-splittable for any splitting tree $\mathcal T$, where $\tau$ is controlled by the second singular values of the graphs $G$ and $H$.
This analysis can also be applied to walks on higher levels of the cascade where the vertex set is $W[0,r]$.

\begin{corollary}[Restatement of~\cref{lemma:splittability_preview}] \label{cor:splittability}
  The collection of walks $W[0,s-1]$ on the $s$-wide replacement product with outer graph $G$ and inner graph $H$ and the collection of walks $W(k)$ on the vertex set $W[0,r]$ with random walk operator $\sswap{r}$ and $r \equiv -1 \pmod s$ are both $\tau$-splittable with $\tau = \sigma_2(G) + 2 \sigma_2(H) + \sigma_2(H)^2$.
\end{corollary}

\begin{proof}
  By~\cref{cor:split_walk_spectral_gap} and~\cref{fact:zig_zag_bound}, the split operator $\tswap{k_1}{k_2}{k_3}$ for any $0 \leq k_1 \leq k_2 < k_3$ satisfies
  	$$\sigma_2(\tswap{k_1}{k_2}{k_3}) = \sigma_2((I \otimes \Aye_H) G_{k_2} (I \otimes \Aye_H)) \leq \sigma_2(G) + 2 \sigma_2(H) + \sigma_2(H)^2,$$
  so $W[0,s-1]$ is $\tau$-splittable with $\tau = \sigma_2(G) + 2 \sigma_2(H) + \sigma_2(H)^2$, as any internal node $(k_1,k_2,k_3)$ of any $s$-interval splitting tree will have $\sigma_2(\tswap{k_1}{k_2}{k_3}) \leq \tau$.
  The split operators of any $k$-interval splitting tree for the collection $W(k)$ are of the form $\tswap{k_1}{k_2}{k_3}$ with $k_1 \equiv 0 \pmod s$ and $k_2, k_3 \equiv -1 \pmod s$, which means $W(k)$ is $\tau$-splittable as well.
\end{proof}

\subsection{Integration with Sum-of-Squares}\label{sec:sos_integration}

\newcommand{\V}[2]{{v}_{(#1,#2)}}
\newcommand{\W}[2]{{w}_{(#1,#2)}}
\newcommand{\Vempty}{\V{\emptyset}{\emptyset}}

Before defining tensoriality and obtaining it in our setting, we examine how the Sum-of-Squares hierarchy is used in the list decoding algorithm in more detail.

\subsubsection{SOS Preliminaries: $p$-local PSD Ensembles}

The SOS hierarchy gives a sequence of increasingly tight semidefinite programming relaxations for several optimization problems, including CSPs.
Since we will use relatively few facts about the SOS hierarchy, already developed in the analysis of Barak, Raghavendra and Steurer \cite{BarakRS11}, we will adapt their notation of \emph{$p$-local distributions} to describe the relaxations.

Solutions to a semidefinite relaxation of a CSP on $n$ boolean variables using $p$ levels of the SOS hierarchy induce probability distributions $\mu_S$ over $\F_2^S$ for any set $S \subseteq [n]$ with $\abs{S} \leq p$.
These distributions are consistent on intersections: for $T \subseteq S \subseteq [n]$, we have $\mu_{S|T} = \mu_T$, where $\mu_{S|T}$ denotes the restriction of the distribution $\mu_S$ to the set $T$.
We use these distributions to define a collection of random variables $\rv Z_1, \ldots, \rv Z_n$ taking values in $\F_2$ such that for any set $S$ with $\abs{S} \leq p$, the collection of variables $\inbraces{\rv Z_i}_{i \in S}$ has joint distribution $\mu_S$.
Note that the entire collection $\{\rv Z_1, \ldots, \rv Z_n\}$ \emph{may not} have a joint distribution: this property is only true for sub-collections of size at most $p$.
We will refer to the collection $\{\rv Z_1, \ldots, \rv Z_n\}$ as a \emph{$p$-local ensemble} of random variables.

For any $T \subseteq [n]$ with $\abs{T} \leq p-2$ and any $\xi \in \F_2^T$, we can define a $(p-\abs{T})$-local ensemble $\{\rv Z_1', \ldots, \rv Z_n'\}$ by ``conditioning'' the local distributions on the event $\rv Z_T = \xi$, where $\rv Z_T$ is shorthand for the collection $\inbraces{\rv Z_i}_{i \in T}$.
For any $S$ with $\abs{S} \leq p-\abs{T}$, we define the distribution of $\rv Z_S'$ as $\mu_S' := \mu_{S \cup T} | \set{\rv Z_T = \xi}$.

Finally, the semidefinite program also ensures that for any such
conditioning, the conditional covariance matrix
	$$\Emm_{(S_1, \alpha_1)(S_2,\alpha_2)} ~=~ \cov\inparen{\One_{[\rv Z_{S_1}' = \alpha_1]}, \One_{[\rv Z_{S_2}' = \alpha_2]}}$$
is positive semidefinite, where $\abs{S_1}, \abs{S_2} \leq (p-\abs{T})/2$.  Here, for each pair $S_1, S_2$ the covariance is
computed using the joint distribution $\mu_{S_1 \cup S_2}'$.
In this paper, we will only consider $p$-local ensembles such that for every conditioning on a set of size at most $(p-2)$, the conditional covariance matrix is PSD. We will refer to these as \emph{$p$-local PSD ensembles}.
We will also need a simple corollary of the above definitions.

\begin{fact}\label{fact:set-ensemble}
  Let $\{\rv Z_1, \ldots, \rv Z_n\}$ be a $p$-local PSD ensemble and $W(k) \subseteq [n]^k$
  For $1 \leq i < k$, define $W(i) \subseteq [n]^i$ to be the collection of tuples of size $i$ appearing in elements of $W(k)$.
  For all $p' \leq p/2$, the collection $\inbraces{\rv Z_{\operatorname{set}(w)}}_{w \in W(\leq p')}$ is a $(p/p')$-local PSD ensemble, where $W(\leq p') = \bigcup_{i = 1}^{p'} W(i)$.
\end{fact}

For random variables $\rv Z_S$ in a $p$-local PSD ensemble, we use the notation $\inbraces{\rv Z_S}$ to denote the distribution of $\rv Z_S$ (which exists when $\abs{S} \leq p$).
As we will work with ordered tuples of variables instead of sets, we define $\rv Z_w$ for $w \in [n]^k$ based on the set $S_w = \operatorname{set}(w)$, taking care that repeated elements of $w$ are always assigned the same value.

\begin{definition}[Plausible assignment]
  Given $w = (w_1,\dots,w_k) \in [n]^k$ and an assignment $\alpha \in \F_2^w$, we say that $\alpha$ is plausible for $w$ if there are no distinct $i,j \in [k]$ such that $w_i = w_j$ but $\alpha_i \ne \alpha_j$.
\end{definition}

The distribution $\{\rv Z_w\} = \mu_w$ is defined as $\mu_w(\alpha) = \mu_{S_w}(\alpha|_{S_w})$ if $\alpha \in \F_2^w$ is plausible for $w$, and $\mu_w(\alpha) = 0$ otherwise.

\subsubsection{Tensoriality}
A key algorithm in the list decoding framework is propagation rounding (\cref{algo:prop-rd}), which solves a CSP to find solutions close to a codeword.
Suppose $W(k) \subseteq [n]^k$ is a collection of walks, or more generally, a collection of any $k$-tuples.
The algorithm starts with a local PSD ensemble $\set{\rv Z_1, \ldots, \rv Z_n}$ which is the solution to an SOS program for list decoding.
Propagation rounding takes this solution and conditions some of the variables according to a random assignment to these variables to yield another local PSD ensemble $\rv Z'$.

\begin{algorithm}{Propagation Rounding Algorithm, adapted from~\cite{AJQST19}}{An $(L+2k)$-local PSD ensemble $\set{\rv Z_1, \ldots, \rv Z_n}$ and collection $W(k) \subseteq [n]^k$.}{A random assignment $(\assn_1, \ldots, \assn_n) \in  \F_2^n$ and $2k$-local PSD ensemble $\rv Z'$.}\label{algo:prop-rd}
    \begin{enumerate}
        \item Choose $m \in \set*{1, \ldots, L/k}$ uniformly at random.
        \item For $j = 1, \dots, m$, sample a walk $w_j$ independently and uniformly from $W(k)$.
        \item Write $S = \bigcup_{j = 1}^m \textup{set}(w_j)$ for the set of the seed vertices.
        \item Sample an assignment $\assn: S \to \F_2$ according to the local distribution $\set{\rv Z_{S}}$.
        \item Set $\rv Z' = \set{\rv Z_1, \ldots, \rv Z_n | \rv Z_S = \assn}$, i.e.~the local ensemble
            $\rv Z$ conditioned on agreeing with $\assn$.
        \item For all $i \in [n]$, sample independently $\assn_i \sim \set{\rv Z'_i}$.
        \item Output $(\assn_1, \ldots, \assn_n)$ and $\rv Z'$.
    \end{enumerate}
\end{algorithm}

If the collection $W(k) \subseteq [n]^k$ used in the direct sum lifting is amenable to SOS rounding, the conditioned ensemble $\rv Z'$ will be able to recover a word close to some codeword on the list.
This is quantified by the following \textit{tensorial} properties.
We will see shortly how splittability will be used to obtain tensoriality in our setting.

\begin{definition}[Tensorial Walk Collection]
  Let $W(k) \subseteq [n]^k$, $\mu \in [0,1]$, and $L \in \mathbb{N}$.
  Define $\Omega$ to be the set of all tuples $(m, S, \sigma)$ obtainable in propagation rounding (\cref{algo:prop-rd}) on $W(k)$ with SOS degree parameter $L$.
  We say that $W(k)$ is $(\mu,L)$-tensorial if the local PSD ensemble $\rv Z'$ returned by propagation rounding satisfies
  \begin{equation}
    \ExpOp_{\Omega} \ExpOp_{w \in W(k)}{ \norm{\set{\rv Z_{w}'} - \set*{\rv Z_{w(1)}'}\cdots \set*{\rv Z_{w(k)}'}}_1} \le \mu.
  \end{equation}
\end{definition}

The framework actually uses a strengthening of the above property, in which variables for pairs of walks chosen independently approximately behave as a product.

\begin{definition}[Two-Step Tensorial Walk Collection]\label{def:two_step_tensorial}
  Let $W(k) \subseteq [n]^k$, $\mu \in [0,1]$, and $L \in \mathbb{N}$.
  Define $\Omega$ to be the set of all tuples $(m, S, \sigma)$ obtainable in propagation rounding (\cref{algo:prop-rd}) on $W(k)$ with SOS degree parameter $L$.
  We say that $W(k)$ is $(\mu,L)$-two-step tensorial if it is $(\mu,L)$-tensorial and the local PSD ensemble $\rv Z'$ returned by propagation rounding satisfies the additional condition
    $$\ExpOp_{\Omega} \ExpOp_{w,w' \in W(k)}{ \norm{\set{\rv Z_{w}' \rv Z_{w'}'} - \set*{\rv Z_{w}'}\set*{\rv Z_{w'}'}}_1 } \le \mu.$$
\end{definition}

\subsubsection{From Directed to Undirected}

In order to apply the list decoding framework using the directed split operator $\tswap{k_1}{k_2}{k_3}$, we will replace it with the symmetrized version
  $$
  \mathcal{U}(\tswap{k_1}{k_2}{k_3}) = 
  \begin{pmatrix}
    0 & \tswap{k_1}{k_2}{k_3}\\
    \left(\tswap{k_1}{k_2}{k_3}\right)^{\dag} & 0
  \end{pmatrix}
  $$
and show how $\mathcal{U}(\tswap{k_1}{k_2}{k_3})$ corresponds to a particular undirected graph.

\begin{definition}
  Let $0 \le k_1 \le k_2 < k_3$. We define the operator
  $\fswap{k_2}{k_3}{k_1} \colon \mathbb{R}^{W[k_1,k_2]} \to \mathbb{R}^{W[k_2+1,k_3]}$ such that for every $f \in \R^{W[k_1,k_2]}$,
  $$
  \left(\fswap{k_2}{k_3}{k_1} (f)\right)(w') \coloneqq \E_{w: ww' \in W[k_1,k_3]} [f(w)],
  $$
  for every $w' \in W[k_2+1,k_3]$.
\end{definition}

The operator $\mathcal{U}(\tswap{k_1}{k_2}{k_3})$ defines an undirected weighted bipartite graph on the vertices $W[k_1,k_2] \cup W[k_2+1,k_3]$.
We can see that $\fswap{k_2}{k_3}{k_1}$ is the adjoint of $\tswap{k_1}{k_2}{k_3}$, which means that each edge $ww'$ in this graph is weighted according to the transition probability from one walk to the other whenever one of $w$, $w'$ is in $W[k_1,k_2]$ and the other is in $W[k_2+1,k_3]$.

\begin{claim}
  $$
  \left(\tswap{k_1}{k_2}{k_3}\right)^{\dag} = \fswap{k_2}{k_3}{k_1}.
  $$
\end{claim}

\begin{proof}
  Let $f \in C^{W[k_1,k_2]}$ and $g \in C^{W[k_2+1,k_3]}$.
  For $i \leq j$, define $\Pi_{i,j}$ to be the uniform distribution on $W[i,j]$.
  We show that $\ip{f}{\tswap{k_1}{k_2}{k_3} g} = \ip{\fswap{k_2}{k_3}{k_1} f}{g}$.
  On one hand we have
  \begin{align*}
    \ip{f}{\tswap{k_1}{k_2}{k_3} g} &= \E_{w \in W[k_1,k_2]} \left[f(w) \E_{w': ww' \in W[k_1,k_3]}  [g(w')] \right] \\
                                               &= \E_{w \in W[k_1,k_2]} \left[f(w) \sum_{w' \in W[k_2+1,k_3]} \frac{\Pi_{k_1,k_3}(ww')}{\Pi_{k_1,k_2}(w)} g(w') \right] \\
                                               &= \sum_{w \in W[k_1,k_2]} \Pi_{k_1,k_2}(w) f(w) \sum_{w' \in W[k_2+1,k_3]} \frac{\Pi_{k_1,k_3}(ww')}{\Pi_{k_1,k_2}(w)}  g(w')\\
                                               & = \sum_{ww' \in W[k_1,k_3]} f(w) g(w') \Pi_{k_1,k_3}(ww').
  \end{align*}
  On the other hand we have
  \begin{align*}
  \ip{\fswap{k_2}{k_3}{k_1} f}{ g} &= \E_{w' \in W[k_2+1,k_3]} \left[\E_{w: ww' \in W[k_1,k_3]} [f(w)] g(w') \right] \\
                                               &= \E_{w' \in W[k_2+1,k_3]} \left[ \sum_{w \in W[k_1,k_2]} \frac{\Pi_{k_1,k_3}(ww')}{\Pi_{k_2+1,k_3}(w')} f(w) g(w') \right]\\
                                               &= \sum_{w' \in W[k_2+1,k_3]} \Pi_{k_2+1,k_3}(w') \sum_{w \in W[k_1,k_2]} \frac{\Pi_{k_1,k_3}(ww')}{\Pi_{k_2+1,k_3}(w')} f(w) g(w')\\
                                               &= \sum_{ww' \in W[k_1,k_3]} f(w) g(w') \Pi_{k_1,k_3}(ww').
  \end{align*}
  Hence, $\fswap{k_2}{k_3}{k_1} = (\tswap{k_1}{k_2}{k_3})^{\dag}$ as claimed.
\end{proof}

\subsubsection{Variables for Walks on the $s$-wide Replacement Product}

When analyzing walks on the $s$-wide replacement product, we actually need to use two separate, but related, local PSD ensembles.
In Ta-Shma's construction, the vertices of the outer graph $G$ correspond to positions in the base code $\Cc_0 \subseteq \F_2^n$, where $n = \abs{V(G)}$.
Given a vertex $(v,h) \in V(G) \times V(H)$ in the $s$-wide replacement product and codeword $z \in \Cc_0$, $(v,h)$ is assigned bit $z_v$, regardless of the vertex $h$ of the inner graph.
We will enforce this property by working with variables in $V(G)$ rather than the full $V(G) \times V(H)$.
The local PSD ensemble $\rv Z = \{\rv Z_v\}_{v \in V(G)}$ contains one variable for every vertex of $G$, with local distributions for sets of variables up to a given size.
For a walk $w$ on the $s$-wide replacement product, we will use $\rv Z_w$ as an abbreviation for $\rv Z_{S_w}$, where $S_w$ is the set of all $G$-components of vertices visited on the walk.

The constraints of the CSP are placed on walks on the $s$-wide replacement product that do care about the $H$-component of the vertices, so we define a second local PSD ensemble $\rv Y = \{\rv Y_{(v,h)}\}_{(v,h) \in V(G) \times V(H)}$ with a variable for each vertex of the $s$-wide replacement product of $G$ and $H$.
It is this collection $\rv Y$ for which we need to prove tensoriality in order to use the list decoding framework.
When we perform propagation rounding, we condition the ensemble $\rv Z$ on a random assignment $\sigma$ to a subset $S \subseteq V(G)$, rather than conditioning $\rv Y$ on a random assignment to a subset of $V(G) \times V(H)$.
Working with $\rv Z$ ensures that the rounded assignments will be consistent on each cloud of the $s$-wide replacement product.
Since the bit assigned to a vertex $(v,h)$ only depends on $v$, independent rounding of $\{\rv Z \mid \rv Z_S = \sigma\}$ will also yield the desired rounding of $\{\rv Y \mid \rv Z_S = \sigma\}$.

We can define $\rv Y$ based on the ensemble $\rv Z$ more concretely.
Suppose $S' \subseteq V(G) \times V(H)$ is a subset of size at most $p$, where $p$ is the locality of the ensemble, and define $T = \{v \mid (v,h) \in S'\}$.
The distribution $\mu_{S'}$ of $\rv Y_{S'}$ is defined based on the distribution $\mu_T$ of $\rv Z_T$ by $\mu_{S'}(\alpha) = \mu_T(\alpha|_T)$, where $\alpha \in \F_2^{S'}$ is an assignment to $S'$ whose value on each vertex $(v,h)$ only depends on $v$.

Observe that the introduction of the ensemble $\rv Y$ is only necessary on the first level of the Ta-Shma code cascade between the codes $\Cc_0$ and $\Cc_1$, which takes place on the $s$-wide replacement product.
Higher levels of the cascade use walks on graphs whose vertices are the walks from the level below.
The association of the bits of a codeword to the vertices of this graph has no consistency requirement, so we simply use a single local ensemble $\rv Z$ with a variable for each vertex.

\subsection{Splittability Implies Tensoriality}\label{sec:tensoriality}

The connection between splittability and tensoriality will be made with the help of a version of the triangle inequality.

\begin{claim}[Triangle inequality, adapted from~\cite{AJQST19}]\label{claim:glorified_triangle_ineq}
  Let $s \in \mathbb{N}^+$ and $\tree$ be an $s$-interval splitting tree. Then
  $$
   \ExpOp_{w \in W[0,s-1]}{ \norm{\set{\rv Z_{w}} - \prod_{i=0}^{s-1}\set*{\rv Z_{w(i)}}}_1} \le \sum_{(k_1,k_2,k_3)\in \tree} ~\ExpOp_{w \in W[k_1,k_3]}{ \norm{\set{\rv Z_{w}} - \set*{\rv Z_{w(k_1,k_2)}}\set*{\rv Z_{w(k_2+1,k_3)}} }_1},
  $$
  where the sum is taken over the labels of the internal nodes of $\tree$.
\end{claim}

To prove tensoriality, we will use the method of~\cite{BarakRS11} and~\cite{AJT19} to show that we can break correlations over expanding collections of tuples arising in the $s$-wide replacement product of the form
$$
\ExpOp_{\substack{ww' \in W[k_1,k_3] \\ w \in W[k_1,k_2], w' \in  W[k_2+1,k_3]}}{\norm{\set{\rv Z_{ww'}} - \set{\rv Z_w} \set{\rv Z_{w'}}}_1}
$$
appearing on the right-hand side of the triangle inequality.

\subsubsection{The First Level of the Cascade}\label{sec:tensorial_first_s_steps}

We now check the technical details to obtain tensoriality for the first lifting in the code cascade between the codes $\Cc_0$ and $\Cc_1$, which corresponds to taking $s$ steps in Ta-Shma's construction.
Recall that in order to obtain an assignment $z' \in \mathbb{F}_2^n$ whose lifting is consistent on vertices with the same $G$-component, we need to prove tensoriality for the ensemble $\rv Y$ with a variable for each vertex in $V(G) \times V(H)$.

The proof of tensoriality will make use of a specific entropic potential function.
For an arbitrary random variable $\rv X$ taking values in a finite set $[q]$, define the function $\mathcal{H}(\rv X)$ as
$$
\mathcal{H}(\rv X) ~\coloneqq~ \frac{1}{q} \sum_{a \in [q]} \textup{H}(\One_{[\rv X = a]}) ~=~ \E_{a \in [q]} \textup{H}(\One_{[\rv X = a]}),
$$
where $\textup{H}$ is the binary entropy function.
Using this, we define a potential function for a weighted undirected graph $G$.

\begin{definition}[Graph Potential]
Let $G=(V,E)$ be a weighted graph with edge distribution $\Pi_E$.  Let
$\Pi_V$ be the marginal distribution on $V$. Suppose that $\set{\rv
Y_i}_{i\in V}$ is a $p$-local PSD ensemble for some $p \ge 1$. We
define $\Phi^G$ to be
$$
\Phi^G ~\coloneqq~ \Ex{i \sim \Pi_V}{\mathcal{H}(\rv Y_i)}.
$$
\end{definition}

Let $\tree$ be an $s$-interval splitting tree associated with the
$s$-wide replacement product of graphs $G$ and $H$. We define
$$
\Phi^{\tree} \coloneqq \sum_{(k_1,k_2,k_3) \in \tree} \Phi^{\mathcal{U}(\tswap{k_1}{k_2}{k_3})},
$$
where $\mathcal{U}(\tswap{k_1}{k_2}{k_3})$ is the associated bipartite undirected graph
of the operator $\tswap{k_1}{k_2}{k_3}$.

\begin{lemma}[Splittability Implies Tensoriality]\label{lemma:tensorial}
  Let $W[0,s-1]$ be the walk collection of the $s$-wide replacement product of two graphs $G$ and $H$.
  If $L \ge 128 \cdot (s^4 \cdot 2^{4s}/\mu^4)$ and $W[0,s-1]$ is $\tau$-splittable with $\tau \le \mu/(4s \cdot 2^{4s})$, then $W[0,s-1]$ is $(\mu,L)$-tensorial.
\end{lemma}

\begin{proof}
  We need to show that
  $$
  \ExpOp_{w \in W[0,s-1]}{ \norm{\set{\rv Y_{w}'} - \prod_{i=0}^{s-1} \set*{\rv Y_{w(i)}'}}_1 } \le \mu,
  $$
  which can be proven by adapting a
  potential argument technique from~\cite{BarakRS11}. First, set the potential
    \begin{equation}
        \Phi_m = \ExpOp_{S \sim \Pi_m}{\ExpOp_{\sigma \sim \set{\rv Z_S}}{ \Phi^{\tree}_{\mid \rv Z_S = \sigma }}},\label{eq:pot-def}
    \end{equation}
  where the distribution $\Pi_m$ on $S \subseteq V(G)$ is obtained from the process of choosing $S$ in propagation rounding (\cref{algo:prop-rd}) once $m$ has been fixed.
  Consider the error term
    \begin{equation}
        \mu_m \coloneqq \ExpOp_{S \sim \Pi_m}{\ExpOp_{\sigma \sim \set{\rv Z_S}}{D(S,\sigma)}},\label{eq:assume-large}
    \end{equation}
    where $D(S,\sigma) \coloneqq \ExpOp_{w \in W[0,s-1]} \norm{\set{\rv Y_{w} \mid \rv Z_S = \sigma} - \prod_{i=0}^{s-1} \set*{\rv Y_{w(i)} \mid \rv Z_S = \sigma} }_1$.
    If $\mu_m \ge \mu/2$, then
    \[ \ProbOp_{S \sim \Pi_m, \sigma \sim \{\rv Z_S\}}\left[ D(S,\sigma) \ge \mu_m/2 \right] \ge \frac{\mu}{4}. \]

    For each choice of $S$ and $\sigma$ such that $D(S,\sigma) \ge \mu/2$, applying the triangle inequality from~\cref{claim:glorified_triangle_ineq} to the conditioned variables gives us
    \begin{align*}
    \frac{\mu}{2} &\le \ExpOp_{w \in W[0,s-1]}{ \norm{\set{\rv Y_{w} \mid \rv Z_S = \sigma} - \prod_{i=0}^{s-1} \set*{\rv Y_{w(i)} \mid \rv Z_S = \sigma}}_1}\\
                  &\le \sum_{(k_1,k_2,k_3)\in \tree} ~\ExpOp_{w \in W[k_1,k_3]}{ \norm{\set{\rv Y_{w} \mid \rv Z_S = \sigma} - \set*{\rv Y_{w(k_1,k_2)} \mid \rv Z_S = \sigma}\set*{\rv Y_{w(k_2+1,k_3)} \mid \rv Z_S = \sigma} }_1}.
    \end{align*}
    Hence, there exists $(k_1,k_2,k_3)$ such that
    $$
    \frac{\mu}{2s} \le \ExpOp_{w \in W[k_1,k_3]}{ \norm{\set{\rv Y_{w} \mid \rv Z_S = \sigma} - \set*{\rv Y_{w(k_1,k_2)} \mid \rv Z_S = \sigma}\set*{\rv Y_{w(k_2+1,k_3)} \mid \rv Z_S = \sigma} }_1}.
    $$
    Note that choosing $w \in W[0,s-1]$ uniformly and restricting to $w(k_1,k_3)$ gives a uniformly random element of $W[k_1,k_3]$.
    If we choose $w(k_1,k_2)$ or $w(k_2+1,k_3)$ with equal probability, then the final walk is distributed according
    to the stationary measure of $\mathcal{U}(\tswap{k_1}{k_2}{k_3})$. Let $w'$ denote the chosen walk. Observe that $\rv Y_{w'}$
    is a deterministic function of $\rv Z_{w'} \mid \rv Z_S = \sigma$. Now, we sample $\rv Z_{w'} \mid \rv Z_S =\sigma$,
    which gives us a sample of $\rv Y_{w'}$. Applying~\cref{lemma:progress_lemma}, we have
    $$
    \Phi^{\mathcal{U}(\tswap{k_1}{k_2}{k_3})}_{\vert\set{ \rv Y_{w'} \vert \rv Z_S =\sigma}} \le \Phi^{\mathcal{U}(\tswap{k_1}{k_2}{k_3})}_{\rv Z_S =\sigma} - \frac{\mu^2}{16s^2 \cdot 2^{4s}}.
    $$
    This conditioning on an assignment to $\rv Z_{\textup{set}(w')} \mid \rv Z_S = \sigma$ does not
    increase the other terms of $\Phi^{\tree}$ associated to split operators other than $\mathcal{U}(\tswap{k_1}{k_2}{k_3})$ since entropy is
    non-increasing under conditioning. Similarly, conditioning on the
    remaining variables that are part of $w$ but not $w'$
    does not increase $\Phi^{\tree}$. Then, we obtain
    \[
    \Phi_m -\Phi_{m + 1} \ge \ProbOp_{S \sim \Pi_m, \sigma \sim \{\rv Z_S\}}\left[D(S,\sigma) \ge \mu_m/2 \right] \cdot \frac{\mu^2}{16s^2 \cdot 2^{4s}}.
    \]
    Since $s \ge \Phi_1 \ge \cdots \ge \Phi_{L/(s+1)} \ge 0$, there can be
    at most $32 s^3 \cdot 2^{4s}/\mu^3$ indices $m \in [L/s]$ such that
    $\mu_m \ge \mu/2$.  In particular, since the total number of indices is $L/s$, we have
    \[ \ExpOp_{m \in [L/s]}[\mu_m] \le \frac{\mu}{2} + \frac{s}{L} \cdot \frac{32s^3 \cdot 2^{4s}}{\mu^3}. \]
    Our choice of $L$ is more than enough to ensure $\Ex{m \in [L/s]}{\mu_m} \le \mu$.
\end{proof}

Applying the list decoding framework will require the stronger property of two-step tensoriality, which we can obtain under the same assumptions.

\begin{lemma}[Splittability Implies Two-step Tensoriality]\label{lemma:two_step_tensorial}
  Let $W[0,s-1]$ be the walk collection of the $s$-wide replacement product of two graphs $G$ and $H$.
  If $L \ge 128 \cdot (s^4 \cdot 2^{4s}/\mu^4)$ and $W[0,s-1]$ is $\tau$-splittable with $\tau \le \mu/(4s \cdot 2^{4s})$, then $W[0,s-1]$ is $(\mu,L)$-two-step tensorial.
\end{lemma}

\begin{proof}
  Under our assumptions the $(\mu,L)$-tensorial property follows
  from~\cref{lemma:tensorial} (which is the only place where the
  assumption on $\tau$ is used), so we only need to show
  $$
  \ExpOp_{w,w' \in W[0,s-1]}{ \norm{\set{\rv Y_{w}' \rv Y_{w'}'} - \set*{\rv Y_{w}'}\set*{\rv Y_{w'}'}}_1 } \le \mu,
  $$
  which can be proven by adapting a
  potential argument technique from~\cite{BarakRS11}. First, set the potential
    \begin{equation}
        \Phi_m = \ExpOp_{S \sim \Pi_m}{\ExpOp_{\sigma \sim \set{\rv Z_S}}{\ExpOp_{w \in W[0,s-1]}{\mathcal{H}(\rv Y_w \mid \rv Z_S = \sigma)}}},\label{eq:pot-def}
    \end{equation}
  where the distribution $\Pi_m$ on $S \subseteq V(G)$ is obtained from the process of choosing $S$ in propagation rounding (\cref{algo:prop-rd}) once $m$ has been fixed.
  Consider the error term
    \begin{equation}
        \mu_m \coloneqq \ExpOp_{S \sim \Pi_m}{\ExpOp_{\sigma \sim \set{\rv Z_S}}{D(S,\sigma)}},\label{eq:assume-large}
    \end{equation}
    where $D(S,\sigma) \coloneqq \Ex{w,w' \in W[0,s-1]}{\norm{\set{\rv Y_{w}\rv Y_{w'} \mid \rv Z_S = \sigma} - \set{\rv Y_{w} | \rv Z_S = \sigma}\set{\rv Y_{w'} | \rv Z_S = \sigma} }_1}$.
    If $\mu_m \ge \mu/2$, then
    \[ \ProbOp_{S \sim \Pi_m, \sigma \sim \{\rv Z_S\}}\left[ D(S,\sigma) \ge \mu_m/2 \right] \ge \frac{\mu}{4}. \] 
    
    Let $G' = (V=W[0,s-1],E)$ be the graph with edges $E = \{\{w,w'\} \mid w,w' \in W[0,s-1]\}$.
    Local correlation (expectation over the edges) on this
    graph $G'$ is the same as global correlation (expectation over two independent
    copies of vertices). Then, we obtain~\footnote{See \cite{AJT19} or~\cite{BarakRS11} for the details.}
    \[
    \Phi_m -\Phi_{m + 1} \ge \ProbOp_{S \sim \Pi_m, \sigma \sim \{\rv Z_S\}}\left[D(S,\sigma) \ge \mu_m/2 \right] \cdot \frac{\mu^2}{2 \cdot 2^{2s}}.
    \]
    Since $1 \ge \Phi_1 \ge \cdots \ge \Phi_{L/(s+1)} \ge 0$, there can be
    at most $8 \cdot 2^{2s} /\mu^3$ indices $m \in [L/s]$ such that
    $\mu_m \ge \mu/2$.  In particular, since the total number of indices is $L/s$, we have
    \[ \ExpOp_{m \in [L/s]}{\mu_m } \le \frac{\mu}{2} + \frac{k}{L} \cdot \frac{8 \cdot 2^{2s}}{\mu^3}. \]
    Our choice of $L$ is more than enough to ensure $\Ex{m \in [L/s]}{\mu_m} \le \mu$.
\end{proof}

We have already established that $W[0,s-1]$ is $\tau$-splittable with $\tau = \sigma_2(G) + 2\sigma_2(H) + \sigma_2(H)^2$ in~\cref{cor:splittability}, so we can obtain $(\mu, L)$-two-step tensoriality for any $\mu$ if this quantity is small enough.

\subsubsection{Higher Levels of the Cascade}

We now discuss tensoriality of the other levels of the cascade between $\Cc_{i-1}$ and $\Cc_{i}$ for $i \ge 2$.
Tensorial properties are simpler to establish here than on the first level of the cascade.
The relevant split operators are special cases of $\tswap{k_1}{k_2}{k_3}$ where $k_1 \equiv 0 \pmod s$ and $k_2, k_3 \equiv -1 \pmod s$.
The main difference now is that we can associate the parity bits of $\Cc_{i-1}$ with the vertices of $\mathcal{U}(\sswap{r})$, which themselves represent walks.
As this association of parity bits doesn't need to satisfy a consistency condition, we only need to work with a single ensemble $\rv Z$ instead of working with two different ensembles as in the previous case.
The proofs of~\cref{lemma:tensorial} and~\cref{lemma:two_step_tensorial} with these slight modifications give us two-step tensoriality.

\begin{lemma}[Two-step Tensoriality for Higher Levels]\label{lemma:two_step_tensorial_subsequent_steps}
  Let $W(k)$ be the set of walks defined using $(k-1)$ steps of the operator $\sswap{r}$.
  If $L \ge 128 \cdot (k^4 \cdot 2^{4k}/\mu^4)$ and $W(k)$ is $\tau$-splittable with $\tau \le \mu/(4k \cdot 2^{4k})$, then $W(k)$ is $(\mu,L)$-two-step tensorial.
\end{lemma}

We know that the collection of walks obtained from $\sigma_2(\sswap{r})$ is $(\sigma_2(G) + 2 \cdot \sigma_2(H) + \sigma_2(H)^2)$-splittable, so the parameters necessary to obtain two-step tensoriality are the same as in the first level of the cascade.

\section{Choosing Parameters for Ta-Shma's Construction}\label{sec:ta-shma_param_basic}

We explore how some choices of parameters for Ta-Shma's
construction interact with the requirements of our decoding
algorithm. The analysis is divided into rounds of increasingly
stronger decoding guarantees with later rounds relying on the codes
obtained in previous rounds. Naturally, the stronger
guarantees come with more delicate and technical considerations. We
briefly summarize the goals of each round and some key parameters.

\begin{enumerate}
  \item Round I: For any constant $\beta > 0$, we obtain \emph{efficient unique decodable} codes $\Cc_{\ell}$ with distance at least $1/2-\epsilon$ and
                 rate $\Omega(\epsilon^{2+\beta})$ for infinitely many \emph{discrete} values of $\epsilon > 0$ (with $\epsilon$ as close to $0$ as desired).
                 In this regime, it suffices for the expansion of $H$ to be constant. This round leads to~\cref{theo:main_1}. 
  \item Round II: Similar to Round I, but now $\epsilon$ can be any value in an interval $(0,b)$ with $b < 1/2$ being a function of $\beta$. Again
                  the expansion of $H$ can be constant. This round leads to~\cref{theo:main_2}.
  \item Round III: We want $\beta$ to vanish as $\epsilon$ vanishes (this is qualitatively similar to Ta-Shma's result). In this regime, we make the
                   expansion of $H$ be a function of $\epsilon$, and we rely on the uniquely decodable codes of Round II. This round leads to~\cref{theo:main}.
  \item Round IV: For any constant $\beta_0 > 0$, we obtain \emph{efficient list decodable} codes $\Cc_{\ell}$ with list decoding radius
                  $1/2 - \beta_0$ and rate $\Omega(\epsilon^{2+\beta})$ with $\beta \to 0$ as $\epsilon \to 0$. In this regime, we make
                  the expansion of $H$ be a function of $\epsilon$, and we rely on the uniquely decodable codes of Round III. This round leads to~\cref{theo:gentle_list_decoding}.
\end{enumerate}

The way we choose parameters for Ta-Shma's construction
borrows heavily from Ta-Shma's arguments in~\cite{Ta-Shma17}. We fix
some notation common to all rounds. A graph is said to be an
$(n,d,\lambda)$-graph provided it has $n$ vertices, is $d$-regular, and
has second largest singular value of its normalized adjacency matrix
at most $\lambda$.
\begin{notation}
  We use the following notation for the graphs $G$ and $H$ used in
  the $s$-wide replacement product.
  \begin{itemize}
    \item The outer graph $G$ will be an $(n',d_1,\lambda_1)$-graph.
    \item The inner graph $H$ will be a $(d_1^s,d_2,\lambda_2)$-graph.
  \end{itemize}
  The parameters $n', d_1,d_2, \lambda_1,\lambda_2$ and $s$ will be chosen in the subsequent sections.
\end{notation}

\subsection{Round I: Initial Analysis}\label{sec:round_1}

We are given the dimension $D$ of the desired code and $\epsilon \in (0,1/2)$.
We set a parameter $\alpha \le 1/128$ such that (for convenience)
$1/\alpha$ is a power of $2$ and
\begin{equation}\label{eq:alpha_epsilon}
\frac{\alpha^5}{4 \log_2(1/\alpha)} \ge \frac{1}{\log_2(1/\epsilon)}.
\end{equation}
We can assume that $\alpha \le 1/128$ satisfy~\cref{eq:alpha_epsilon}
since otherwise $\epsilon$ is a constant and we can use the list
decodable codes from~\cite{AJQST19}. The use
of~\cref{eq:alpha_epsilon} will be clear shortly. It becomes a
necessity from round III onward. For rounds I and II, the parameter
$\alpha$ will be a constant, but it will be useful to establish the
analysis in more generality now so that subsequent rounds can reuse
it.

\noindent \textbf{The inner graph $H$.} \enspace The choice of $H$ is
similar to Ta-Shma's choice. More precisely, we set $s=1/\alpha$ and
$d_2 = s^{4s^2}$ (Ta-Shma took $d_2 = s^{4s}$).
We obtain a Cayley graph $H
= \textup{Cay}(\mathbb{F}_2^{4s\log_2(d_2)}, A)$ such that $H$ is an
$(n_2=d_2^{4s}, d_2, \lambda_2)$ graph where
$\lambda_2 = b_2/\sqrt{d_2}$ and $b_2 = 4 s \log_2(d_2)$.
(The set of generators, $A$, comes from a small bias code derived from
a construction of Alon et al.~\cite{AGHP92}, but we will rely on
Ta-Shma's analysis embodied in~\cref{lemma:aghp} and not discuss it
further.)

\noindent \textbf{The base code $\Cc_0$.} \enspace
Set $\epsilon_0 = 1/d_2^2 = \lambda_2^4/b_2^4 \le \lambda_2^4/3$ (this
choice differs from Ta-Shma's and it appears because we are
essentially working with $H^2$ rather than $H$). We will choose a base code $\Cc_0$ such that the desired code will be obtained as a direct sum lifting of $\Cc_0$, and because this lifting preserves the dimension, the dimension of $\Cc_0$ should be $D$. We choose $\Cc_0$ to be an
$\epsilon_0$-balanced code with dimension $D$ and block length $n=O_{\epsilon_0}(D)$.
For instance, we can start with
any good (constant rate and relative distance) linear base code
$\Cc_0$ that has an efficient unique decoding algorithm and obtain a
$\epsilon_0$-balanced lifted code that can be efficiently unique decoded (as
long as $\epsilon_0$ is constant) using the framework in~\cite{AJQST19}.

\noindent \textbf{The outer graph $G$.} \enspace
Set $d_1 = d_2^4$ so that $n_2 = d_1^s$ as required by the
$s$-wide replacement product. We apply Ta-Shma's explicit Ramanujan
graph~\cref{lemma:explicit_ramanujan} with parameters $n$, $d_1$ and
$\theta$ to obtain an $(n',d_1,\lambda_1)$ Ramanujan graph $G$ with
$\lambda_1 \le 2\sqrt{2}/\sqrt{d_1}$ and $n' \in [(1-\theta)n,n]$ or $n'\in [(1-\theta)2n,2n]$. Here, $\theta$ is an error parameter that we set as $\theta = \lambda_2^4/6$ (this choice of $\theta$ differs from
Ta-Shma). Because we can construct words with block length $2n$
(if needed) by duplicating each codeword, we may assume w.l.o.g. that $n'$ is close to $n$ and
$(n-n') \le \theta n \le 2 \theta n'$. See \cref{app:explicit_structures} for a more formal description of this graph.

Note that
$\lambda_1 \le \lambda_2^4/6$ since $\lambda_1 \le 3/\sqrt{d_1} =
3/d_2^2 = 3 \cdot \lambda_2^4/b_2^4 \le \lambda_2^4/6$. Hence,
$\epsilon_0 + 2\theta + 2\lambda_1 \le \lambda_2^4$.

\noindent \textbf{The walk length.} \enspace
Set the walk length $t-1$ to be the smallest integer such that
$$
(\lambda_2^2)^{(1-5\alpha)(1-\alpha)(t-1)} \le \epsilon.
$$
This will imply using Ta-Shma's analysis that the bias of the final code is at most $\epsilon$ as shown later.

\begin{center}
\fbox{\begin{minipage}{30em}

$s=1/\alpha, \text{ such that } \frac{\alpha^5}{4 \log_2(1/\alpha)} \ge \frac{1}{\log_2(1/\epsilon)}$
\vskip 0.3cm
$H: (n_2,d_2,\lambda_2),\quad n_2=d_1^s,\quad d_2=s^{4s^2},\quad \lambda_2=\frac{b_2}{\sqrt{d_2}},\quad b_2 = 4s\log d_2$
\vskip 0.3cm
$G: (n',d_1,\lambda_1),\quad n' \approx n = O(D/\epsilon_0^c),\quad d_1=d_2^4,\quad \lambda_1\leq \frac{2\sqrt{2}}{d_1}$
\vskip 0.3cm
$t: \text{ smallest integer such that } (\lambda_2^2)^{(1-5\alpha)(1-\alpha)(t-1)} \leq \epsilon$

\end{minipage}}
\end{center}

\begin{claim}\label{claim:t_lower_bound}
  We have $t-1 \ge s/\alpha = s^2$.
\end{claim}

\begin{proof}
  Using $d_2 = s^{4s^2}$ and~\cref{eq:alpha_epsilon}, we have
  \begin{align*}
  \left(\frac{1}{\lambda_2^2}\right)^{(1-5\alpha)(1-\alpha)s/\alpha} & \le \left(\frac{1}{\lambda_2^2}\right)^{s/\alpha} = \left(\frac{d_2}{b_2^2}\right)^{s/\alpha} \le \left(d_2\right)^{s/\alpha} = s^{4s^3/\alpha}\\
                                                              & = 2^{4s^3 \log_2(s)/\alpha} = 2^{4 \log_2(1/\alpha)/\alpha^4} \le  2^{\log_2(1/\epsilon)} = \frac{1}{\epsilon}.
  \end{align*}
  Hence, $\epsilon \le (\lambda_2^2)^{(1-5\alpha)(1-\alpha)s/\alpha}$ and thus
  $t-1$ must be at least $s/\alpha$.
\end{proof}

\begin{remark}\label{remark:choice_of_t}
 By our choice of $t$, we have $(\lambda_2^2)^{(1-5\alpha)(1-\alpha)(t-2)} \ge \epsilon$.
 Since $1/(t-1) \le \alpha$, we get $(\lambda_2^2)^{(1-5\alpha)(1-\alpha)^2 (t-1)} \ge \epsilon$.
\end{remark}

\noindent \textbf{Final Bias.} We denote by $\Cc_{\ell}$ the final code obtained
by $t$ steps of the $s$-wide replacement product. The bias of
$\Cc_{\ell}$ is given by~\cref{cor:tweaked_ta-shma_spectral_analysis}
(which in turn is a simple corollary of
Ta-Shma's~\cref{fact:ta-shma_main}) as shown next.
\begin{corollary}
  The code $\Cc_{\ell}$ is $\epsilon$-balanced.
\end{corollary}

\begin{proof}
  Using~\cref{cor:tweaked_ta-shma_spectral_analysis}, we have that the final bias  
  \begin{align*}
    b \coloneqq \left(\sigma_2(H^2)^{s-1} + (s-1) \cdot \sigma_2(H^2)^{s-2} + (s-1)^2 \cdot \sigma_2(H^2)^{s-4}\right)^{\lfloor (t-1)/s \rfloor}
  \end{align*}
  is bounded by
  \begin{align*}
    b &\le (3(s-1)^2\sigma_2(H^2)^{s-4})^{((t-1)/s)-1} && (\text{Using } \sigma_2(H^2)\leq 1/3s^2) \\
      &\le ((\sigma_2(H^2)^{s-5})^{(t-1-s)/s} \\
      &= \sigma_2(H^2)^{(1-5/s)(1-s/(t-1))(t-1)}\\
      &\le \sigma_2(H^2)^{(1-5\alpha)(1-\alpha)(t-1)}\\
      &= \left(\lambda_2^2\right)^{(1-5\alpha)(1-\alpha)(t-1)} \le \epsilon,                                                                                                                          
  \end{align*}
  where the last inequality follows from $s = 1/\alpha$ and $t-1 \ge s/\alpha$, the latter from~\cref{claim:t_lower_bound}.
\end{proof}

\noindent \textbf{Rate.}
The proof of the rate follows a similar structure of Ta-Shma's
original argument except that we take $s$ to be a constant independent
of $\epsilon$ so that $\epsilon_0$, $\lambda_1$, and $\lambda_2$ are also
constants independent of $\epsilon$. Note that we previously said $\alpha = 1/s$ needs to satisfy \autoref{eq:alpha_epsilon}, but that implies only an upper bound for $s$, and smaller (even constant) values for $s$ are still permissible.
\begin{claim}\label{claim:rate_round_i}
  $\Cc_{\ell}$ has rate $\Omega(\epsilon^{2+ 26 \cdot \alpha})$ provided $\epsilon_0 > 0$ is constant.
\end{claim}

\begin{proof}
  The support size is the number of walks of length $t$ on the
  $s$-wide replacement product of $G$ and $H$ (each step of the walk has $d_2^2$ options),
  which is
  \begin{align*}
    |V(G)| |V(H)| d_2^{2(t-1)} = n' \cdot d_1^s \cdot d_2^{2(t-1)} &= n' \cdot  d_2^{2(t-1)+4s} \le n \cdot  d_2^{2(t-1)+4s}\\
                                                      &= \Theta_{\epsilon_0}\left(D \cdot d_2^{2(t-1)+4s}\right)\\
                                                      &= \Theta\left(D \cdot (d_2^2)^{t-1+2s}\right)\\
                                                      &= O\left(D \cdot (d_2^2)^{(1+2\alpha)(t-1)}\right),                                                      
  \end{align*}
  where the penultimate equality follows from the assumption that $\epsilon_0$ is a constant.

  Note that $d_2^{\alpha} = d_2^{1/s} = s^{4s} \ge b_2$ since
  % $b_2 = 4\sqrt{2} \log_2(d_2) s = 16\sqrt{2} s^2 \le s^4$
  $b_2 = 4 s \log_2(d_2) = 16 s^3 \log_2(s) \le s^4$
  (recall that $s=1/\alpha \ge 128$). Thus,
  $$
  d_2^{1-2\alpha} = \frac{d_2}{d_2^{2\alpha}} \le \frac{d_2}{b_2^2} = \frac{1}{\sigma_2(H^2)}.
  $$
  We obtain
  \begin{align*}
  (d_2^2)^{(t-1)} &\le \left(\frac{1}{\sigma_2(H^2)}\right)^{\frac{2(t-1)}{1-2\alpha}} \\
  &\le \left(\frac{1}{\epsilon}\right)^{\frac{2}{(1-2\alpha)(1-5\alpha)(1-\alpha)^2}} && \text{(Using~\cref{remark:choice_of_t})}\\
  &\le \left(\frac{1}{\epsilon}\right)^{2(1+10\alpha)},
  \end{align*}
  which implies a block length of
  $$
  O\left(D \cdot (d_2^2)^{(1+2\alpha)(t-1)}\right) = O\left(D \left(\frac{1}{\epsilon}\right)^{2(1+10\alpha)(1+2\alpha)}\right) = O\left(D \left(\frac{1}{\epsilon}\right)^{2(1+13\alpha)}\right).
  $$
\end{proof}

\begin{lemma}[Codes Near the GV bound I]\label{lemma:codes_near_gv_i}
  For every constant $\beta > 0$, there exists a sufficiently large constant $s$ in the above
  analysis so that for any dimension value $D \in \mathbb{N}^+$ (sufficiently large) and $\epsilon > 0$
  (sufficiently small) the final code $\Cc_{N,\epsilon,\beta}$, where $N$ is the block length,
  satisfies
  \begin{itemize}
    \item $\Cc_{N,\epsilon,\beta}$ is $\epsilon$-balanced, 
    \item $\Cc_{N,\epsilon,\beta}$ has rate $\Omega(\epsilon^{2+\beta})$, and
    \item $\Cc_{N,\epsilon,\beta}$ is a linear code of dimension $D$.
  \end{itemize}
\end{lemma}

\begin{remark}
  As a consequence of code cascading, the final attainable walk lengths have the form $s^{\ell}-1$
  where $\ell$ is a positive integer. Given $\beta > 0$, we have infinitely many values of $\epsilon$
  attainable by such walk lengths which gives infinitely many codes $\Cc_{N,\epsilon,\beta}$. This means
  that although the bias $\epsilon$ cannot be arbitrary, we have an infinite sequence of values of $\epsilon$
  for which the rates of the codes $\Cc_{N,\epsilon,\beta}$ are near the GV bound. In~\cref{sec:careful_analysis_continum_for_beta}, we show how to bypass this artificial limitation. These codes are used in the proof of~\cref{theo:main_1}.
\end{remark}

We can view the above analysis as defining a function $\Gamma$ that
receives
\begin{itemize}
  \item the dimension $D \in \mathbb{N}^+$,
  \item the final bias $\epsilon > 0$,
  \item the approximating error $\alpha \in (0,1/128]$ with $s\coloneqq1/\alpha$ being a power of two, and
  \item a multiplying factor $Q \in \mathbb{N}^+$ such that $d_2 = s^{4s^2 \cdot Q}$ (in the above $Q$ was $1$).
\end{itemize}
and outputs a tuple of parameters $(t,\epsilon_0, \theta,d_1,\lambda_1,
n')$, graphs $G$ and $H$ (as above) where, in particular, the number
of steps $t \in \mathbb{N}^+$ is such that the final code $\Cc_{\ell}$
has bias at most $\epsilon$ and rate $\Omega(\epsilon^{2+
26 \cdot \alpha})$.

In future rounds, $\Gamma$ may be called with $Q=s$ instead of $Q=1$. This will cause $d_2$ to increase from $s^{4s^2}$ to $s^{4s^2\cdot Q}$, and so in the proof of \autoref{claim:t_lower_bound}, $2^{4\log_2(1/\alpha) / \alpha^4}$ will be replaced by $2^{4\log_2(1/\alpha) / \alpha^5}$. This explains why \cref{eq:alpha_epsilon} has a stricter requirement than needed in the $Q=1$ case above.

\subsection{Round II: A More Careful Analysis}\label{sec:careful_analysis_continum_for_beta}

We are given the dimension of the code $D$ and $\epsilon \in (0,1/2)$.
As before, we set a parameter $\alpha \le 1/128$ such that (for convenience)
$1/\alpha$ is a power of $2$. Set $s = 1/\alpha$ and $Q = s$.

Apply $\Gamma$ to $(D,\epsilon,\alpha,Q)$ to obtain all parameters except $t$. Choose $t$ to be the smallest integer satisfying
$$
(\lambda_2^2)^{(1-5\alpha)(1-2\alpha)(1-\alpha)(t-1)} \le \epsilon,
$$
where observe that an extra $(1-2\alpha)$ factor appears in the exponent. This change in $t$ will worsen the rate but by losing a factor of $\frac{1}{1-2\alpha}$ in the exponent, we can lower bound the rate. That is, $(d_2^2)^{-(t-1)} = \Omega(\eps^{\frac{2+26\cdot \alpha}{1-2\alpha}})$. 

Set $\ell \in \mathbb{N}^+$ to be the smallest value such that
$s^{\ell} \ge t$ (here we are implicitly assuming that $t >
s$). If $s^{\ell} = t$, we are done since we can use all the
parameters returned by $\Gamma$ for the construction of
$\Cc_{\ell}$. Now assume $s^{\ell} > t$ and let $\zeta =
t/s^{\ell-1}$. Note that $\zeta \in (1,s)$. Choose $P$ to be the
integer in the interval $[Q,s\cdot Q]$ such that
$$
0 \le \frac{P}{Q} - \zeta \le \frac{1}{Q}.
$$

Because $s^\ell > t$, and only powers of $s$ may be chosen for walk length, we might overshoot in walk length by a multiplicative factor of $s$. This will cause a corresponding decay in rate computation that we cannot afford. To overcome this, in the last level of the cascade between codes $\Cc_{\ell-1}$ and
$\Cc_{\ell}$, perform the direct sum over walks of length $(P-1)$ instead
of length $(s-1)$. The new total number of vertices is $t' =
P s^{\ell-1}$. Note that $P$ can be as large as $s^2$, so our
splittability guarantee of $W(P)$ (the walk collection from the lift
between $\Cc_{\ell-1}$ and $\Cc_{\ell}$) has to be strong enough to
accommodate this larger arity and not only arity $s$.

\begin{claim}\label{claim:approx_walk_length}
  We have $t-1 \le \frac{t'-1}{Q} \le (1+2\alpha) (t-1)$.
\end{claim}

\begin{proof}
 By construction, we have the sequence of implications
 \begin{align*}
    & 0 \le \frac{P}{Q}s^{\ell-1} - \zeta s^{\ell-1} \le \frac{s^{\ell-1}}{Q} \\
    \implies & 0 \le \frac{t'}{Q} - t \le \frac{s^{\ell-1}}{Q} \le \frac{t}{Q} \\ 
    \implies & t - \frac{1}{Q} \le \frac{t'-1}{Q} \le (t-1)\left(1+\frac{1}{Q} \right)+1,
 \end{align*}
 from which we obtain
 	$$t-1 \le t - \frac{1}{Q} \le \frac{t'-1}{Q}$$
 and
 	$$\frac{t'-1}{Q} \le (t-1) \left(1+\frac{1}{Q}\right) +1 = (1+\alpha) (t-1) + 1 < (1+2\alpha)(t-1),$$
 the latter using $Q = s = 1/\alpha$.
\end{proof}
 
We apply $\Gamma$ again but this time to $(D,\epsilon,\alpha,1)$ to
obtain new parameters $(t'', \epsilon_0'$, $\theta'$, $d_1'$,
$\lambda_1'$, $n'')$, and graphs $G'$ and $H'$.

\begin{claim}
  The code $\Cc_{\ell}'$ obtained by $t'$ walk steps on the
  $s$-wide replacement product of $G'$ and $H'$ from the second
  application of $\Gamma$ has bias at most $\epsilon$ and rate
  $\Omega(\epsilon^{2+40\alpha})$.
\end{claim}

\begin{proof}
  Let $d_2 = s^{4s^2 \cdot Q}$, $b_2 = 4 s \log_2(d_2)$ and $\lambda_2 = b_2/\sqrt{d_2}$ be
  the parameters of the first invocation of $\Gamma$. Recall that $t$ was chosen to
  be the smallest integer satisfying
  $$
  (\lambda_2^2)^{(1-5\alpha)^2(1-\alpha)(t-1)} \le \epsilon.
  $$
  Let $d_2' = s^{4s^2}$, $b_2' = 4 s \log_2(d_2')$ and $\lambda_2' = b_2'/\sqrt{d_2'}$
  be the parameters of the second invocation of $\Gamma$. Observe that
  \begin{align*}
  (\lambda_2')^{Q} &= \frac{(b_2')^Q}{\sqrt{(d_2')^Q}} = \frac{(b_2'^Q)}{\sqrt{d_2}} = \frac{(16 s^3 \log_2(s))^Q}{s^{2s^2 \cdot Q}} \\
  &\le \frac{s^{4Q}}{s^{2s^2\cdot Q}} = \frac{1}{s^{2s^2\cdot Q(1 - \frac{2}{s^2})}} = \left( \frac{1}{s^{2s^2\cdot Q}}\right)^{1-2\alpha} \leq \left( \frac{b_2}{\sqrt{d_2}}\right)^{1-2\alpha} = \lambda_2^{1-2\alpha}.
  \end{align*}
  Then the bias of $\Cc'_{\ell}$ is at most
  \begin{align*}
  (((\lambda_2')^Q)^2)^{(1-5\alpha)(1-\alpha)(t'-1)/Q} &\le (\lambda_2^2)^{(1-5\alpha)(1-2\alpha) (1-\alpha)(t'-1)/Q}\\
                                            &\le (\lambda_2^2)^{(1-5\alpha)(1-2\alpha) (1-\alpha)(t-1)} \le \epsilon.
  \end{align*}
  For the rate computation of $\Cc'_{\ell}$, we will lower bound the term $((d'_2)^2)^{-(t'-1)}$. Since $d_2 = (d_2')^{Q}$, $(d_2^2)^{-(t-1)} = \Omega(\eps^{\frac{2+26\cdot \alpha}{1-2\alpha}})$ and $\frac{t'-1}{Q} \le (1+2\alpha)(t-1)$
  (the latter by~\cref{claim:approx_walk_length}), the rate of $\Cc_{\ell}'$ is
  $$
  \Omega(((d'_2)^2)^{-(t'-1)}) = \Omega((d_2^2)^{-(t'-1)/Q}) = \Omega((d_2^2)^{-(1+2\alpha)(t-1)}) = \Omega((\epsilon^{2+ 26 \cdot \alpha})^{\frac{1+2\alpha}{1-2\alpha}}) = \Omega(\epsilon^{2+ 40 \cdot \alpha}).
  $$
\end{proof}

\subsection{Round III: Vanishing $\beta$ as $\epsilon$ Vanishes}\label{sec:round_vanishing_esp}

We are given the dimension of the code $D$ and $\epsilon \in (0,1/2)$.
As before, we set a parameter $\alpha \le 1/128$ such that (for convenience)
$1/\alpha$ is a power of $2$. Set $s \coloneqq 1/\alpha$.

We will consider the regime where $s$ is a function of $\epsilon$. As a
consequence, the parameters $d_2,\lambda_2,d_1,\lambda_1,\epsilon_0$ will
also depend on $\epsilon$. Since $x \le 1/\log_2(1/x)$ for $x \le 1/2$
(and $\alpha \le 1/2$), if $\alpha$ satisfies $\alpha^6/4 \ge
1/\log_2(1/\beta)$, it also satisfies \cref{eq:alpha_epsilon} (we
lose a log factor by replacing $1/\log_2(1/\alpha)$ by $\alpha$, but
we will favor simplicity of parameters). In particular, we can set
$\alpha$ so that $s$ is
$$
s = \Theta((\log_2(1/\epsilon))^{1/6}),
$$
and satisfy~\cref{eq:alpha_epsilon}.

We follow the same choices as in Round II except for the base code
$\Cc_0$.

\noindent \textbf{The base code $\Cc_0$.} \enspace
Set $\epsilon_0 = 1/d_2^2 = \lambda_2^4/b_2^4 \le \lambda_2^4/3$. We
choose an $\epsilon_0$-balanced code $\Cc_0$ with support size
$n=O(D/\epsilon_0^c)$ where $c=2.001$ (this choice of $c$ is arbitrary,
it is enough to have $c$ as a fixed small constant) using the
construction from Round II.  It is crucial that we can unique decode
$\Cc_0$ (using our algorithm), since this is required in order to
apply the list decoding framework.

Note that $\epsilon_0$ is no longer a constant. For this reason, we need
to consider the rate computation of the final code $\Cc_{\ell}$ more
carefully. The proof will follow an argument similar to Ta-Shma's.
       
\begin{claim}
  $\Cc_{\ell}$ has rate $\Omega(\epsilon^{2+ 26 \cdot \alpha})$ where $\alpha = \Theta(1/(\log_2(1/\epsilon))^{1/6})$.
\end{claim}

\begin{proof}
  The support size is the number of walks of length $t-1$ on the
  $s$-wide replacement product of $G$ and $H$ (each step of the walk has $d_2^2$ options),
  which is
  \begin{align*}
    |V(G)| |V(H)| d_2^{2(t-1)} = n' \cdot d_1^s \cdot d_2^{2(t-1)} &= n' \cdot  d_2^{2(t-1)+4s} \le n \cdot  d_2^{2(t-1)+4s}\\
                                                      &= \Theta\left(\frac{D}{\epsilon_0^c} \cdot d_2^{2(t-1)+4s}\right)\\
                                                      &= \Theta\left(D \cdot (d_2^2)^{(t-1)+2s+2.001}\right)\\
                                                      &= O\left(D \cdot (d_2^2)^{(1+2\alpha)(t-1)}\right).                                                      
  \end{align*}
  From this point the proof continues exactly as the proof of~\cref{claim:rate_round_i}.
\end{proof}

\subsection{Round IV: Arbitrary Gentle List Decoding}\label{sec:round_4}

In round III, when we take
$$
s=\Theta((\log_2(1/\epsilon))^{1/6}),
$$
we will have $\lambda_2 = 4s\log(s^{4s^2})/s^{2s^2} \le s^{-s^2}$ provided $s$ is large enough. This non-constant $\lambda_2$ will allow us
perform ``gentle'' list decoding with radius arbitrarily close to
$1/2$. More precisely, we have the following.

\begin{theorem}[Gentle List Decoding (restatement of~\cref{theo:gentle_list_decoding})]
  For every $\epsilon > 0$ sufficiently small, there are explicit binary linear Ta-Shma codes $\Cc_{N,\epsilon,\beta} \subseteq \mathbb{F}_2^N$
  for infinitely many values $N \in \mathbb{N}$ with
  \begin{enumerate}[(i)]
   \item distance at least $1/2 - \epsilon/2$ (actually $\epsilon$-balanced),
   \item rate $\Omega(\epsilon^{2 + \beta})$ where $\beta = O(1/(\log_2(1/\epsilon))^{1/6})$, and
   \item a list decoding algorithm that decodes within radius $1/2 - 2^{-\Theta((\log_2(1/\epsilon))^{1/6})}$
                in time $N^{O_{\epsilon,\beta}(1)}$.
  \end{enumerate}
\end{theorem}

\begin{proof}
  We consider some parameter requirements in order to apply the list
  decoding framework~\cref{theo:restatement:list_dec_hammer} between
  $\Cc_{\ell-1}$ and $\Cc_{\ell}$. Suppose we want to list decode
  within radius $1/2 - \sqrt{\eta}$. For parity sampling, we need
  $$
  s \ge \Theta(\log_2(1/\eta)).
  $$
  Since the number of vertices in a walk can be at most $s^2$,
  for splittability we need
  $$
  \eta^{8}/(s^2 \cdot 2^{2s^2}) \ge 2 \cdot s^{-s^2}.
  $$
  In particular, we can take $\eta = 2^{-\Theta(s)}$ and satisfy both
  conditions above.
\end{proof}

\section{Instantiating the List Decoding Framework}\label{sec:instantiation_list_dec}

We established the tensoriality (actually two-step tensoriality) and
parity sampling properties of every lifting between consecutive codes
$\Cc_{i-1}$ and $\Cc_i$ in Ta-Shma's cascade.
Using these properties, we will be able to invoke the list
decoding framework from~\cite{AJQST19} to obtain the following list
decoding result.

\begin{theorem}[Restatement of~\cref{theo:list_dec_hammer}]\label{theo:restatement:list_dec_hammer}
  Let $\eta_0 \in (0,1/4)$ be a constant, $\eta \in (0,\eta_0)$, and
  	$$k \geq k_0(\eta) \coloneqq \Theta(\log(1/\eta)).$$
  Suppose $\Cc \subseteq \F_2^n$ is an $\eta_0$-balanced linear code and $\Cc' = \dsum_{W(k)}(\Cc)$ is the direct sum lifting of $\Cc$ on a $\tau$-splittable collection of walks $W(k)$, where $W(k)$ is either the set of walks $W[0,s]$ on an $s$-wide replacement product graph or a set of walks using the random walk operator $\sswap{r}$.
  There exists an absolute constant $K > 0$ such that if
  	$$\tau \leq \tau_0(\eta,k) \coloneqq \frac{\eta^{8}}{K \cdot k \cdot 2^{4k}},$$
  then the code $\Cc'$ is $\eta$-balanced and can be efficiently list decoded in the following sense:

  If $\tilde{y}$ is \lict{\sqrt{\eta}} to $\Cc'$, then we
         can compute the list
         $$
         \mathcal{L}(\tilde{y},\Cc,\Cc')\coloneqq \left\{(z,\dsum_{W(k)}(z)) \mid z \in \Cc, \Delta\parens*{\dsum_{W(k)}(z),\tilde{y}} \le \frac{1}{2} - \sqrt{\eta}\right\}
         $$
         in time
         $$
         n^{O(1/\tau_0(\eta,k)^4)} \cdot f(n),
         $$
         where $f(n)$ is the running time of a unique decoding algorithm for $\Cc$.
         Otherwise, we return $\mathcal{L}(\tilde{y},\Cc,\Cc')=\emptyset$ with the same running time
         of the preceding case
~\footnote{In the case $\tilde{y}$ is not \lict{\sqrt{\eta}} to $\Cc'$, but the SOS program turns out to be feasible, some of the calls to the unique decoding algorithm of $\Cc$ (issued by the list decoding framework) might be outside all unique decoding balls. Such cases may be handled by returning failure if the algorithm does not terminate by the time $f(n)$. Even if a codeword in $\Cc$ is found, the pruning step of list decoding~\cite{AJQST19} will return an empty list for $\mathcal{L}(\tilde{y},\Cc,\Cc')$ since $\tilde{y}$ is not \lict{\sqrt{\eta}} to $\Cc$.}.
\end{theorem}

\subsection{List Decoding Framework}

We recall the precise statement of the list decoding
framework tailored to direct sum lifting.

\begin{theorem}[List Decoding Theorem (Adapted from~\cite{AJQST19})]\label{theo:list_dec_framework}
  Suppose $\lift_{W(k)}$ is an $(\eta^8/2^{30},L)$-two-step tensorial direct sum lifting from an $\eta_0$-balanced code $\Cc \subseteq \mathbb{F}_2^n$
  to $\Cc'$ on a multiset $W(k) \subseteq [n]^k$ which is a $(1/2+\eta_0/2, \eta)$-parity sampler.

  Let $\tilde{y} \in \F_2^{W(k)}$ be \lict{\sqrt{\eta}} to $\Cc'$. Then
  the List Decoding algorithm returns the coupled code list
  $\mathcal{L}(\widetilde{y},\Cc,\Cc')$. Furthermore, the
  running time is $n^{O(L+k)}\left(\polylog(1/\eta)
  + f(n)\right)$ where $f(n)$ is the running time of an unique decoding algorithm of
  $\Cc$.
\end{theorem}

We apply the list decoding framework of~\cref{theo:list_dec_framework} to the liftings arising in the Ta-Shma cascade to obtain~\cref{theo:restatement:list_dec_hammer}.
This requires choosing parameters so that both the parity sampling and tensoriality requirements are met at every level of the cascade, which we do by appealing to our results from~\cref{sec:satisfying_framework}.

\begin{proof}[Proof of~\cref{theo:restatement:list_dec_hammer}]
  We want to define parameters for $\tau$-splittability so that $W(k)$ satisfies strong enough parity sampling and tensoriality assumptions to apply~\cref{theo:list_dec_framework}.
  
  For parity sampling, we require $W(k)$ to be an $(1/2+\eta_0/2, \eta)$-parity sampler.
  Suppose $W(k)$ is $\tau$-splittable with $\tau < 1/16$. 
  By \cref{cor:code_cascade_parity_sampling_1} or \cref{cor:code_cascade_parity_sampling_2} and splittability, the collection of walks $W(k)$ is an $(\eta_0', \eta')$-parity sampler, where $\eta' \leq (\eta_0' + 2\tau)^{\lfloor (k-1)/2 \rfloor}$.
  To achieve the desired parity sampling, we take $\eta_0' = 1/2+\eta_0/2$ and choose a value of $k$ large enough so that $\eta' \leq \eta$.
  Using the assumption $\eta_0 < 1/4$, we compute
  	$$\eta' = (\eta_0' +  2\tau)^{\lfloor (k-1)/2 \rfloor} \leq (1/2 + \eta_0/2 + 2\tau)^{k/2-1} < (3/4)^{k/2-1},$$
  which will be smaller than $\eta$ as long as $k$ is at least
  	$$k_0(\eta) = 2 \left(1 + \frac{\log(1/\eta)}{\log(4/3)} \right) = \Theta(\log(1/\eta)).$$
  Achieving this level of parity sampling also ensures that the lifted code $\Cc'$ is $\eta$-balanced.

  The list decoding theorem also requires $(\eta^8/2^{30}, L)$-two-step tensoriality.
  \Cref{lemma:two_step_tensorial} (with $s=k$) and~\cref{lemma:two_step_tensorial_subsequent_steps} each provide $(\mu,L)$-two-step tensoriality for $\tau$-splittable walk collections on the $s$-wide replacement product and using $\sswap{r}$, respectively, with
  	$$L \geq \frac{128k^4 \cdot 2^{4k}}{\mu^4} \text{\quad and \quad} \tau \leq \frac{\mu}{4k \cdot 2^{4k}}.$$
  To get $\mu = \eta^8/2^{30}$, we require
  	$$L \geq \frac{K' \cdot k^4 \cdot 2^{4k}}{\eta^{32}} \text{\quad and \quad} \tau \leq \tau_0(\eta,k) = \frac{\eta^8}{K \cdot k \cdot 2^{4k}},$$
  where $K$ and $K'$ are (very large) constants.
  This ensures that $\tau$ is small enough for the parity sampling requirement as well.
  With these parameters, the running time for the list decoding algorithm in~\cref{theo:list_dec_framework} becomes
  	$$n^{O(L+k)}(\polylog(1/\eta)+f(n)) = n^{O(L)} \cdot f(n) = n^{O(1/\tau_0(\eta,k)^4)} \cdot f(n).$$
\end{proof}

For decoding in fixed polynomial time, we also need a variation of list decoding where we don't run the unique decoding algorithm of the base code and only obtain an approximate list of solutions.
The proof is very similar to the proof of \cref{theo:restatement:list_dec_hammer} above.

\begin{theorem}[Restatement of \cref{theo:list_dec_mallet}]\label{theo:restatement:list_dec_mallet}
  Let $\eta_0 \in (0,1/4)$ be a constant, $\eta \in (0,\eta_0)$, $\zeta = 1/8-\eta_0/8$, and
  	$$k \geq k_0'(\eta) \coloneqq \Theta(\log(1/\eta)).$$
  Suppose $\Cc \subseteq \F_2^n$ is an $\eta_0$-balanced linear code and $\Cc' = \dsum_{W(k)}(\Cc)$ is the direct sum lifting of $\Cc$ on a $\tau$-splittable collection of walks $W(k)$, where $W(k)$ is either the set of walks $W[0,s]$ on an $s$-wide replacement product graph or a set of walks using the random walk operator $\sswap{r}$.
  There exists an absolute constant $K > 0$ such that if
  	$$\tau \leq \tau_0(\eta,k) \coloneqq \frac{\eta^{8}}{K \cdot k \cdot 2^{4k}},$$
  then the code $\Cc'$ is $\eta$-balanced, $W(k)$ is a $(1-2\zeta, \eta)$-parity sampler, and we have the following:

  If $\tilde{y}$ is \lict{\sqrt{\eta}} to $\Cc'$, then we can compute a $\zeta$-cover $\mathcal{L}'$ of the list
         $$
         \mathcal{L}(\tilde{y},\Cc,\Cc')\coloneqq \left\{(z,\dsum_{W(k)}(z)) \mid z \in \Cc, \Delta\parens*{\dsum_{W(k)}(z),\tilde{y}} \le \frac{1}{2} - \sqrt{\eta}\right\}
         $$
         in which $\Delta(y',\tilde y) \leq 1/2 - \sqrt{\eta}$ for every $(z',y') \in \mathcal{L}'$~\footnote{A randomized rounding will ensure this, but see~\cref{app:derandomization} for obtaining this property deterministically.}, in time
         $$
         n^{O(1/\tau_0(\eta,k)^4)}.
         $$
         Otherwise, we return $\mathcal{L}'=\emptyset$ with the same running time of the preceding case.
\end{theorem}

\begin{proof}
  The list decoding framework produces a cover $\mathcal{L}'$ of the list $\mathcal{L}(\tilde y, \Cc, \Cc')$, and, as its final step, corrects the cover to obtain the actual list $\mathcal{L}(\tilde y, \Cc, \Cc')$ by running the unique decoding algorithm of $\Cc$ on each entry of $\mathcal{L}'$ (see~\cite{AJQST19} for details).
  Using \cref{theo:list_dec_framework} with a $(1-2\zeta, \eta)$-parity sampler and omitting this final step of the algorithm, we can obtain the $\zeta$-cover $\mathcal{L}'$ in time $n^{O(L+k)} \polylog(1/\eta)$.
  
  The tensoriality part of the proof of \cref{theo:restatement:list_dec_hammer} applies here unchanged, so we need only make sure $k$ is large enough to yield the stronger parity sampling necessary for this theorem.
  As in that proof, we have that $W(k)$ is an $(\eta_0', \eta')$-parity sampler with $\eta' \leq (\eta_0' + 2\tau)^{\lfloor (k-1)/2 \rfloor}$.
  Take $\eta_0' = 1-2\zeta = 3/4 + \eta_0/4$.
  Using $\eta_0 < 1/4$ and assuming $\tau < 1/16$, we have
  	$$\eta' \leq (\eta_0' + 2\tau)^{\lfloor (k-1)/2 \rfloor} \leq (3/4 + \eta_0/4 + 2\tau)^{k/2-1} < (15/16)^{k/2-1},$$
  which will be smaller than $\eta$ as long as $k$ is at least
  	$$k_0'(\eta) = 2 \left(1 + \frac{\log(1/\eta)}{\log(16/15)} \right) = \Theta(\log(1/\eta)).$$
\end{proof}

\section*{Acknowledgement}

We thank Amnon Ta-Shma for suggesting the problem of decoding in fixed
polynomial running time (with the exponent of $N$ independent of
$\eps$) which led us to think about \cref{theo:main_fixed_poly_time}.
Part of this work was done when some of the authors were visiting the Simons Institute in 
Berkeley, and we thank them for their kind hospitality.

\bibliographystyle{alphaurl}
\bibliography{macros,madhur}

\appendix

\section{Auxiliary Results to Obtain Tensoriality}\label{app:tensoriality}

A key result used in the SOS rounding analysis is embodied
in~\cref{lemma:progress_lemma} below. Roughly speaking, it quantifies
the decrease in the potential $\Phi^G$, under conditioning on a random
$\rv Y_i$ for $i \sim V$, when the ensemble $\set{\rv Y_i}$ has
non-trivial correlation over the edges and $G$ is a strong enough
expander graph. A generalization of this result to low threshold rank
graphs was present in~\cite{BarakRS11}. To derive sharper parameters
in the simpler expander case and to make the presentation
self-contained, we give (essentially) a full proof of this result.

\begin{lemma}[Progress Lemma]\label{lemma:progress_lemma}
  Suppose $G$ satisfy $\lambda_2(G) \le \beta^2/q^4$. If
  $$
  \Ex{i \sim j}{\norm{\{\rv Y_i \rv Y_j\} - \{\rv Y_i\}\{\rv Y_j\} }_1 } \ge \beta,
  $$
  then
  $$
  \Ex{j \sim V}{ \Phi^G_{\vert \rv Y_j} } \le  \Phi^G -  \frac{\beta^2}{4 \cdot q^4}.
  $$
\end{lemma}

\subsection{Expander Case}

We will need the following characterization of the spectral gap of
regular graph $G$. We denote by $\matr A_G$ its adjacency operator and by
$\matr L_G$ its Laplacian operator~\cite{Chung97}.

\begin{fact}[Spectral Gap~\cite{Chung97}]\label{fact:spectral_gap}
  $$
  \lambda_2(\matr L_G) ~=~ \min_{v_1,\dots,v_n \in \mathbb{R}^n} \frac{\E_{i\sim j} \norm{v_i - v_j}^2}{\E_{i,j\sim V} \norm{v_i - v_j}^2}.
  $$
\end{fact}

Using the above characterization, we derive the following
local-to-global result.
\begin{lemma}[Local-to-Global]\label{lemma:local_to_global}
  Let $v_1,\dots,v_n \in \mathbb{R}^n$ be vectors in the unit
  ball. Suppose $\lambda_2(\matr L_G) \ge 1-\beta/2$ (equivalently
  $\lambda_2(\matr A_G) \le \beta/2$). If $\E_{i\sim j} \ip{v_i}{v_j} \ge
  \beta$, then
  $$
  \E_{i,j \sim V} \ip{v_i}{v_j} ~\ge~ \frac{\beta}{2}.
  $$
\end{lemma}

\begin{proof}
  Using~\cref{fact:spectral_gap}, we have
  $$
  \lambda_2(\matr L_G) ~\le~ \frac{\E_{i \sim V} \norm{v_i}^2  - \E_{i \sim j} \ip{v_i}{v_j} }{\E_{i \sim V} \norm{v_i}^2 - \E_{i,j\sim V} \ip{v_i}{v_j}}.
  $$
  Set $\lambda_2 = \lambda_2(\matr L_G)$. We consider two cases: $\lambda_2 \le 1$ and $\lambda_2 > 1$.
  First, suppose $\lambda_2 \le 1$. Then
  \begin{align*}
    \E_{i,j\sim V} \ip{v_i}{v_j} &~\ge~ \frac{1}{\lambda_2} \E_{i \sim j} \ip{v_i}{v_j}  - \left(\frac{1-\lambda_2}{\lambda_2}\right) \E_{i \sim V} \norm{v_i}^2 \\
                              &~\ge~ \frac{1}{\lambda_2}\left( \beta  - \left(1-\lambda_2\right)  \right)\\
                              &~\ge~ \frac{1}{\lambda_2}\left( \beta  - \left( \frac{\beta}{2} \right)  \right) \ge \frac{\beta}{2}.
  \end{align*}
  Now suppose $\lambda_2 > 1$. Then
  \begin{align*}
    \E_{i,j\sim V} \ip{v_i}{v_j} &~\ge~ \frac{1}{\lambda_2} \E_{i \sim j} \ip{v_i}{v_j}  - \left(\frac{1-\lambda_2}{\lambda_2}\right) \E_{i \sim V} \norm{v_i}^2 \\
                              &~\ge~ \frac{1}{\lambda_2}\E_{i \sim j} \ip{v_i}{v_j}
                              \ge \frac{1}{\lambda_2}\cdot \beta  \ge \frac{\beta}{2},
  \end{align*}
  where the last inequality follows from $\lambda_2 \le 2$ for any graph $G$.
\end{proof}

\subsubsection*{More Preliminaries}

We will need some standard notions in information
theory~\cite{CoverT06}.
\begin{definition}[Relative Entropy/Kullback-Leibler Divergence]
  The relative entropy of two distributions $D_1$ and $D_2$ with
  support contained in $\mathcal{Q}$ is
  $$
  \textup{KL}(D_1, D_2) ~\coloneqq~ \sum_{a \in \mathcal{Q}} D_1(a) \log\left(\frac{D_1(a)}{D_2(a)}\right).
  $$
\end{definition}

\begin{notation}
  Let $\rv X$ be a random variable. We denote by $\{\rv X\}$ the distribution of $\rv X$.
\end{notation}

\begin{definition}[Mutual Information]
  Let $\rv X, \rv Y$ be two random variables. The mutual information $\textup{I}(\rv X, \rv Y)$ is
  $$
  \textup{I}(\rv X, \rv Y) ~\coloneqq~ \textup{KL}(\{\rv X,\rv Y\}, \{\rv X\}\{\rv Y\}).
  $$
\end{definition}

\begin{fact}
 $$
 \textup{I}(\rv X,\rv Y) ~=~ \textup{H}(\rv X) - \textup{H}(\rv X\vert \rv Y).
 $$
\end{fact}

\begin{fact}[Fact B.5 of Raghavendra and Tan~\cite{RaghavendraT12}]\label{fact:rt_ineq}
  Let $\rv X_a$ and $\rv X_b$ be indicator random variables. Then
  $$
  \textup{Cov}(\rv X_a, \rv X_b)^2 ~\le~ 2 \cdot \textup{I}(\rv X_a,\rv X_b).
  $$
\end{fact}

\subsubsection*{Progress Lemma}

We are ready to prove~\cref{lemma:progress_lemma} which we restate
below for convenience.
\begin{lemma}[Progress Lemma (restatement of~\cref{lemma:progress_lemma})]
  Suppose $G$ satisfy $\lambda_2(G) \le \beta^2/q^4$. If
  $$
  \Ex{i \sim j}{\norm{\{\rv Y_i \rv Y_j\} - \{\rv Y_i\}\{\rv Y_j\} }_1 } \ge \beta,
  $$
  then
  $$
  \Ex{j \sim V}{ \Phi^G_{\vert \rv Y_j} } \le  \Phi^G -  \frac{\beta^2}{4 \cdot q^4}.
  $$
\end{lemma}

\begin{proof}
  Firstly, we show how to relate the distances $\norm{\{\rv Y_i \rv Y_j\} - \{\rv Y_i\}\{\rv Y_j\} }_1$ over the edges $i \sim
  j$ to certain covariances. Let $a,b \in [q]^2$. Observe that
  $$
  \left\lvert \textup{Cov}\left(\rv Y_{i,a},\rv Y_{j,b}\right) \right\rvert ~=~ \left\lvert \Pr[\rv Y_i = a \wedge \rv Y_j = b] - \Pr[\rv Y_i = a]\Pr[\rv Y_j = b] \right\rvert.
  $$
  We have
  \begin{align*}
    \Ex{i \sim j}{\frac{1}{q^2} \sum_{a,b \in [q]^2} \textup{Cov}\left(\rv Y_{i,a},\rv Y_{j,b}\right)^2} ~\ge~ &
    \left(\Ex{i \sim j}{\frac{1}{q^2} \sum_{a,b \in [q]^2} \left\vert\textup{Cov}\left(\rv Y_{i,a},\rv Y_{j,b}\right) \right\vert}\right)^2 \\
    ~\ge~ & \frac{1}{q^4} \left(\Ex{i \sim j}{\norm{\{\rv Y_i \rv Y_j\} - \{\rv Y_i\}\{\rv Y_j\} }_1}\right)^2 \ge \frac{\beta^2}{q^4}.\\
  \end{align*}
  Note that the graph $\mathcal{F} \coloneqq G \otimes J/q$ is an expander with $\lambda_2(G \otimes J/q)=\lambda_2(G)$.
  Moreover, the matrix $\matr C \coloneqq \set{\textup{Cov}\left(\rv Y_{i,a},\rv Y_{j,b}\right)}_{i,j \in V;a,b \in [q]^2}$ is
  PSD since the vectorization $\set{v_{i,a} - \E[\rv Y_{i,a}] \cdot v_{\emptyset}}_{i \in V; a \in [q]}$ gives a
  Gram matrix decomposition of $\matr C$. Thus, the covariance matrix $\set{\textup{Cov}\left(\rv Y_{i,a},\rv Y_{j,b}\right)^2}_{i,j \in V;a,b \in [q]^2}$
  is also PSD since it is the Schur product (i.e., entrywise product) of two PSD matrices, namely, $\matr C \circ \matr C$. Therefore,
  we are in position of applying the local-to-global~\cref{lemma:local_to_global} with the expander $\mathcal{F}$ and a vectorization for
  $\matr C \circ \matr C$. We have
  \begin{align*}
    \frac{\beta^2}{q^4} & ~\le~ \Ex{i \sim j}{\frac{1}{q^2} \sum_{a,b \in [q]^2} \textup{Cov}\left(\rv Y_{i,a},\rv Y_{j,b}\right)^2}\\
                           & ~\le~ 2 \Ex{i,j \sim V^{\otimes 2}}{\frac{1}{q^2} \sum_{a,b \in [q]^2} \textup{Cov}\left(\rv Y_{i,a},\rv Y_{j,b}\right)^2} && \text{(local-to-global~\cref{lemma:local_to_global})}\\
                           & ~\le~ \frac{4}{q^2} \Ex{i,j \sim V^{\otimes 2}}{\sum_{a,b \in [q]^2} \textup{I}\left(\rv Y_{i,a},\rv Y_{j,b}\right)} && \text{(\cref{fact:rt_ineq})}\\
                           & ~\le~ \frac{4}{q^2} \Ex{i,j \sim V^{\otimes 2}}{\sum_{a,b \in [q]^2} \textup{H}\left(\rv Y_{i,a}\right) - \textup{H}\left(\rv Y_{i,a}\vert \rv Y_{j,b}\right)}\\
                           & ~\le~ \frac{4}{q} \left[\Ex{i \sim V}{\sum_{a \in [q]} \textup{H}\left(\rv Y_{i,a}\right)} - \Ex{i,j \sim V^{\otimes 2}}{\sum_{a \in [q]}  \textup{H}\left(\rv Y_{i,a}\vert \rv Y_j\right)} \right]\\
                           & ~=~ 4 \left[\Ex{i \sim V}{\mathcal{H}\left(\rv Y_i\right)} - \Ex{i,j \sim V^{\otimes 2}}{ \mathcal{H}\left(\rv Y_i\vert \rv Y_j\right)} \right]\\
                           & ~=~ 4 \left[ \Phi^G - \Ex{j \sim V}{ \Phi^G_{\vert \rv Y_j} } \right].
  \end{align*}
  Therefore, we have $\Ex{j \sim V}{\Phi^G_{\vert \rv Y_j}} \le  \Phi^G -  \beta^2/(4 \cdot q^4)$,
  as claimed.
\end{proof}

\section{Explicit Structures}\label{app:explicit_structures}

We recall some explicit structures used in Ta-Shma's
construction.

\subsection{Explicit Ramanujan Graphs}

The outer graph $G$ in the $s$-wide replacement product was
chosen to be a Ramanujan graph. Ta-Shma provides a convenient lemma to
efficiently obtain explicit Ramanujan graphs given the intended number
of vertices $n$ (which might end up being nearly twice this much), the
expansion $\lambda$ and an error parameter $\theta > 0$. These
Ramanujan graphs are based on the LPS construction~\cite{LPS88}. Due
to number theoretic reasons, we might be forced to work with slightly
different parameters, but this is not an issue.

\begin{lemma}[Lemma 2.10~\cite{Ta-Shma17}]\label{lemma:explicit_ramanujan}
 For every $\theta > 0$, there exists an algorithm that given $n$ and
 $\lambda \in (0,1)$ runs in time $\poly(n)$ and outputs a Ramanujan
 graph $G$ such that
 \begin{itemize}
   \item $G$ has degree $d \le 8/\lambda$,
   \item $\sigma_2(G) \le \lambda$, and
   \item $\abs{V(G)}$ is either in the range $[(1-\theta)n,n]$ or in the range
         $[(1-\theta)2n,2n]$.
 \end{itemize}
 Moreover, the algorithm outputs a locally invertible function $\varphi \colon [d] \to [d]$
 computable in polynomial time in its input length.
\end{lemma}

\subsection{Explicit Biased Distribution}

The inner graph $H$ in the $s$-wide replacement product is chosen to
be a Cayley graph on $\mathbb{Z}_2^m$ for some positive integer $m$.
Ta-Shma uses the construction of Alon et al.~\cite{AGHP92} (AGHP) to
deduce a result similar to~\cref{lemma:aghp} below. To compute the
refined parameter version of our main result~\cref{theo:main}, we will
need the specifics of the AGHP construction.

\begin{lemma}[Based on Lemma 6~\cite{Ta-Shma17}]\label{lemma:aghp}
  For every $\beta=\beta(m)$, there exists a fully explicit set $A \subseteq \mathbb{Z}_2^m$ such that
  \begin{itemize}
    \item $\abs{A} \le 4 \cdot m^2/\beta^2$, and
    \item for every $S \subseteq [m]$, we have $\abs{\E_{z \in A} \chi_S(z)} \le \beta$.
  \end{itemize}
  Furthermore, if $m/\beta$ is a power of $2$, then $\abs{A} = m^2/\beta^2$.
  In particular, the graph $\textup{Cay}(\mathbb{Z}_2^m,A)$ is a $(n=2^m,d=\abs{A},\lambda=\beta)$
  expander graph.
\end{lemma}

\begin{remark}
  Given $d,m \in \mathbb{N}^+$ such that $d$ is the square of a power
  of $2$ with $d \le 2^m$, by setting $\beta = m/\sqrt{d}$ we can
  use~\cref{lemma:aghp} with $\beta$ and $m$ (note that
  $m/\beta$ is a power of $2$) to obtain a Cayley graph
  $\textup{Cay}(\mathbb{Z}_2^m,A)$ with parameters
  $(n=2^m,d=\abs{A},\lambda=\beta)$.
\end{remark}

\section{Zig-Zag Spectral Bound}

We prove the zig-zag spectral bound \autoref{fact:zig_zag_bound}.

\begin{claim}
  Let $G$ be an outer graph and $H$ be an inner graph used in the $s$-wide replacement product.
  For any integer $0 \leq i \leq s-1$,
  $$
  \sigma_2((I \otimes \Aye_H) G_i (I \otimes \Aye_H)) \leq \sigma_2(G) + 2 \cdot \sigma_2(H) + \sigma_2(H)^2.
  $$
\end{claim}

\begin{proof}
Let $v$ be a unit vector such that $v \bot \one$, and decompose it into $v = u+w$ such that $u \in \mathcal{W}^{\parallel} = \textup{span}\set{ a \otimes b \in \mathbb{R}^{V(G)} \otimes \mathbb{R}^{V(H)} \mid b = \one }$ and $w\in \mathcal{W}^{\bot} = (\mathcal{W}^{\parallel})^{\bot}$.
\begin{align*}
	\left| \ip{v}{(I \otimes \Aye_H) G_i (I \otimes \Aye_H)v} \right| \leq & \left| \ip{u}{(I \otimes \Aye_H) G_i (I \otimes \Aye_H)u} \right| + \left| \ip{u}{(I \otimes \Aye_H) G_i (I \otimes \Aye_H)w} \right| + \\
	 & \left| \ip{w}{(I \otimes \Aye_H)G_i (I \otimes \Aye_H) u} \right| + \left| \ip{w}{(I \otimes \Aye_H) G_i (I \otimes \Aye_H)w} \right| \\
	\leq & \left| \ip{u}{ G_i u} \right| + |(I \otimes \Aye_H)w| +\\
	 & |(I \otimes \Aye_H)w| + |(I \otimes \Aye_H)w|^2 \\
	\leq & \left| \ip{u}{ G_i u} \right| + 2 \sigma_2(H) + \sigma_2^2(H)
\end{align*}
To bound $\left| \ip{u}{(I \otimes \Aye_H) G_i (I \otimes \Aye_H)u} \right|$, observe that $u = x \otimes \one$ for some $x \in \mathbb{R}^{V(G)}$. Then,
$$
0 = \ip{v}{\one} = \ip{u}{\one}+\ip{w}{\one} = \ip{u}{\one}  =\ip{x}{\one _G}
$$
so that $x\bot \one_G$. Because $u$ is uniform over $H$-component, $\left| \ip{u}{G_i u} \right| = \left| \ip{x}{G x} \right| \leq \sigma_2(G)$, which completes the proof.
\end{proof}

We also derive a (simple) tighter bound for the expansion of the zig-zag
product in a particular parameter regime.
\begin{claim}\label{claim:zig_zag_refined}
  Let $G$ be a $\lambda_1$-two-sided expander and $H$ be a
  $\lambda_2$-two-sided expander such that both are regular graphs. If
  $\lambda_1 \le \lambda_2$, then
  $$
  \sigma_2(G \zigzag H) \le 2 \cdot \lambda_2.
  $$
\end{claim}

\begin{proof}
  Let $v = a \cdot v^{\parallel} + b \cdot v^{\perp}$ with $a^2 + b^2
  = 1$ be such that $v \perp 1$. In particular, if $v^{\parallel} =
  v^G \otimes 1^H$, then $v^G \perp 1^G$ since otherwise $\ip{v}{1}
  = \ip{v^G}{1^G} \ne 0$. We have
  \begin{align*}
    \max_{a,b \in \mathbb{R} \colon a^2 + b^2 = 1} a^2 \cdot \lambda_1 + 2 ab \cdot \lambda_2 + b^2 \cdot \lambda_2^2
         & \le \max_{a,b \in \mathbb{R} \colon a^2 + b^2 = 1} a^2 \cdot \lambda_2 + 2 ab \cdot \lambda_2 + b^2 \cdot \lambda_2 \\
         &= \max_{a,b \in \mathbb{R} \colon a^2 + b^2 = 1} \lambda_2 + 2 ab \cdot \lambda_2,
  \end{align*}
  where the inequality follows from the assumption $\lambda_1 \le \lambda_2$ (and trivially $\lambda_2^2 \le \lambda_2$)
  and the equality follows from $a^2 + b^2 = 1$. Since we also have $2ab = (a+b)^2 - (a^2+b^2) \le 1$,
  the result follows.
\end{proof}

\section{Derandomization}\label{app:derandomization}

To unique decode in fixed polynomial time (i.e., $\poly(n/\epsilon)$) we
need to prune the list of coupled words $\mathcal{L}$ covering the
list $\mathcal{L}^*(\tilde{y}) = \set{(z,y=\lift(z)) \mid
z \in \Cc, \Delta(\tilde{y},y) \le 1/2 - \sqrt{n}}$ of codewords we
want to retrieve. To do so, given
$(z^*,y^*=\lift(z^*)) \in \mathcal{L}^*(\tilde{y})$, we need to have
$(z,y=\lift(z)) \in \mathcal{L}$ such that
\begin{enumerate}
 \item $\abs{\ip{y^*}{\lift(z)}}$ is not too small, and \label{enum:derand_item_i}
 \item $\ip{\tilde{y}}{\lift(z)}$ is not too small (in order to apply~\cref{lemma:cover_compactness}). \label{enum:derand_item_ii}
\end{enumerate}
The slice $(S,\sigma)$ of the SOS solution from which $y^*$ is
recoverable satisfies in expectation
$$
\E_{z \sim \set{\rv Z^{\otimes}\vert_{(S,\sigma)}}}\left[\ip{y^*}{\lift(z)}^2\right] \ge 3 \eta^2,
$$
and
$$
\E_{z \sim \set{\rv Z^{\otimes}\vert_{(S,\sigma)}}}\left[\ip{\tilde{y}}{\lift(z)}\right] \ge 3\sqrt{\eta}/2.
$$
Moreover, since $z \mapsto \ip{y^*}{\lift(z)}^2$ and
$z \mapsto \ip{\tilde{y}}{\lift(z)}$ are
$O(1/n)$-Lipschitz~\footnote{In this fixed polynomial time regime, the
parameters $s,d_1,d_2,\epsilon_0,\eta$ are constant independent of the
final bias $\epsilon$.} with respect to the $\ell_1$-norm, Hoeffding's
inequality gives
$$
\ProbOp_{z \sim \set{\rv Z^{\otimes}\vert_{(S,\sigma)}}}\left[\ip{y^*}{\lift(z)}^2 < \eta^2 \right] \le  \exp\left(-\Theta(n)\right),
$$
and
$$
\ProbOp_{z \sim \set{\rv Z^{\otimes}\vert_{(S,\sigma)}}}\left[\ip{\tilde{y}}{\lift(z)} < \sqrt{\eta}\right] \le \exp\left(-\Theta(n)\right).
$$
At least randomly, such a $z$ can be easily found. In~\cite{AJQST19},
alternatively to satisfying~\cref{enum:derand_item_i} it was shown
that by choosing $z' \in \set{\pm 1}^n$ by majority vote, i.e.
$$
z'_i = \argmax_{b \in \set{\pm 1}} \ProbOp[\rv Z_i = b]
$$
for $i \in [n]$, one has that $\abs{\ip{z^*}{z'}}$ is large which is
enough to address the first item. More precisely implicit
in~\cite{AJQST19}, for any constant $\beta \in (0,1)$ as long as
parity sampling is sufficiently strong we have
$$
\E_{z \sim \set{\rv Z^{\otimes}\vert_{(S,\sigma)}}}\left[\ip{z'}{z}^2 \right] \ge 1 - \beta.
$$
Similarly $z \mapsto \ip{z'}{z}^2$ is $O(1/n)$-Lipschitz with respect to the
$\ell_1$-norm, so Hoeffding's inequality yields
$$
\ProbOp_{z \sim \set{\rv Z^{\otimes}\vert_{(S,\sigma)}}}\left[\ip{z'}{z}^2 < 1 - \beta/2 \right] \le \exp\left(-\Theta(n)\right).
$$
However, we want to efficiently and deterministically find a $z$
satisfying $\ip{z'}{z}^2 \ge 1 -\beta/2$ as well as
satisfying~\cref{enum:derand_item_ii}. Note that at this stage in the
decoding process $y^*$ is not known (without issuing a recursive
unique decoding call), so running expectation maximization to satisfy
item~\cref{enum:derand_item_i} would not be possible. Fortunately, the
majority $z'$ can be cheaply computed without a recursive call to an
unique decoder. On the other hand $z$ satisfying
only~\cref{enum:derand_item_ii} can be found by expectation
maximization. The challenge is to satisfy both conditions at the same
time. For this reason, we design a simultaneous expectation
maximization derandomization procedure tailored to our setting.

\subsection{Abstract Derandomization: Simultaneous Expectation Maximization}

Suppose that $\Omega$ is a probability space where two random
variables $\rv A$ and $\rv B$ are defined satisfying the following
first moment conditions
$$
\E\left[ \rv A\right] \ge a \quad \text{ and } \quad \E\left[\rv B \right] \ge 1-\beta.
$$
We provided a sufficient conditions so that $\omega \in \Omega$ satisfying
$$
\rv A(\omega) \ge a' \quad \text{ and } \quad \rv B(\omega) \ge 1-\beta'
$$
can be efficiently deterministically found with the aid of an oracle,
where $a \approx a'$ and $\beta \approx \beta'$. More precisely,
we have the following lemma.

\begin{lemma}
  Let $\Omega = (\pmone^n, \nu_1 \times \cdots \times \nu_n)$ be a probability space with a product distribution.
  Suppose $\rv A \in [-1,1]$ is a random variable on $\Omega$ satisfying, for $a > 0$ and for some function $e_A \colon \mathbb{N} \to \mathbb{R}^+$,
  $$
  \ProbOp\left[ \rv A < a  \right] \le e_A(n).
  $$
  Suppose $\rv B \in [-1,1]$ is a random variable on $\Omega$ satisfying,
  for some function $e_B \colon \mathbb{N}\times \mathbb{R}^+ \to \mathbb{R}^+$,
  $$
  \ProbOp\left[ \rv B < 1 - \beta \right] \le e_{B}(n,\beta).
  $$
  Suppose that there is an oracle to evaluate $\E\left[ \rv A \rv B^{2k} \right]$ under any product distribution
  $\mu_1' \times \cdots \times \mu_n'$ for $k \in \mathbb{N}$. Given $\delta,\beta \in (0,1)$,
  if
  \begin{equation}\label{eq:drand_err_bound}
    e_A(n) + e_{B}(n, \beta/(4\lceil -\ln(a(1-\beta) + 1)/\delta \rceil)) \le a \frac{\beta}{2},
  \end{equation}
  then it is possible to find $\omega \in \set{\pm 1}^n$ using $2n$ invocations to the oracle
  and satisfying
  $$
  \rv A(\omega) \ge a (1-\beta)  \quad \text{ and } \quad \abs{\rv B(\omega)} \ge 1 - \delta.
  $$
\end{lemma}

\begin{proof}
 Set $k = \lceil -\ln(a(1-\beta) + 1)/\delta \rceil$.
 Set $\beta' = \beta/(4k)$. Note that
 $$
 \E\left[ \rv A \rv B^{2k} \right] \ge a \left(1-\frac{\beta}{4k}\right)^{2k} - e_A(n) - e_{B}(n,\beta') \ge a \left(1- \beta \right),
 $$
 where we use~\cref{eq:drand_err_bound} in the last inequality.
 Do expectation maximization to deterministically find $\omega \in \set{\pm 1}^n$, with $2 \cdot n$ invocations to the oracle
 of $\E\left[ \rv A \rv B^{2k} \right]$, such that
 $$
 \rv A(\omega) \rv B(\omega)^{2k} \ge a \left(1 - \beta \right).
 $$
 Since $\rv B(\omega)^{2k} \le 1$, we have $\rv A(\omega) \ge a \left(1 - \beta \right)$.
 Towards a contradiction suppose $\abs{\rv B(\omega)} \le 1 - \delta$. Using that $\rv A(\omega) \le 1$, we
 have
 \begin{equation}\label{eq:derand_lb}
   e^{-2k \cdot \delta} \ge (1-\delta)^{2k} \ge \rv A(\omega) \rv B(\omega)^{2k} \ge a (1-\beta).
 \end{equation}
 By our choice of $k$, we get
 $$
 e^{-2k \cdot \delta} < a (1-\beta),
 $$
 contradicting~\cref{eq:derand_lb}.
\end{proof}

\subsection{Implementing the Oracle}

Now, we provide an efficient deterministic oracle for our setting. We
take
$$
\rv A \coloneqq \ip{\tilde{y}}{\lift(z)} \quad \text{ and } \quad \rv B \coloneqq \ip{z'}{z}^2,
$$
where $z'_i = \argmax_{b \in \set{\pm 1}} \ProbOp[\rv Z_i = b]$. Note that
$$
\rv A \rv B^{2k} = \sum_{T \subset [n] \colon \abs{T}=O(1)} \alpha_T \prod_{i \in T} z_i.
$$
To compute $\E\left[\rv A \rv B^{2k}\right]$ under any product
distribution $\mu_1' \times \cdots \times \mu_n'$, use linearity of
expectation and sum at most $n^{O(1)}$ terms
$\alpha_T \E\left[\prod_{i \in T} z_i \right]$ where each can be
computed in $O(1)$ since restricted to $T$ we have a product
distribution taking values in $\set{\pm 1}^T$.

\end{document}